\documentclass[reprint,superscriptaddress,amsmath,amssymb,aps,pra]{revtex4-2}

\usepackage{graphicx}
\usepackage{bm}
\usepackage[colorlinks=true,
  linkcolor=black,
  citecolor=blue,
  urlcolor =blue]{hyperref}
\usepackage[capitalize]{cleveref}
\usepackage{amsmath}
\usepackage[normalem]{ulem}
\usepackage{amssymb}
\usepackage{amsthm}
\usepackage{graphicx}
\usepackage{algorithmic}
\usepackage{mathrsfs}
\usepackage{braket}
\usepackage[dvipsnames]{xcolor}
\usepackage{soul,xcolor}
\usepackage{subcaption}
\newtheorem{thm}{\protect\theoremname}
\theoremstyle{plain}
\newtheorem{lem}[thm]{\protect\lemmaname}
\theoremstyle{plain}
\newtheorem{rem}[thm]{\protect\remarkname}
\theoremstyle{plain}
\newtheorem*{lem*}{\protect\lemmaname}
\theoremstyle{plain}
\newtheorem{prop}[thm]{\protect\propositionname}
\theoremstyle{plain}

\newtheorem{assumption}[thm]{Assumption}
\newtheorem{defn}[thm]{Definition}

\usepackage{comment}

\usepackage[USenglish]{babel}
\providecommand{\corollaryname}{Corollary}
\providecommand{\lemmaname}{Lemma}
\providecommand{\propositionname}{Proposition}
\providecommand{\remarkname}{Remark}
\providecommand{\theoremname}{Theorem}

\usepackage{complexity}

\usepackage{orcidlink}

\newcommand{\mc}[1]{\mathcal{#1}}

\newcommand{\eps}{\varepsilon}
\newcommand{\abs}[1]{\left\lvert#1\right\rvert}
\newcommand{\norm}[1]{\left\lVert#1\right\rVert}
\newcommand{\tnorm}[1]{{\left\vert\kern-0.25ex\left\vert\kern-0.25ex\left\vert#1\right\vert\kern-0.25ex\right\vert\kern-0.25ex\right\vert}}

\newcommand{\ud}{\,\mathrm{d}}
\newcommand{\Or}{\mathcal{O}}
\newcommand{\RR}{\mathbb{R}}
\newcommand{\CC}{\mathbb{C}}
\newcommand{\wt}{\widetilde}
\newcommand{\Tr}{\mathrm{Tr}}

\renewcommand{\ket}[1]{\ensuremath{\left|#1\right\rangle}}
\renewcommand{\bra}[1]{\ensuremath{\left\langle#1\right|}}

\newcommand{\diag}{\operatorname{diag}}
\newcommand{\bvec}[1]{\mathbf{#1}}

\newcommand{\vb}{\bvec{b}}

\newcommand{\rev}[1]{{#1}}
\newcommand{\revv}[1]{{#1}}
\newcommand{\revvv}[1]{{#1}}

\newcommand{\ketbra}[2]{|#1\rangle\!\langle #2 |}

\begin{document}

\setstcolor{red}
\title{Rapid quantum ground state preparation via dissipative dynamics}
\author{Yongtao Zhan}
\thanks{These authors contributed equally to this work.}
\affiliation{Institute for Quantum Information and Matter, California Institute of Technology}
\affiliation{Division of Physics, Mathematics, and Astronomy, California Institute of Technology}
\author{Zhiyan Ding}
\thanks{These authors contributed equally to this work.}
\affiliation{Department of Mathematics, University of California, Berkeley}
\author{Jakob Huhn}
\affiliation{Department of Physics and Arnold Sommerfeld Center for Theoretical Physics, Ludwig-Maximilians-Universit{\"a}t M{\"u}nchen}
\affiliation{Department of Mathematics, University of California, Berkeley}
\author{Johnnie Gray}
\affiliation{Division of Chemistry and Chemical Engineering, California Institute of Technology}
\author{John Preskill}
\affiliation{Institute for Quantum Information and Matter, California Institute of Technology}
\affiliation{Division of Physics, Mathematics, and Astronomy, California Institute of Technology}
\affiliation{AWS Center for Quantum Computing}
\author{Garnet Kin-Lic Chan}
\affiliation{Institute for Quantum Information and Matter, California Institute of Technology}
\affiliation{Division of Chemistry and Chemical Engineering, California Institute of Technology}
\author{Lin Lin\orcidlink{0000-0001-6860-9566}}
\thanks{linlin@math.berkeley.edu}
\affiliation{Department of Mathematics, University of California, Berkeley}
\affiliation{Applied Mathematics and Computational Research Division, Lawrence Berkeley National Laboratory}

\begin{abstract}
Inspired by natural cooling processes, dissipation has become a promising approach for preparing low-energy states of quantum systems. However, the potential of dissipative protocols remains unclear beyond certain commuting Hamiltonians. This work provides significant analytical and numerical insights into the power of dissipation for preparing the ground state of noncommuting Hamiltonians.  For quasi-free dissipative dynamics, including certain 1D spin systems with boundary dissipation, our results reveal a new connection between the mixing time in trace distance and the spectral properties of a non-Hermitian Hamiltonian, leading to an explicit and sharp bound on the mixing time that scales polynomially with system size.  For more general spin systems, we develop a tensor network-based algorithm for constructing the Lindblad jump operator and for simulating the dynamics. Using this algorithm, we demonstrate numerically that dissipative ground state preparation protocols can achieve rapid mixing for certain 1D local Hamiltonians under bulk dissipation, with a mixing time that scales logarithmically with the system size. We then prove the rapid mixing result for certain weakly interacting spin and fermionic systems in arbitrary dimensions, extending recent results for high-temperature quantum Gibbs samplers to the zero-temperature regime.  Together, these results show that dissipation can be a powerful tool for ground state preparation, with potential applications across condensed matter physics, quantum materials science, and beyond.

\end{abstract}

\maketitle


\section{Introduction}\label{sec:intro}

Ground state preparation is one of the most important challenges in quantum many-body physics, quantum chemistry, and materials science. Quantum algorithms, such as quantum phase estimation (QPE), quantum singular value transformation (QSVT), adiabatic state preparation (ASP) and their variants \cite{FarhiGoldstoneGutmannEtAl2000,AlbashLidar2018,OBrienTarasinskiTerhal2019,GilyenSuLowEtAl2019,GeTuraCirac2019,LinTong2020a,LinTong2022,WanBertaCampbell2022},  offer a pathway to tackle challenging ground state preparation problems beyond the capabilities of classical computers. Dissipative dynamics,  such as Lindblad dynamics, provides a distinct approach to state preparation. This approach evolves the system density matrix under engineered dissipation and Hamiltonian dynamics and encodes the target state as the stationary-state solution of the Lindblad master equation.

Dissipative techniques, and state preparation methods employing midcircuit measurements in general, have been widely applied to prepare matrix product states, ground states of stabilizer codes, spin systems, and other states exhibiting long-range entanglement \cite{TerhalDiVincenzo2000,KrausBuchlerDiehlEtAl2008,VerstraeteWolfIgnacioCirac2009,RoyChalkerGornyiEtAl2020,zhou2021symmetry,Cubitt2023,WangSnizhkoRomitoEtAl2023,LuLessaKimEtAl2022,Foss-FeigTikkuLuEtAl2023,KalinowskiMaskaraLukin2023,langbehn2024dilute}. Compared to traditional unitary quantum algorithms as well as adiabatic algorithms, dissipative approaches offer certain inherent robustness to noise, and may bypass the need for complex initialization procedures, making them attractive for implementation on early fault-tolerant quantum devices~\cite{CubittLuciaMichalakisEtAl2015,TrivediFrancoRubioCirac2024,PhysRevA.95.042302}. However, many existing dissipative protocols are tailored for highly structured and frustration-free Hamiltonians. For instance, the ground state of the parent Hamiltonian of a stabilizer code can be efficiently prepared using either quantum error correction protocols or dissipative dynamics \cite{VerstraeteWolfIgnacioCirac2009}, but most physical Hamiltonians \rev{(i.e., Hamiltonians that are actually relevant for scientific applications)} lack these favorable structures.

Encouragingly, recent years have seen significant advances in developing new dissipative protocols for Gibbs state preparation~\cite{Temme_2011,MozgunovLidar2020,ChenBrandao2021,shtanko2021preparing,RallWangWocjan2023,ChenKastoryanoBrandaoEtAl2023,ChenKastoryanoGilyen2023,DingLiLin_KMS,gilyen2024quantum,JiangIrani2024}, \rev{as well as in understanding their effectiveness by analyzing the mixing time~\cite{TemmeKastoryanoRuskaiEtAl2010,KastoryanoTemme2013,BardetCapelGaoEtAl2023,rouz2024,DingLiLinZhang2024,kochanowski2024rapid,rouze2024optimal,gamarnik2024slow,tong2024fast}, which quantifies the time required to drive any initial state to the steady state of the dissipative dynamics (see definition in~\cref{sec:Lindalgo}).} Several dissipative protocols~\cite{Cubitt2023,MiMichailidisShabaniEtAl2024,ChenHuangPreskillEtAl2024,DingChenLin2024,li2024dissipative,LambertCirioLinEtAl2024,eder2024quantum,motlagh2024ground,LloydMichailidisMiEtAl2025} have also been designed to prepare the ground state of noncommuting Hamiltonians. Such protocols will still encounter the Quantum Merlin-Arthur hardness of ground state preparation in the worst-case scenario, where the challenge can manifest as exponentially long mixing times. Nonetheless, these methods more closely resemble  cooling processes in nature, and offer the potential for significantly shorter mixing times in certain physical Hamiltonians.

The theoretical characterization of efficient ground state preparation protocols is however much more challenging than that for thermal states. A key distinction lies in the invertibility of thermal states, which is essential for the concept of quantum detailed balance conditions (DBC)~\cite{Alicki1976,FagnolaUmanita2010,CarlenMaas2017,ChenKastoryanoBrandaoEtAl2023,DingLiLin_KMS}. In contrast, \rev{the density matrix of a pure ground state has rank one and is therefore non-invertible, which makes most existing theoretical tools inapplicable in this setting}. Numerically, these protocols can also be difficult to simulate for systems beyond the reach of exact diagonalization methods, as constructing the corresponding Lindblad jump operators is significantly more complex than that in typical Lindblad dynamics.

In this work, we make significant progress in understanding the capabilities of dissipative ground state protocols through both analytical and numerical investigations.~\rev{A concise overview of the numerical and theoretical results is provided in~\cref{sec:summary}.}

\rev{The rest of the paper is organized as follows. In \cref{sec:Lindalgo}, we review the Lindblad-based ground-state preparation algorithm, introducing the notion of mixing time and discussing considerations for estimating the resource requirements of dissipative protocols. Before the full discussion, \cref{sec:summary} provides an overview of the main results. We develop a tensor-network method for simulating general Lindblad dynamics on classical computers in \cref{sec:TNmethod}.  In \cref{sec:numer_quasifree}, we present the performance of the ground-state preparation protocol for a variety of systems governed by quasi-free Lindblad dynamics. In \cref{sec:numer_TNmethod}, we report the numerical performance of the tensor-network method for simulating the ground state preparation process beyond quasi-free systems. \cref{sec:compare_adiabatic} provides a concrete example comparing dissipative protocols with adiabatic state preparation methods for preparing ground states. On the theoretical side, in \cref{sec:mixingtime_quasifree} we estimate the convergence rate in trace distance for quasi-free systems, confirming the numerical results of \cref{sec:numer_quasifree}. Our rigorous analysis of rapid mixing for ground-state preparation is presented in \cref{sec:provable_rapid}. Background material, detailed proofs, and additional numerical results are collected in the Appendices.}

\section{Lindblad-based ground state preparation algorithm}\label{sec:Lindalgo}

\rev{The main goal of this work is to examine the performance of the Lindblad-based ground state preparation algorithm introduced in Ref.~\cite{DingChenLin2024}. This dissipative algorithm, inspired by gradient descent dynamics in classical systems, employs carefully designed jump operators to iteratively reduce the system's energy and can prepare ground states for noncommuting Hamiltonians.}

\paragraph*{\rev{Algorithmic construction---}}

The Lindblad master equation for ground state preparation proposed in\rev{~\cite[Eq. (1)]{DingChenLin2024}} takes the form
\begin{equation}\label{eq:lindblad}
\frac{\mathrm{d} \rho}{\mathrm{d}t}=\mathcal{L}[\rho]=-i[H, \rho]+\sum_a K_a \rho K_a^{\dagger}-\frac{1}{2}\left\{K_a^{\dagger} K_a, \rho\right\}.
\end{equation}
We refer to $K_a$ as a jump operator, $-i[H, \rho]$ as the coherent part of the dynamics,  and $\sum_a K_a \rho K_a^{\dagger}-\frac{1}{2}\left\{K_a^{\dagger} K_a, \rho\right\}$ as the dissipative part of the dynamics, respectively.

Starting from a set of  coupling operators $\{A_a\}$, whose selection will be discussed in detail later, the corresponding jump operator $K_a$ is engineered to ``shovel'' high energy components in the density matrix towards lower energy ones  (\cref{fig:main_illustrate}a). \rev{The success of the ground state preparation algorithm relies on the assumption that, starting from a simple initial state \rev{(e.g. the all-zero state or the maximally mixed state)}, which contains contributions from many high energy states such as $\ket{\psi_{j_1}}$ for some $j_1>0$, there exist efficient transition pathways  $\psi_{j_1} \to \psi_{j_2} \to \cdots \to \psi_0$. }


In the energy eigenbasis, the jump operator takes the form
\begin{equation}\label{eq:jump_operators}
K_a =\sum_{i, j} \hat{f}\left(\lambda_i-\lambda_j\right)\left|\psi_i\right\rangle \left\langle\psi_i|A_a| \psi_j\right\rangle\left\langle\psi_j\right|.
\end{equation}
Here $\{\lambda_i,\ket{\psi_i}\}$ represent eigenpairs of the system Hamiltonian $H$ ordered such that $\lambda_0<\lambda_0+\Delta=\lambda_1\leq \cdots $, and $\hat{f}\left(\omega\right)$ is a filter function in the frequency domain. \rev{The filter function $\hat{f}(\omega)$ is supported only on the negative axis $(-\omega_{\max}, 0)$ for some $\omega_{\max}$ to be specified later. As a result, in the energy eigenbasis, only transitions from $\ket{\psi_j}$ to $\ket{\psi_i}$ satisfying $-\omega_{\max} \le \lambda_i - \lambda_j \le 0$ are allowed. The parameter $\omega_{\max}$ therefore characterizes the maximal energy change per application of the jump operator. Moreover, for any choice of $A_a$, we have $K_a \ket{\psi_0} = 0$, since there is no eigenstate with energy lower than $\lambda_0$. Hence, the ground state $\sigma = |\psi_0\rangle\langle\psi_0|$ is always a fixed point, or stationary state of the dynamics.}

\rev{\cref{eq:jump_operators} expresses the jump operator $K_a$ using the eigendecomposition of $H$. It can be equivalently represented in the time domain as follows. By expressing $\hat{f}(\omega)$ as a filter function in the time domain via the Fourier transform
\begin{equation}
f(s) := \frac{1}{2\pi} \int_{\RR} \hat{f}(\omega) e^{-i\omega s} \, \mathrm{d}\omega,
\end{equation}
and using the spectral decomposition of $H$, we obtain
}
\begin{equation}\label{eqn:jump_time}
K_a =\int_{-\infty}^{\infty} f(s) e^{iHs} A_a e^{-iHs} \mathrm{d}s.
\end{equation}
\rev{Although the construction may appear relatively complicated, we can represent the jump operator coherently on a quantum computer using a block encoding~\cite{GilyenSuLowEtAl2019}. The resulting linear combination of Heisenberg evolutions of $A_a$ involves only queries to Hamiltonian simulation and does not require diagonalizing the Hamiltonian $H$. This, in turn, requires efficiently approximating the integral in \eqref{eqn:jump_time} through an appropriate numerical quadrature scheme. }

\rev{Let $\norm{H}$ and $\Delta$ denote the spectral radius and the spectral gap of the Hamiltonian $H$, respectively. To construct this quadrature, $f(s)$ should decay rapidly as $\abs{s} \to \infty$ so that the integration range can be truncated. By the duality between the real-space and frequency-space representations of a function, $\hat{f}(\omega)$ should be as smooth as possible in the frequency domain, while still allowing the jump operator to efficiently induce a transition from $\ket{\psi_1}$ to $\ket{\psi_0}$. This implies that $\hat{f}(\lambda_0-\lambda_1)=\hat{f}(-\Delta)$ should have a non-negligible value. Together with $\hat{f}(0)=0$, the function $\hat{f}$ must make a sharp transition within an energy window of size $\Delta$. This implies that in the time domain, $f(s)$ is approximately supported on an interval whose size is proportional to $\Delta^{-1}$.}

\rev{For efficient discretization of the integral, we note that $f(s)$ oscillates in the time domain with a wavelength on the order of $\omega_{\max}^{-1}$, which implies that $\omega_{\max}$ should not be chosen excessively large.
Naturally, we choose $\omega_{\max} \le 2\norm{H}$, since no energy transition beyond this range can occur. For simplicity of the analysis, in this work we always choose $\omega_{\max} = 2\norm{H}$. In practice, it is often sufficient to choose $\omega_{\max}$ to be much smaller and independent of the system size. The behavior of $\hat{f}(\omega)$ and $f(s)$ is illustrated in \cref{fig:main_illustrate}~(b),(c), respectively. Based on the discussion above, the construction of this filter function requires only a lower-bound estimate for $\Delta$, and optionally, an upper-bound estimate for $\norm{H}$.}

\vspace{1em}
\paragraph*{\rev{Quasilocality of the jump operator---}}

A fundamental question in dissipative state preparation is as follows: Given a target quantum many-body state $\sigma$, under what conditions must the jump operators be chosen so that $\sigma$ is a fixed point of the dynamics? For pure state preparation $\sigma = \ketbra{\psi}{\psi}$, several necessary conditions on the jump operators are known~\cite{KrausBuchlerDiehlEtAl2008,TicozziViola2012,TicozziViola2014}. In particular, the target state must be annihilated by each jump operator (up to a constant shift)~\cite[Proposition 1]{TicozziViola2014}.
\begin{equation}\label{eqn:annihilation_cond}
K_a \ket{\psi} = 0, \quad \forall a.
\end{equation}
This requirement places strong restrictions on the class of pure states (for example, ground states) that can be prepared using strictly \emph{local} dissipative protocols, where each $K_a$ acts nontrivially only on a fixed number of sites~\cite[Corollary 2]{TicozziViola2014}. A key observation in Ref. \cite{DingChenLin2024} and in this work is that if one allows the jump operators to be \emph{quasilocal} (i.e., nonlocal operators with exponentially decaying tails, see \cref{sec:notation} for the definition), then dissipative protocols still satisfy \cref{eqn:annihilation_cond}, and can provably prepare the ground states of a much broader family of Hamiltonians, including noncommuting ones. We note that in some experimental contexts, the term \emph{quasilocal} has been used to describe few-qubit operators that may still be challenging to realize in practice \cite{KrausBuchlerDiehlEtAl2008}. In contrast, throughout this work we adopt the convention common in mathematics and computer science, where such operators are regarded as local, and reserve the term quasilocal for operators with exponentially decaying support.

\vspace{1em}
\paragraph*{\rev{Efficient quantum simulation of Lindblad dynamics---}}

There are two main strategies for simulating the Lindblad dynamics in \eqref{eq:lindblad}. The standard approach is to employ simulation algorithms that are applicable to general forms of Lindblad dynamics. For example, high-order algorithms~\cite{CleveWang2017,LiWang2023,DingLiLin2024} can achieve near-optimal simulation cost per unit time.
Both the construction of $K_a$ and the simulation algorithms~\cite{CleveWang2017,LiWang2023,DingLiLin2024} can use multiple ancilla qubits, complex control logic and are suitable only on full fault-tolerant quantum computers.

The second, and simpler, strategy for simulating the Lindblad dynamics in \eqref{eq:lindblad} is to exploit the specific form of the jump operator in \cref{eqn:jump_time}. In particular, this approach does not explicitly construct $K_a$ but instead embeds it directly into the simulation algorithm.
This in turn can lead to algorithms that overall use only a single ancilla qubit and can be much more suitable for early fault-tolerant quantum devices. We will not discuss such algorithms in detail and refer the reader to \cite[Section III]{DingChenLin2024} for descriptions of such algorithms. We note that while the simulation algorithm proposed in~\cite[Section III]{DingChenLin2024} is designed for simulating a single jump operator using a single ancilla, when multiple jump operators are present, operator splitting can be applied to handle each jump operator separately.


\vspace{1em}
\paragraph*{\rev{Mixing time---}}

Dissipative protocols \revvv{prepare} the ground state as the stationary state of the dynamics. The total simulation time can be characterized by the \emph{mixing time}, which denotes the minimal time required for the system to  drive \textit{any} initial state to one that is close to the ground state.
This closeness can be characterized by the trace distance.
Let $\lVert A\rVert_{1}=\operatorname{Tr}\!\big(\sqrt{A^{\dagger}A}\big)$ denote the trace norm. The trace distance between density matrices $\rho_{1}, \rho_{2}$ is $D(\rho_{1},\rho_{2})=\tfrac12 \lVert \rho_{1}-\rho_{2}\rVert_{1}$.
This metric has a direct operational meaning, since for any bounded observable $O$, $\bigl|\operatorname{Tr}\!\bigl[O\bigl(\rho_1-\rho_2\bigr)\bigr]\bigr| \le \norm{O} \norm{\rho_1-\rho_2}_1 \le 2\norm{O} D(\rho_1,\rho_2)$.
Let $\sigma$ be a stationary state of the dynamics generated by $\mathcal{L}$. The mixing time with respect to the trace distance is defined as
\begin{equation}\label{eqn:mixing_tracedistance_rho_0} \tau_{\operatorname{mix}}(\eta)=\min\left\{\, t \mid D\!\big(e^{t\mathcal{L}}(\rho_{0}),\sigma\big)\le \eta,\,\forall\rho_0\right\}. \end{equation}

This definition of mixing time is widely used in theoretical analysis~\cite{TemmeKastoryanoRuskaiEtAl2010,KastoryanoTemme2013,BardetCapelGaoEtAl2023,rouz2024,DingLiLinZhang2024,kochanowski2024rapid,rouze2024optimal,gamarnik2024slow,tong2024fast}. However, in practice, the trace distance can be difficult to evaluate, and one usually cares about the mixing time for a given initial state $\rho_0$, so one may use surrogate notions of mixing  defined through physically meaningful quantities with respect to an initial state $\rho_0$. For instance, the energy-based mixing time is
\begin{equation}\label{eqn:tau_E}
\tau^{E}_{\operatorname{mix}}(\eta;\rho_0)=\min\left\{\, t \mid \bigl|\Tr[H\,e^{t\mathcal{L}}(\rho_{0})]-\lambda_{0}\bigr|\le \eta\right\}.
\end{equation}
For two density matrices $\rho,\sigma$, the fidelity $F(\rho, \sigma)= \operatorname{Tr}\left[\sqrt{\rho^{\frac{1}{2}} \sigma \rho^{\frac{1}{2}}}\right]$. When $\sigma=\lvert\psi_{0}\rangle\langle\psi_{0}\rvert$ is a pure state ($\rho$ can be a pure or mixed state), the fidelity simplifies to $F(\rho,\sigma)=\sqrt{\Tr[\braket{\psi_0|\rho|\psi_0}]}$. The fidelity-based mixing time is
\begin{equation}\label{eqn:tau_F}
\tau^{F}_{\operatorname{mix}}(\eta;\rho_0)=\min\left\{\, t \mid 1-F^2(e^{t\mathcal{L}}(\rho_{0}),\sigma)\le \eta\right\}.
\end{equation}

Even though $\tau^{E}_{\operatorname{mix}}$ and $\tau^{F}_{\operatorname{mix}}$ measure convergence through specific variables, when maximizing over all initial states $\rho_0$, they can provide both an upper bound and a lower bound of the mixing time defined via trace distance. These relations are derived in \cref{eqn:mixingtime_relations} in \cref{sec:compare_mixing_time}.

We will specify the notion used in the numerical results below. The theoretical justification will be provided directly for the mixing time in terms of the trace distance.

\vspace{1em}
\paragraph*{\rev{Resource estimate---}}

The total cost for dissipative ground-state preparation using Lindblad dynamics can be decomposed into three components: the cost associated with constructing the Lindbladian and in particular jump operators $K_a$, denoted by $C_{\mc{L}}$; the cost of simulating Lindblad dynamics per unit time denoted by $C_S$, using a Lindblad simulation algorithm as discussed earlier; and $\tau_{\operatorname{mix}}$, an upper bound on the total simulation time. Given these factors, the end-to-end resource cost is
\begin{equation}
\text{End-to-end cost} = C_{\mc{L}} \times C_S \times \tau_{\operatorname{mix}}.
\end{equation}

Using the standard approach for simulating the Lindblad dynamics, the cost for constructing $K_a$ to precision $\epsilon$ is $\wt{\Or}(\omega_{\max}\Delta^{-1}\log(1/\epsilon))$~\cite{DingChenLin2024} (see \cref{sec:notation} for the meaning of the notation $\widetilde{\Or}$).
Let $\|\mathcal{L}\|_{\mathrm{be}} := \|H\| + \frac{1}{2}\sum_{a}\|K_a\|^2$. Using the algorithm in \cite{LiWang2023} for simulating the Lindblad dynamics \eqref{eq:lindblad} up to time $t$ with precision $\epsilon$ the cost is $\mathcal{O}\left(t\|\mathcal{L}\|_{\mathrm{be}}\log\left(\frac{t\|\mathcal{L}\|_{\mathrm{be}}}{\epsilon}\right)\right)$.
For a typical physical Hamiltonian (spin, fermion, etc.) defined on $N$ sites, $\|\mathcal{L}\|_{\mathrm{be}}=\poly(N)$ and
the end-to-end cost is
\begin{equation}
\wt{\Or}(\tau_{\operatorname{mix}} \Delta^{-1} \poly(N) \polylog(1/\epsilon)).
\end{equation}
The simplified algorithm in~\cite[Section III]{DingChenLin2024} combines the step of generating $K_a$ and simulating the dynamics. It can be viewed as a first-order algorithm for simulating the continuous-time Lindblad dynamics. The end-to-end cost is~\cite[Theorem 1]{DingChenLin2024}
\begin{equation}
\wt{\Or}(\tau^2_{\operatorname{mix}} \Delta^{-1} \poly(N)/\epsilon).
\end{equation}
\cite[Theorem 2]{DingChenLin2024} further presents a discrete-time algorithm that reduces the cost from quadratic to nearly linear in $\tau_{\operatorname{mix}}$, which we do not discuss here.

Thus, for the end-to-end cost to scale polynomially with the system size $N$, the most important and challenging task is to estimate the mixing time and establish that $\tau_{\operatorname{mix}} = \poly(N)$, a property often referred to as \emph{fast mixing}. In some cases, an even stronger bound $\tau_{\operatorname{mix}} = \polylog(N)$ can be proved, which is known as \emph{rapid mixing}.

The remainder of this manuscript focuses on characterizing the effectiveness of the Lindblad dynamics for ground-state preparation quantified by mixing times.

\begin{figure}
\begin{center}
\includegraphics[width=0.48\textwidth]{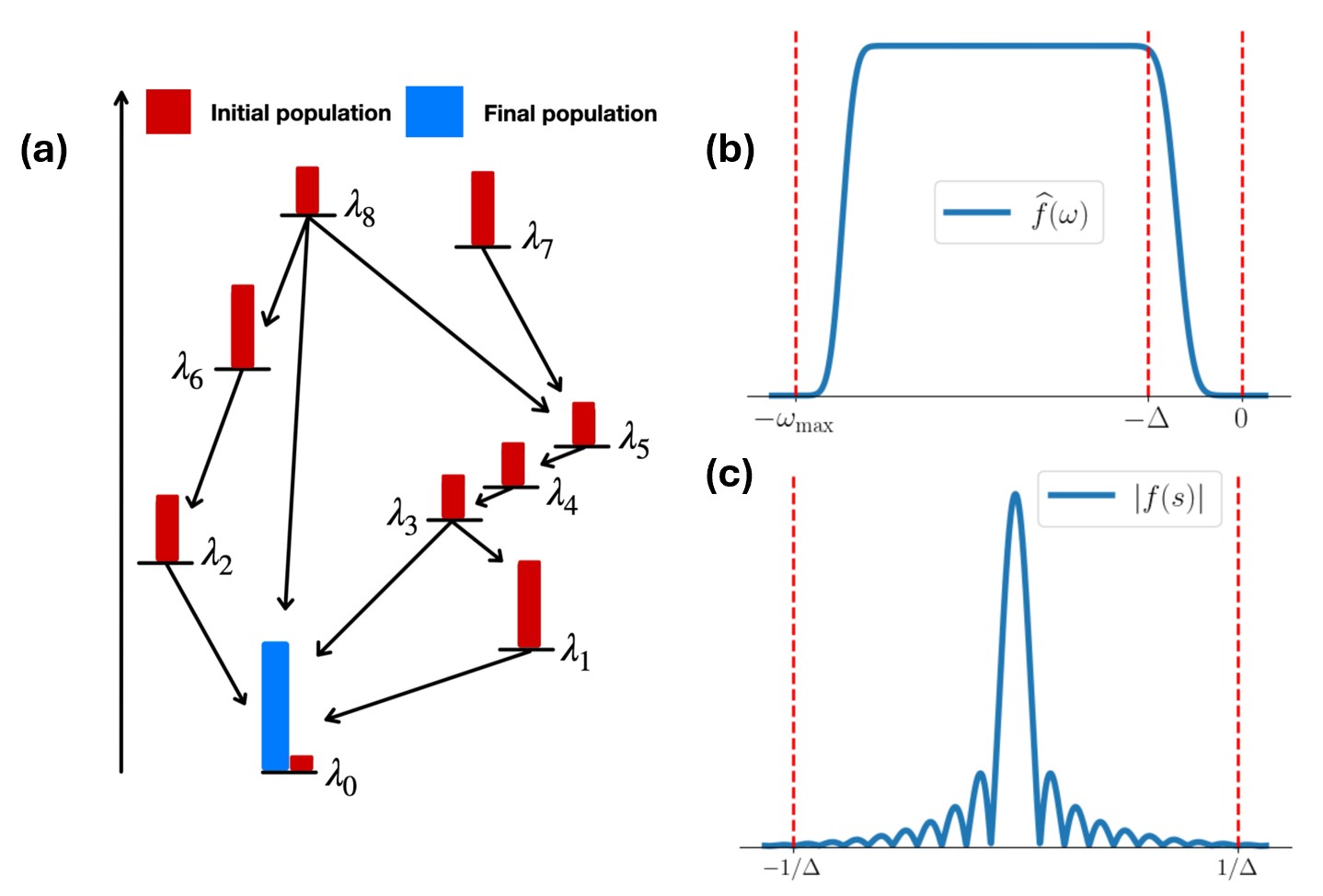}
\end{center}
\caption{\rev{(a) Schematic representation of the ground state preparation algorithm, in which high-energy components are systematically dissipated into lower-energy states until convergence to the ground state is achieved. (b) and (c) depict the associated filter function in the frequency and time domains, respectively.}}
\label{fig:main_illustrate}
\end{figure}

\section{Summary of results}\label{sec:summary}

\subsection{Numerical results}

\rev{We perform numerical simulations for both quasi-free and general dissipative dynamics. For general dissipative dynamics, we develop a new numerical simulation algorithm based on tensor networks, which is described in \cref{sec:TNmethod}.}

\vspace{1em}
\noindent\rev{\emph{Quasi-free dynamics (\cref{sec:numer_quasifree}):}}

We begin by exploring \emph{quasi-free} dissipative dynamics~\cite{Prosen2008,BarthelZhang2022}. \rev{Lindblad dynamics is called quasi-free if the Hamiltonian is quadratic in Majorana operators and jump operators are linear in Majorana operators.}
A hallmark of such systems is that physical observables, such as covariance matrices, form a closed set of equations. This enables efficient simulations of these observables for large systems. Utilizing this framework, we demonstrate numerically that the ground state of a translationally invariant 1D transverse field Ising model (TFIM) chain can be efficiently prepared, even when cooling is applied only at the boundaries of the chain. We observe that, \rev{with boundary dissipation, the mixing time as defined by physical observables scales approximately cubically with the system size, which is consistent with findings from prior numerical studies using different dissipative protocols \cite{PhysRevA.104.012414,Meghana_2020,LloydMichailidisMiEtAl2025}.}



We also observe that \rev{boundary dissipation} efficiently prepares the ground state of a cluster state Hamiltonian with a symmetry-protected topological (SPT) ground state phase. Starting from a trivial topological phase, we find that the protocol allows crossing the phase boundary, as indicated by changes in string order parameters (SOP). In all \rev{the examples we have studied of systems subjected to boundary dissipation,  we observed that} the coherent term in the Lindblad dynamics is essential for achieving convergence, even though it vanishes when applied to the ground state.

\vspace{1em}
\noindent\rev{\emph{General Lindblad dynamics (\cref{sec:TNmethod},~\cref{sec:numer_TNmethod}):}}


For general dissipative dynamics that are not quasi-free, we propose a new numerical algorithm based on the \rev{tensor network methods} to efficiently represent \revvv{jump operators} and the Lindbladian~\cite{sander2025large,PhysRevE.103.L040102}. Using this algorithm, we study the mixing time required to prepare the ground state of 1D anisotropic Heisenberg models in a magnetic field, which includes the TFIM as a special case. Dissipation is applied to each spin site (referred to as bulk dissipation), and the resulting dynamics is not quasi-free even for the integrable TFIM Hamiltonian. \rev{Our numerical results show that the Hamiltonian with bulk dissipation exhibits rapid mixing};
this mixing time scaling also applies to spin systems with weak random perturbations in their on-site interactions, whose ground states cannot be efficiently prepared by boundary dissipation alone due to the obstruction caused by Anderson-type localization. We further verify the robustness of our approach using a nonintegrable cluster-state Hamiltonian, which has a ground state in a symmetry-protected topological (SPT) phase.

\subsection{Theoretical results}

To gain an analytical understanding of the convergence behavior of our dissipative protocol, we provide theoretical guarantees that rigorously establish upper bounds on the mixing time for several physically relevant Hamiltonians. As in our numerical studies, we begin with quasi-free systems and present a general theorem upper bounding the mixing time in this case. We then move beyond quasi-free systems and demonstrate rapid mixing for weakly interacting spin and fermionic systems.

\vspace{1em}
\noindent\rev{\emph{Quasi-free dynamics (\cref{sec:mixingtime_quasifree}):}}

\rev{First, we note that rigorously establishing the mixing time in terms of trace distance poses significant challenges, even for quasi-free systems.} Previous analyses of mixing time estimates, including those for quasi-free systems, typically relied on the assumption that the stationary state is invertible, making it difficult to extend these results to ground state preparations \cite{TemmePastawskiKastoryano2014,li2024dissipative}.
We develop a new method that can overcome this difficulty by examining the spectral properties of a non-Hermitian Hamiltonian. Specifically, in the absence of dissipation, the eigenvalues of the Lindblad dynamics lie entirely on the imaginary axis. With dissipation, we \rev{show} that the mixing time measured by the trace distance is determined by the gap between the eigenvalues of this non-Hermitian Hamiltonian and the imaginary axis.

This new approach enables explicit estimates of the convergence rate in trace distance, with or without a coherent term. In particular, it \rev{proves} that the mixing time of the 1D translationally invariant TFIM with boundary dissipation scales as $\mathcal{O}(N^3 \log N)$, which is consistent with our numerical results.
The cubic scaling of the mixing time is mainly due to the long-wavelength modes, which perturb the eigenvalues of the aforementioned non-Hermitian Hamiltonian away from the imaginary axis by an amount proportional to $N^{-3}$.

\vspace{1em}
\noindent\rev{\emph{General Lindblad dynamics for weakly interacting systems (\cref{sec:provable_rapid},~\cref{sec:provable_rapid_fermion}):}}

\rev{Beyond quasi-free systems,} we consider the preparation of the ground state of a weakly interacting spin Hamiltonian in an arbitrary finite dimension. Specifically, the Hamiltonian is expressed as $H = H_0 + \varepsilon H_1$, where $H_0$ is a gapped Hamiltonian composed of noninteracting terms, and the interaction strength $\varepsilon$ is assumed to be smaller than a constant that is independent of the system size. We establish the convergence of the density matrix by analyzing the convergence of observables in the Heisenberg picture, measured through a quantity known as the oscillator norm, which measures the deviation of an observable from the identity under Heisenberg evolution with the Lindbladian. This strategy was recently utilized in the analysis of mixing times of quantum Gibbs samplers in the high-temperature regime (small inverse temperature $\beta$)~\cite{rouze2024optimal}.

Our first observation is that the definition of the oscillator norm does not rely on an invertible stationary state, and thus serves as a plausible candidate for characterizing convergence to the ground state. However, we need to modify the definition of the oscillator norm of an observable $O$ to track separately the deviation of $O$ from the identity operator along on-diagonal and off-diagonal directions.
\rev{In the presence of perturbations, our proof employs a Lieb-Robinson bound adapted to the ground state setting.}
By integrating these elements, we establish a new stability result for the convergence rate of the oscillator norm, which provides the first rigorous proof of ground state preparation protocols for noncommuting Hamiltonians.

Finally, we extend the result of weakly interacting spin systems to weakly interacting fermionic systems. The fermionic creation and annihilation operators are nonlocal in the spin basis. Therefore we need to employ a fermionic version of the partial trace to define the oscillator norm. We then prove that this modified definition of the oscillator norm can effectively characterize the rapid convergence of observables in the fermionic setting. We conclude that, for bulk dissipation, both weakly interacting spin systems and weakly interacting fermionic systems exhibit rapid mixing.

\section{Classical simulation algorithm of Lindblad dynamics}\label{sec:TNmethod}

For general Lindblad dynamics, we need to simulate the dynamics in \cref{eq:lindblad} directly to estimate the mixing time. For system sizes beyond the reach of exact diagonalization (ED), we propose an algorithm that constructs the jump operators and propagates the Lindblad dynamics using a matrix product operator (MPO) formulation. Recall that an MPO on an $N$-site system (each site of local dimension $d$) with bond dimension $D$ can be written as
\begin{equation}
\begin{aligned}
M=\sum_{s_1, s_1^{\prime}, \ldots, s_N, s_N^{\prime}=1}^d &
\Bigl(M_1^{s_1, s_1^{\prime}} \cdots M_N^{s_N, s_N^{\prime}}\Bigr)\cdot
\\
&\ket{s_1, \ldots, s_N}\!\bra{s_1^{\prime}, \ldots, s_N^{\prime}},
\end{aligned}
\end{equation}
where $M_1^{s_1, s_1^{\prime}} \cdots M_N^{s_N, s_N^{\prime}}$ defines the corresponding matrix product, with $M_1^{s_1, s_1^{\prime}}\in \CC^{1\times D}, M_N^{s_N, s_N^{\prime}}\in \CC^{D\times 1}$, and  $M_i^{s_i, s_i^{\prime}}\in \CC^{D\times D}$ for $2\le i\le N-1$. The cost of storing the matrix products is $\Or(d^2 D^2 N)$, which scales linearly with the system size $N$.

To construct the jump operators $K_a$, we start by representing both the coupling operators $A_a$ and the Hamiltonian $H$ as MPOs. We then compute the Heisenberg evolution $e^{i H s}A_a e^{-i H s}$ using the time-evolving block-decimation (TEBD) algorithm~\cite{PhysRevLett.93.040502,PhysRevLett.93.207204}. Next, we approximate the integral in \cref{eqn:jump_time} by a quadrature rule,
\begin{equation}
K_a \approx \sum_i p_i\, f(s_i)\, e^{i H s_i} A_a\, e^{-i H s_i},
\end{equation}
and compress the resulting sum to maintain a manageable bond dimension. This yields the MPO representation of the jump operator $K_a$ (see \cref{fig:TNdiagramjump} for an illustration).

\begin{figure}
\centering
\includegraphics[width=1.0\linewidth]{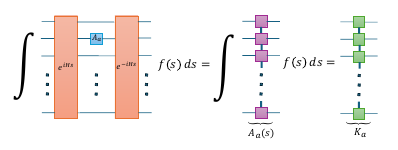}
\caption{%
\label{fig:TNdiagramjump}%
Tensor network representation of the jump operator $K_a$ in the MPO form associated from a local coupling operator $A_a$. First, the MPO representation of the operator $e^{iH s_j} A_a e^{-iH s_j}$ is constructed for a set of time steps $\{s_j\}$. Next, a weighted summation from discretizing the integral
$\sum_j p_j f(s_j) e^{iH s_j} A_a e^{-iH s_j}$ is performed to combine these operators. Finally, the resulting MPO is compressed to reduce the bond dimension, yielding an efficient representation of $K_a$.
}
\end{figure}
In practice, since $A_a$ is an operator rather than a state, the TEBD algorithm is implemented by vectorizing $A_a$ into an matrix produc state (MPS) (often referred to as the Choi isomorphism) \cite{weimer2021simulation}. Concretely, we reshape each site's row and column indices into a single combined index of dimension $d^2$, as illustrated in \cref{fig:ep1}:

\begin{equation}
\ket{M}_{\sharp} = \sum_{i_1=0}^{d^2 - 1} \cdots \sum_{i_N=0}^{d^2 - 1}
c_{i_1,\ldots,i_N} \,\ket{i_1}_{\sharp} \otimes \cdots \otimes \ket{i_N}_{\sharp}.
\end{equation}
In a tensor network diagram, this corresponds to ``gluing'' the row and column indices together on each site. Converting back from an MPS to an MPO is achieved by splitting each combined index back into two separate indices.
Under vectorization, the Heisenberg evolution can be written as

\begin{equation}
\ket{\,e^{iHs} A_a \, e^{-iHs}}_{\sharp} \;=\; e^{-i\,(I \otimes H^T \,-\, H \otimes I)\,s}\,\ket{A_a}_{\sharp},
\end{equation}
so that the standard TEBD algorithm can be applied directly. To compute the jump operator $K_a$, we proceed as follows: (i) convert $A_a$ from its MPO form into an MPS using the Choi isomorphism, (ii) apply TEBD to evolve the vectorized operator over time, and (iii) perform a summation over the discrete time steps, followed by bond dimension compression. The resulting MPS is then converted back into an MPO, yielding $K_a$.

\begin{figure}
\centering
\includegraphics[width=0.6\linewidth]{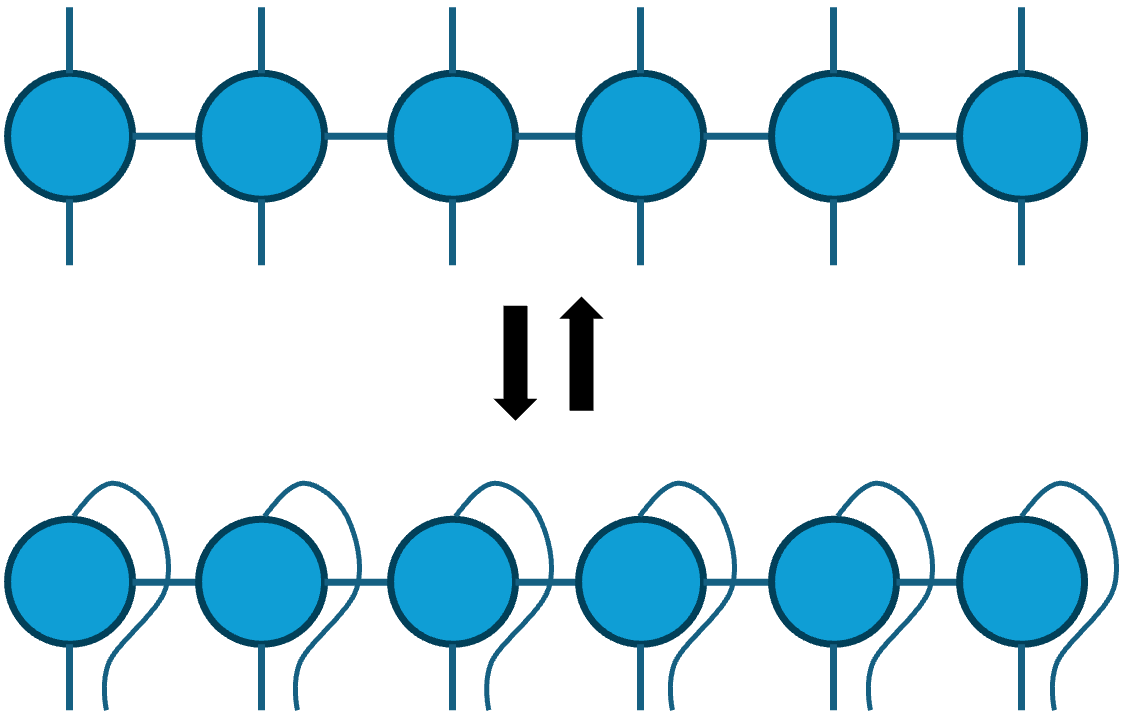}
\caption{%
\label{fig:ep1}%
Illustration of vectorizing an MPO into an MPS using the Choi isomorphism and then converting it back. Each pair of local site indices in the MPO is reshaped into a single combined index, allowing standard MPS techniques, such as TEBD, to be applied. The reverse process restores the original MPO from the MPS by reshaping the combined indices back into pairs.}
\end{figure}

Once each jump operator is expressed as an MPO, and given that $\rho(t)$ is also stored as an MPO, we need to evaluate the right-hand side of \cref{eq:lindblad}, i.e., $\mathcal{L}(\rho(t))$. This \revvv{is illustrated} using tensor network diagrams in \cref{fig:TNdiagram}. However, direct multiplication and addition of MPOs \revvv{tend to increase} the bond dimension quickly. For example, in the absence of a compression step, multiplying two MPOs with bond dimension $D$ results in an MPO with bond dimension $D^2$, while adding two MPOs yields an MPO with bond dimension $2D$. If we choose $\{A_a\}$ to be the set of all Pauli matrices, and assume every operator in the Lindbladian has bond dimension $D$, the bond dimension of the MPO representation for $\mathcal{L}(\rho)$ would become $\Or(ND^3)$.

Forming such an MPO and then compressing it would have an onerous cost of $\mathcal{O}(D^9)$.
Instead, we directly fit an MPO of bond dimension $D$ to the
uncontracted sum of triple MPO products as depicted in Fig.~\ref{fig:TNdiagram}, adapting the method of~\cite{VerstraeteCirac2004}.
This only requires computing the overlap of the ansatz with these terms (see Appendix~\ref{sec:additional_numer_tensor}).
The initial guess is chosen as the ``zip-up'' compression~\cite{StoudenmireWhite2010} of the first term.
Both this and the subsequent fitting iterations have a cost of $\mathcal{O}(D^5)$.

After obtaining the compressed MPO representation of $\mathcal{L}(\rho)$, we may employ any suitable numerical integrator to propagate $\rho(t)$ forward in time. For large systems, each evaluation of $\mathcal{L}(\rho)$ is expensive, so it is beneficial to minimize the number of function evaluations. For bulk dissipation, the cost of evaluating $\mathcal{L}(\rho)$ is large, but the mixing time can be very short. Therefore we adopt a simple forward-Euler method. For boundary dissipation, the mixing time can be much longer, and there we employ a more accurate 4th-order Runge-Kutta method instead. More advanced solvers can be explored in future studies to improve the accuracy or efficiency.

\begin{figure}
\centering
\includegraphics[width=1\linewidth]{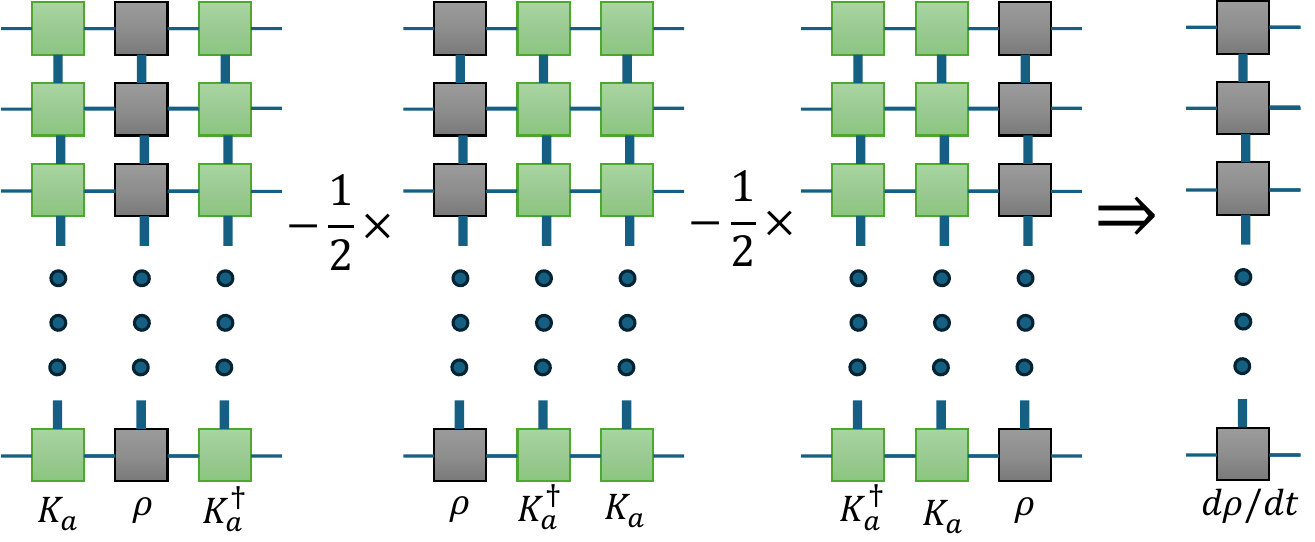}
\caption{%
\label{fig:TNdiagram}%
Illustration of the tensor network computation associated with a single jump operator $K_a$ in $\mc{L}(\rho)$. }
\end{figure}

\section{Numerical Results}
\label{sec:numerics}
\subsection{Quasi-free dissipative dynamics}\label{sec:numer_quasifree}

\paragraph*{Mapping spin systems to quasi-free dynamics---}

The Hamiltonian of a 1D translationally invariant Transverse Field Ising Model (TFIM) with open boundary conditions is
\begin{equation}\label{eqn:H_TFIM}
    H=-g \sum_{i=1}^N Z_i-J \sum_{i=1}^{N-1} X_i X_{i+1}.
\end{equation}
Using the Jordan-Wigner transformation, the Hamiltonian can be written as a quadratic Majorana operator with $2N$ modes (\cref{sec:prelim_majorana})
\begin{equation}\label{eqn:H_TFIM_majorana}
        H =2i J \sum_{j=1}^{N-1} w_{j+N}w_{j+1} +2i g \sum_{j=1}^N  w_{j}w_{j+N}.
\end{equation}
We choose the coupling operators to be Pauli matrices  $X_1$ and $Y_1$ on the boundary of the chain, which are linear in the Majorana operators: $X_1=\sqrt{2}w_1, Y_1=\sqrt{2}w_{1+N}$.
By Thouless's theorem, the Heisenberg evolution of a single Majorana operator under a quadratic Hamiltonian is still linear in Majorana operators.
Therefore the jump operator in \cref{eqn:jump_time} can be expressed as
\begin{equation}\label{eqn:K_a}
K_a=\sum_{j=1}^{2N} \zeta_{ja} w_{j}, \quad A_a\in\{X_1,Y_1\},
\end{equation}
for some coefficients $\zeta_{ja}\in\CC$.

Lindblad dynamics with a Hamiltonian term that is quadratic in Majorana operators, and jump operators linear in Majorana operators is called \emph{quasi-free}.
Using a vectorization process known as ``third quantization''~\cite{Prosen2008,BarthelZhang2022}, each term in the vectorized Lindbladian becomes quadratic in an enlarged set of Majorana operators. Physical observables of a quasi-free dynamics, such as the covariance matrix
\begin{equation}
\Gamma_{pq}=i\braket{w_p w_q}-\frac{i}{2} \delta_{pq},
\end{equation}
 form a closed set of equations, which involves only a matrix of size $2N\times 2N$ in time (\cref{sec:prelim_majorana}). This in turn can be used to evaluate other physical quantities such as the energy. Specifically, for a Hamiltonian quadratic in Majorana operators,
$H=\sum_{p,q=1}^{2N} h_{pq} w_p w_q$,
where $h$ is a Hermitian and purely imaginary (and thus traceless) coefficient matrix, the energy is given by
\begin{equation}\label{eqn:energy_quadratic}
E = \sum_{pq}  h_{pq} \braket{w_p w_q}=i\Tr[ h^T \Gamma].
\end{equation}
For quasi-free systems with
Gaussian initial states, higher order covariance matrices are determined by the covariance matrix $\Gamma$ according to  Wick's theorem~\cite{SuraceTagliacozzo2022}.

\vspace{1em}
\paragraph*{TFIM with boundary dissipation---}

Similar to the derivation in \cref{sec:prelim_majorana}, if we choose the Pauli operators $X_N,Y_N$ on the other end of the boundary, the resulting Lindblad dynamics is also quasi-free. Hence we choose $\{X_1, Y_1, X_N, Y_N\}$ to be the coupling operators, and construct the corresponding jump operators according to \cref{eqn:jump_time}.

\begin{figure}[ht]
  \centering
    \begin{subfigure}[b]{0.35\textwidth}
        \includegraphics[width=\textwidth]{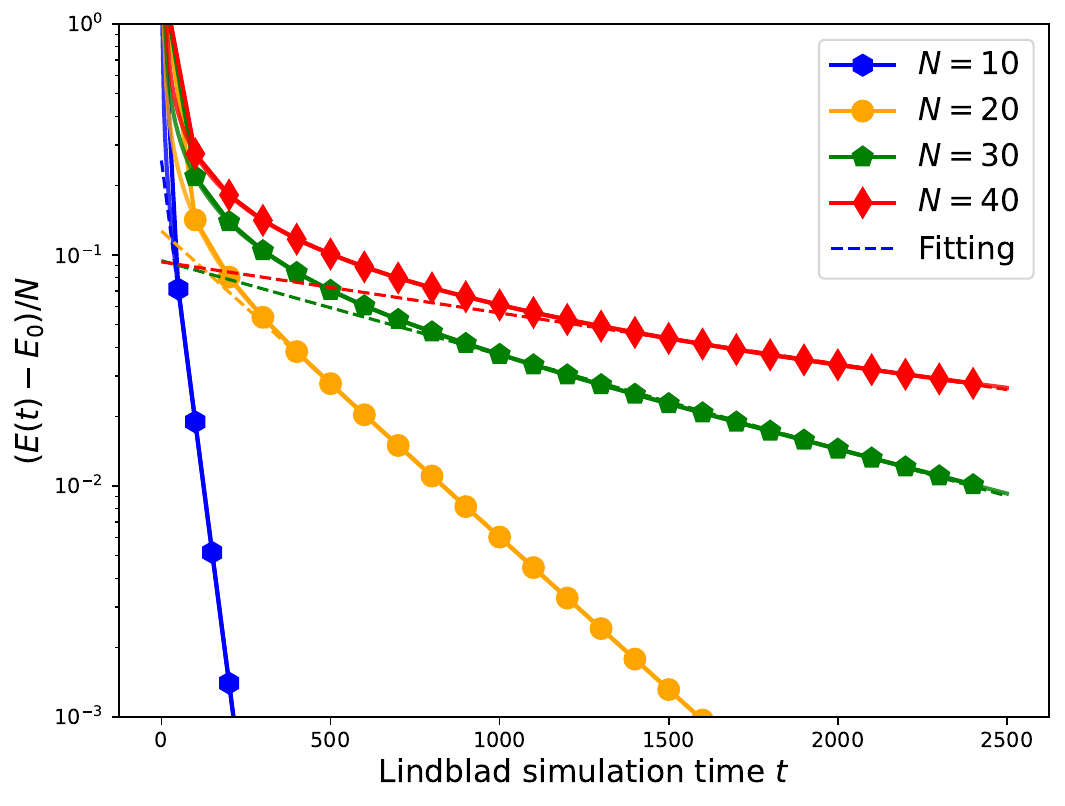}
        \caption{ \centering }
    \end{subfigure}
    \begin{subfigure}[b]{0.35\textwidth}
        \includegraphics[width=\textwidth]{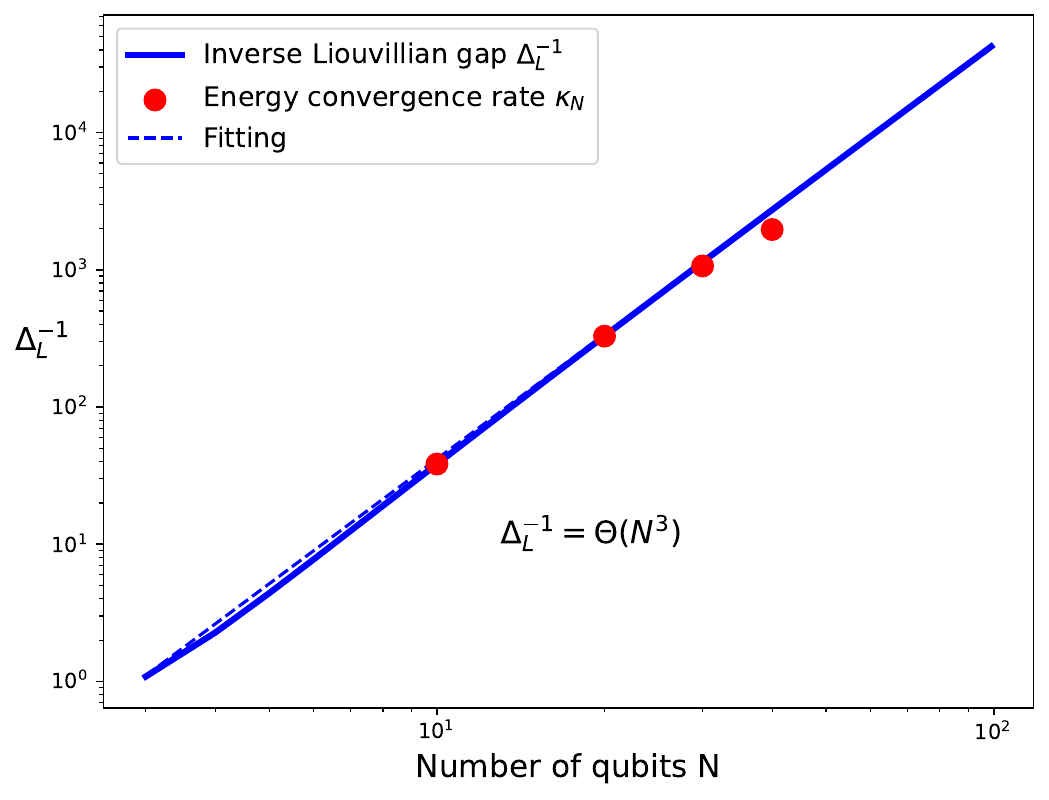}
        \caption{ \centering }
    \end{subfigure}
  \caption{
  Numerical results of 1D TFIM \eqref{eqn:H_TFIM} with $J=1,g=1.5$, using $\{A_a\}=\{X_1, Y_1, X_N, Y_N\}$ as coupling operators. (a) Convergence of energy starting from the maximally mixed state. The dashed lines are exponential fits of the asymptotic behaviour of energy decay, meaning $(E(t)-\rev{E_0})/N=\Theta(\exp(-\kappa_N t))$ for constants $\kappa_N$ when $t$ is sufficiently large. (b) The scaling of the inverse of the Liouvillian gap $\Delta^{-1}_{\mathcal{L}}$ with respect to the system size $N$. Red points are the energy convergence rate $\kappa_N$ calculated from fitting the data in (a).}
  \label{sec:Results: Fig: tfimdynamics}
\end{figure}

Using the covariance matrix $\Gamma(t)$, we can evaluate the energy $E(t)$ via \cref{eqn:energy_quadratic}.  We note that the convergence of the many-body density matrix cannot always be inferred from the covariance matrix $\Gamma(t)$. As a surrogate, in this section, we characterize the mixing time \rev{in terms of how rapidly the energy per site converges to its value in the ground state.
Specifically, we define the mixing time as in \cref{eqn:tau_E}, except that we measure the convergence in terms of the energy per site, and start from a specific initial state, namely the maximally mixed state $\rho_0 = \mathbb{I} / 2^N$.}


\cref{sec:Results: Fig: tfimdynamics}a demonstrates the energy decay of the boundary-dissipated TFIM under Lindbladian dynamics.
 Numerical simulations show that the energy rapidly converges towards the ground state energy initially, and then enters an asymptotic exponentially decaying regime $\propto e^{-\Delta_{\mc{L}} t}$. The convergence rate $\Delta_{\mc{L}}$ is the gap of the Lindbladian (also called the Liouvillian gap). We may extract $\Delta_{\mathcal{L}}$ using an exponential fit of the dynamics, and can also directly compute $\Delta_{\mc{L}}$ by means of the rapidity spectrum for quasi-free systems~\cite{Prosen2010}. In \cref{sec:Results: Fig: tfimdynamics}b, we show how the Liouvillian gap scales with the system size. \rev{The estimates for $\Delta_{\mathcal{L}}$ from the slopes in \cref{sec:Results: Fig: tfimdynamics}a yield excellent agreement with the spectrum calculations in \cref{sec:tfim_boundary_proof}. Using a log-log scaling for the axes, we find that $\Delta_{\mc{L}}=\Theta(N^{-3})$, which matches the scaling of the energy-based mixing time $\tau^E_{\operatorname{mix}}=\Theta(N^3)$ for fixed $\eta$.}

\vspace{1em}
\paragraph*{Cluster state Hamiltonian with boundary dissipation---}

The 1D cluster state Hamiltonian on $N$ sites takes the form
\begin{equation}\label{eqn:H_SPT}
    H=-J\sum_{j=1}^{N-2} X_j Z_{j+1} X_{j+2}-h_1 \sum_{j=1}^N Z_j.
\end{equation}
This system plays a role in measurement-based quantum computing \cite{briegel2009measurement,raussendorf2003measurement}, and  exhibits an interesting symmetry protected topological (SPT) phase, with a fourfold degenerate ground state in the thermodynamic limit \cite{SonAmicoFazioEtAl2011,CongChoiLukin2019}.
The field strength $h_1/J$ drives a phase transition between a simple paramagnetic phase and the SPT phase. Using the Jordan-Wigner transformation, the Hamiltonian can be expressed as a quadratic Majorana operator
\begin{equation}\label{eqn:H_cluster_majorana}
H =2i J \sum_{j=1}^{N-\revvv{2}} w_{j+N}w_{j+2} +2i h_1 \sum_{j=1}^N  w_{j}w_{j+N},
\end{equation}
which is similar to a Kitaev chain Hamiltonian with next-nearest-neighbor (NNN) couplings. The string order parameter (SOP) is a nonlocal order parameter which can be used to distinguish between the SPT phase and the paramagnetic phase. It is defined  as 
\begin{equation}
\begin{aligned}
    S_{a b}&=\revvv{s_m} X_a Z_{a+1} Z_{a+3} \ldots Z_{b-3} Z_{b-1} X_b\\
    &= \revvv{(2i)^{m}}\prod _{i=1} ^{m} w_{a+2i}\prod _{i=0} ^{m-1} w_{a+\revvv{N}+2i},
\end{aligned}
\end{equation}
where $a,b$ are arbitrary starting and ending points in the bulk with $b-a$ being an even number, \revvv{$m=(b-a)/2$, and $s_m=(-1)^{m(m-1)/2}$ is a sign coming from the Jordan--Wigner transformation.} The SOP can be computed using Wick's theorem \cite{SuraceTagliacozzo2022}
\begin{equation}
    \left\langle S_{a b}\right\rangle=\revvv{(2i)^{m}}\left\langle \prod _{i=1} ^{m} w_{a+2i}\prod _{i=0} ^{m-1} w_{a+N+2i} \right\rangle = \operatorname{Pf}\left(2\Gamma_{\left.\right|\{q \}} \right),
\end{equation}
where $\Gamma_{\left.\right|\{q \}}$ is the covariance matrix restricted to the indices $\{ q\}= \{a+2i, a+\revvv{N}+2(i-1) | i=1,\ldots \revvv{(b-a)/2}\}$ and $\operatorname{Pf}$ denotes the Pfaffian.


We choose the coupling operators to be two single Pauli operators on two ends of the boundary $A_1=Y_1, A_2=Y_N$.
The  corresponding jump operators are linear in Majorana operators, and the dissipative dynamics is thus quasi-free. The presence of zero energy edge modes in the SPT phase of the system lead\revvv{s} to nearly degenerate ground states with energy gaps closing exponentially rapidly as the system size increases.  However, the closing of the energy gap is entirely due to the presence of the edge modes which is irrelevant for bulk properties such as the SOP.
Therefore we may define an effective gap, denoted by $\Delta_{\mathcal{L},\rm{eff}}$, by excluding the eigenvalues in the Liouvillian exponentially clustering near $0$, and choose the parameters in the filter function $\hat{f}(\omega)$ based on this effective gap. When choosing effective gap $\Delta_{f}=0.1$ in $f$, the resulting dissipative dynamics converges to a statistical mixture of the nearly degenerate states with the same SOP value in the thermodynamic limit (\cref{fig:SOP}). The SOP is evaluated by setting $a,b$ to be the two ends of the chain. We find again that  boundary dissipation alone is sufficient to drive the system from a paramagnetic phase towards SPT phase.

\begin{figure}[ht]
  \centering
    \begin{subfigure}[b]{0.35\textwidth}
        \includegraphics[width=\textwidth]{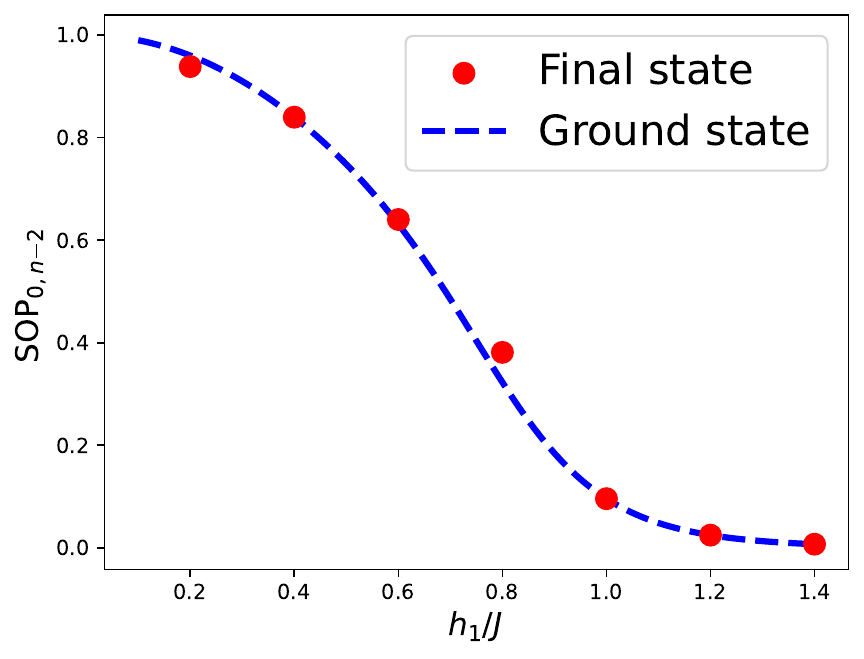}
        \caption{\centering}
    \end{subfigure}
    \begin{subfigure}[b]{0.35\textwidth}
        \includegraphics[width=\textwidth]{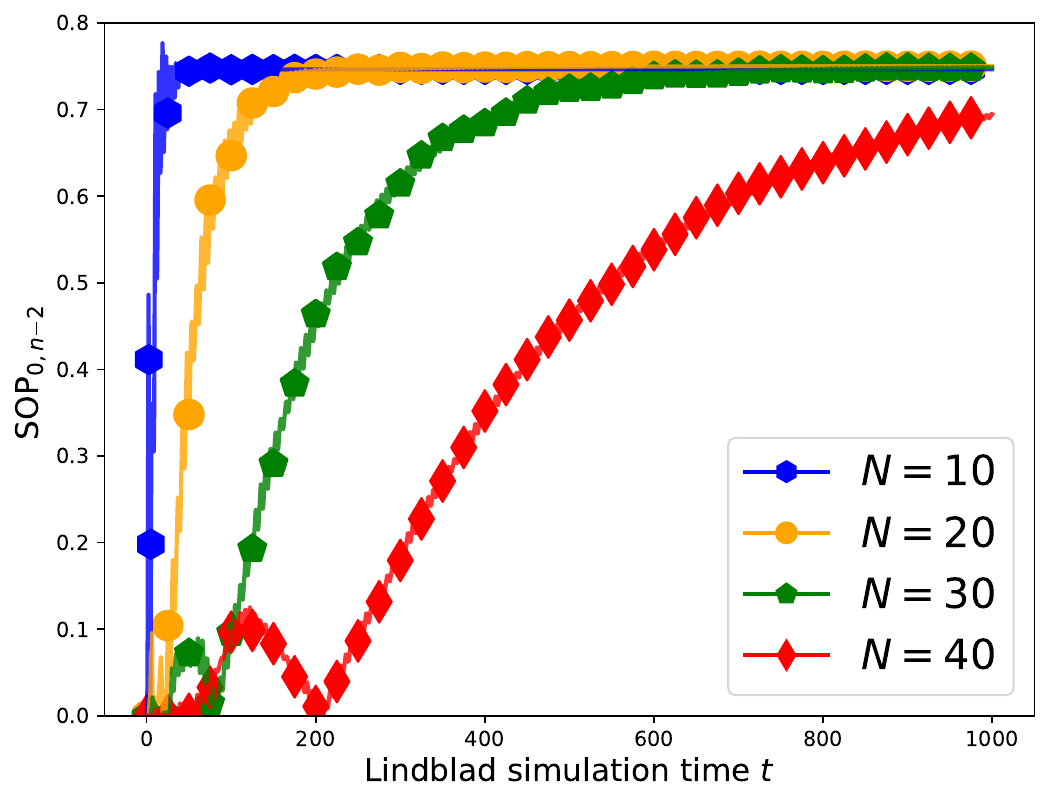}
        \caption{\centering}
    \end{subfigure}
  \caption{Evolution of the string order parameter (SOP) for the cluster state Hamiltonian ~\eqref{eqn:H_SPT} with boundary dissipation. The system size is $N=20$. (a) SOP comparison between the ground state (blue dashed line) and the final state at $T = 1500$ (red points). The final state accurately captures the quantum phase transition from the paramagnetic phase to the SPT phase.
  (b) Evolution of the SOP for several system sizes with $h_1/J = 0.5$, whose ground state is in the SPT phase. The initial state is the all-spin-down state, which is in the paramagnetic phase with $\text{SOP}=0$. The dissipative evolution consistently drives the system from the paramagnetic phase into the SPT phase.}
  \label{fig:SOP}
\end{figure}

Next, we examine how the convergence rate of Lindbladian dynamics scales with the system size. As in the study of the SOP, we set the coupling operators to be single Pauli $Y$ operators at the two ends of the system ($A_1=Y_1$ and $A_2=Y_N$) and fix $\Delta_{f}=0.1$ for all $N$. \revvv{The initial state $\rho_0$ is chosen to be the all-spin-down state.} In \cref{fig:SPT_energy}, we plot the energy decay of the SPT system with $h_1/J=0.4$. For $N=20, 30, 40$, similar to the 1D TFIM case, the energy rapidly approaches the ground state energy at early times before entering an asymptotic regime characterized by exponential decay. We note that for $N=10$ an energy plateau appears. This is because the Hamiltonian in \cref{eqn:H_SPT} has a fourfold degenerate ground state in the thermodynamic limit. For small system sizes, there is still a small energy gap between these nearly degenerate states, which is smaller than our chosen $\Delta_f$. Reducing $\Delta_{f}$ would further lower the plateau. In \cref{fig:SPT_energy}, we illustrate how the effective Liouvillian gap scales with system size. For various values of $h_1/J$, we consistently observe that the inverse effective gap scales as $\Delta_{\mathcal{L},\rm{eff}}^{-1} = \Theta(N^3)$. Consequently, to achieve a fixed accuracy $\eta$, the \rev{energy-based} mixing time scales as $\tau^E_{\operatorname{mix}} = \mathcal{O}(N^3)$.

\begin{figure}[ht]
  \centering
    \begin{subfigure}[b]{0.35\textwidth}
        \includegraphics[width=\textwidth]{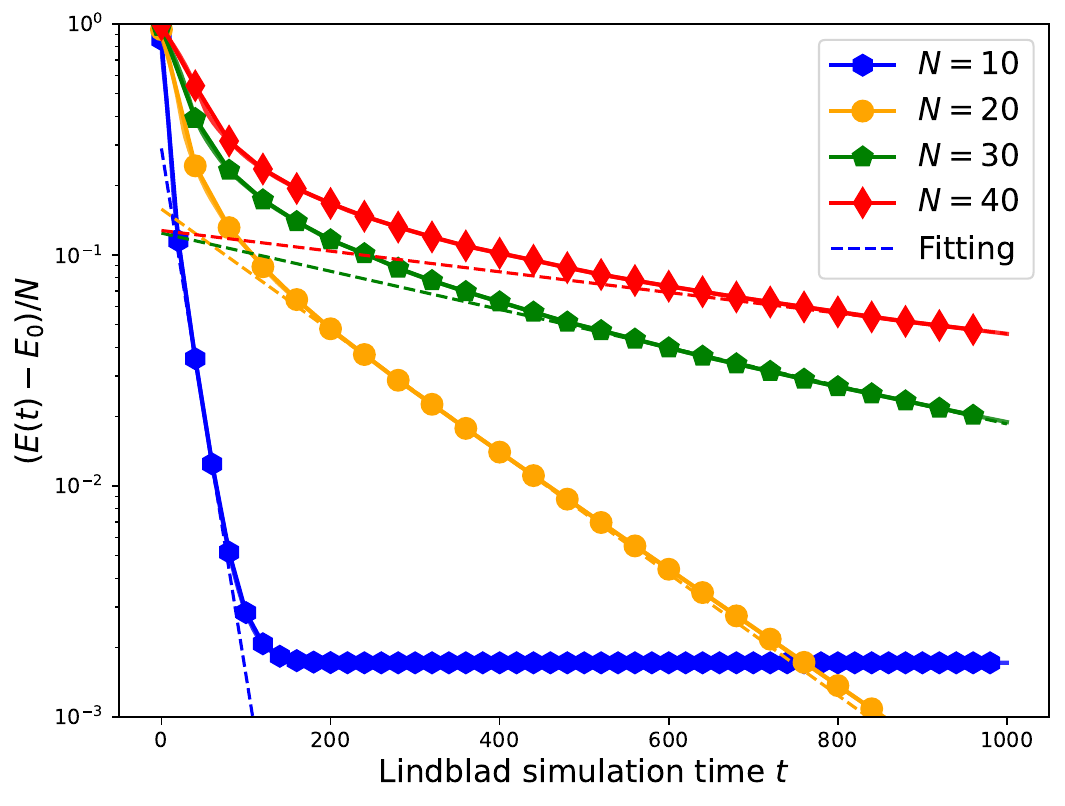}
        \caption{\centering}
    \end{subfigure}
    \begin{subfigure}[b]{0.35\textwidth}
        \includegraphics[width=\textwidth]{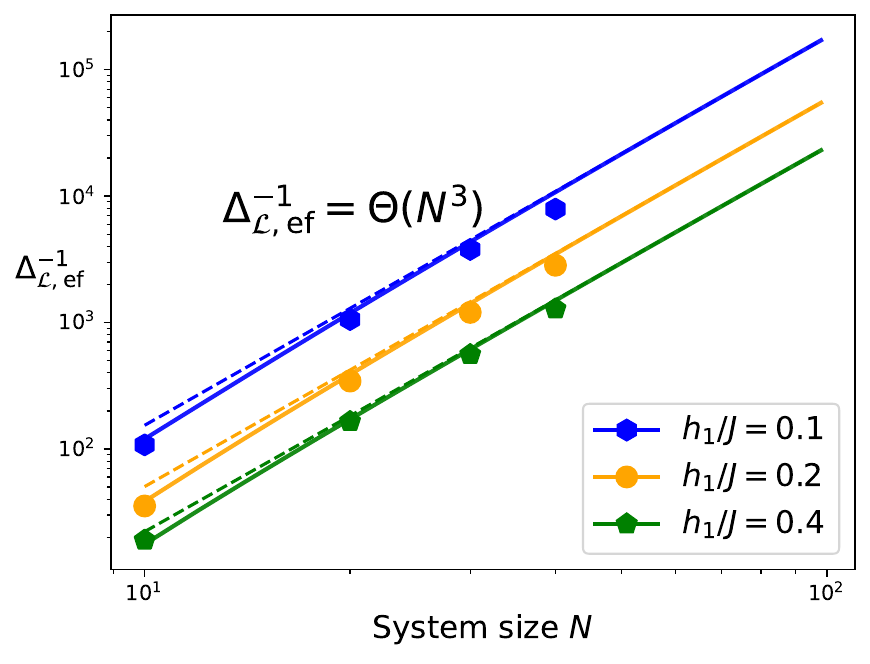}
        \caption{\centering}
    \end{subfigure}
  \caption{
  Numerical results for the cluster state Hamiltonian~\eqref{eqn:H_SPT} with boundary dissipation.
  (a) Convergence of energy starting from the maximally mixed state. Here, we use boundary dissipation and set $\Delta_f=0.1$, $h_1/J=0.4$. The dashed lines are exponential fits of the asymptotic behaviour of energy decay before hitting the energy plateau caused by the exponentially decaying edge modes. (b) The scaling of the inverse of the effective Liouvillian gap $\Delta^{-1}_{\mathcal{L},\rm{eff}}$ with respect to the system size $N$ with different $h_1/J$. Points are the fitting energy convergence rate $\kappa_N$ calculated from dynamics simulation.}
  \label{fig:SPT_energy}
\end{figure}

\subsection{Numerical results of tensor network simulation}\label{sec:numer_TNmethod}

\paragraph*{Rapid mixing of 1D gapped local Hamiltonians}

\begin{figure}
\begin{center}
\includegraphics[width=0.35\textwidth]{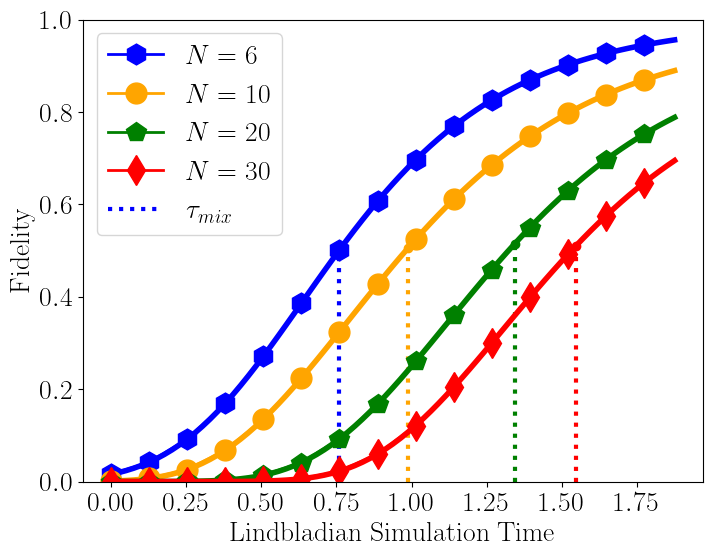}
\end{center}
\caption{ The mixing time, $\tau^F_{\mathrm{mix}}$, as defined in~\eqref{eqn:tau_F}. During the Lindbladian dynamics, the fidelity increases steadily and eventually converges to one. However, the convergence speed decreases as the system size increases.}
\label{fig:dynamics}
\end{figure}

\begin{figure}
\begin{center}
\includegraphics[width=0.35\textwidth]{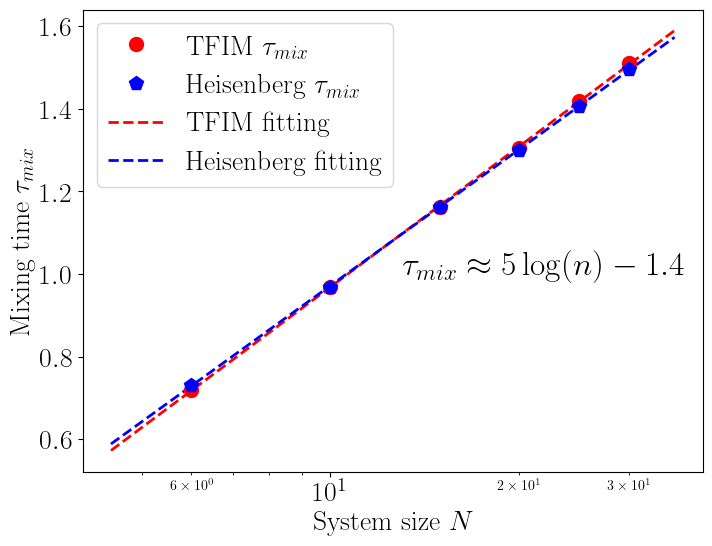}
\end{center}
\caption{Scaling of the mixing time for TFIM (red) and the Heisenberg model in a magnetic field (blue) under bulk dissipation, shown as a function of system size. The results indicate a logarithmic scaling, consistent with $\tau^F_{\mathrm{mix}} = \Theta(\log(N))$. }
\label{fig:log}
\end{figure}

As an application of our tensor network based algorithm, we prepare the ground state of two examples of 1D gapped, noncommuting local Hamiltonians using bulk dissipation, i.e., the coupling operators are chosen to be the Pauli operators $\{X_i,Y_i,Z_i\}$ on all sites.
The first one is the same TFIM example in \cref{eqn:H_TFIM} with $J=1, g=1.5$.
The second example, which cannot be transformed into a free fermionic system, is an anisotropic Heisenberg model in a magnetic field, described by the Hamiltonian \begin{equation}
    H=-J \sum_{i=1}^{N-1} X_i X_{i+1}- \xi \sum_{i=1}^{N-1} (Y_i Y_{i+1}+Z_i Z_{i+1})-g \sum_{i=1}^N Z_i.
\end{equation}
Here we choose the parameters $g=1.5, J=1$ and $\xi=0.1$.   We simulate their Lindblad dynamics for system sizes up to $N=30$. The bond dimension $D$ of the MPO representation for both the jump operators and the density matrix is set to 50. We validate this choice of the bond dimension in \cref{sec:additional_numer_tensor}.

In this section, we measure the mixing time with respect to the fidelity \rev{as in \cref{eqn:tau_F}}.
Unless otherwise mentioned, we choose $\eta=\frac12$ and start from the maximally mixed state $\rho_0=\mathbb{I}/2^N$. The fidelity increase during the Lindbladian dynamics of the
TFIM model is shown in \cref{fig:dynamics}. The results in \cref{fig:log} demonstrate that the mixing times in both cases scale logarithmically with the system size; this scaling is often referred to as ``rapid mixing'' \cite{BardetCapelGaoEtAl2023,KastoryanoTemme2013}.

\vspace{1em}
\paragraph*{TFIM with random transverse field---}

Now consider the 1D transverse field Ising model but with a random transverse field
\begin{equation}
H=-\sum_{i=1}^N g_i Z_i-J \sum_{i=1}^{N-1} X_i X_{i+1},
\end{equation}
where the strength of the transverse field $g_i\sim \mathcal{N}(2, \sigma^2)$ and $\sigma^2$ is the variance parameter.
Because of the Anderson localization, for any $\sigma>0$, the eigenfunctions of $H$ are exponentially localized in space. This means that choosing single Pauli operators on the boundary produces a large number of inaccessible ``dark modes'' (i.e., $\left\langle\psi_i|A_a| \psi_j\right\rangle\approx 0$), which means that boundary dissipation alone may lead to an exponentially long mixing time, or fail to converge to the ground state altogether. Nonetheless, the bulk dissipation is not subject to this failure mechanism due to Anderson localization. We choose the set of coupling operators $\{A_a\}$ to be all Pauli operators $\{X_i,Y_i,Z_i\}_{i=1}^N$.
We simulate the resulting Lindblad dynamics using the tensor network methods for system sizes up to $N=16$. \cref{fig:rand} illustrates the scaling behavior of the mixing time as a function of the system size $N$. Our results indicate a logarithmic scaling of the mixing time, suggesting that ground state preparation can be efficiently achieved using bulk dissipation.

\begin{figure}
\centering
\includegraphics[width=0.4\textwidth]{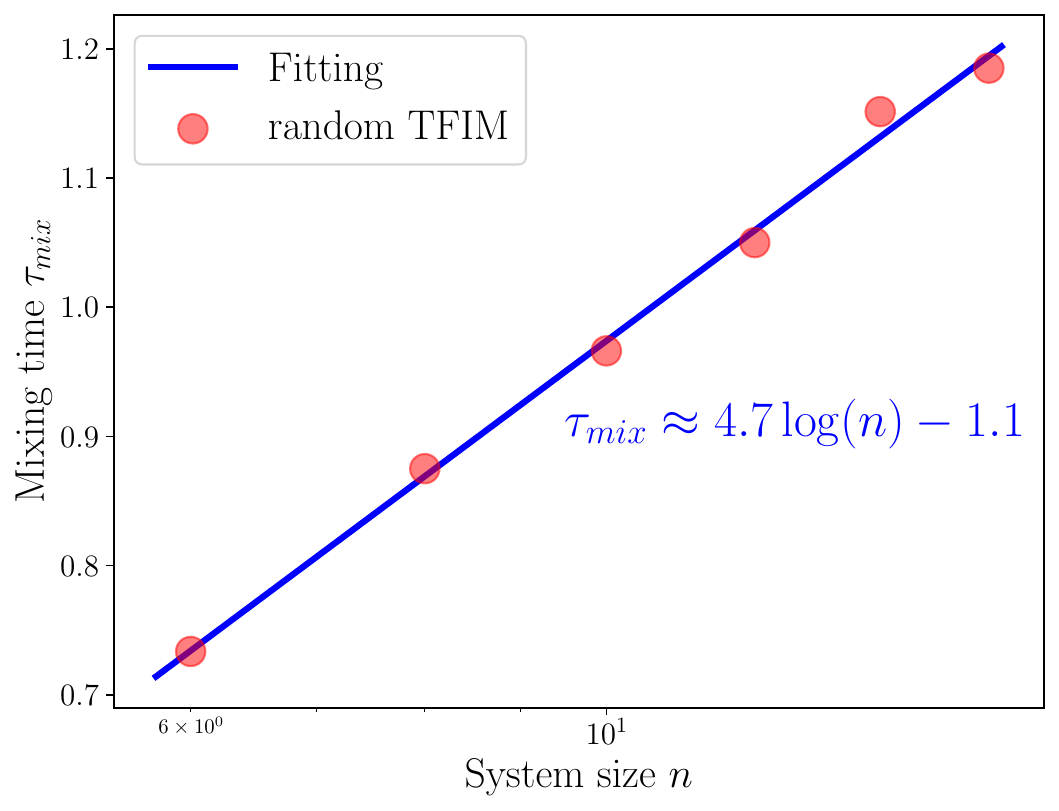}
\caption{\label{fig:rand} The scaling of the mixing time of random TFIM under bulk dissipation, i.e. we use $3N$ separate jump operators with the set of all Pauli operators on each site as coupling operators $\{X, Y, Z \}^N$. The Hamiltonian parameters are set to $J=1, \sigma^2=0.5$. We observe a logarithmic scaling of the mixing time.}
\end{figure}

\vspace{1em}
\paragraph*{Nonintegrable cluster state Hamiltonian---}

We now consider the following generalization of the cluster state Hamiltonian \begin{equation}\label{eq:NI-cluster-Ham}
    H=-\sum_{i=1}^{N-2} X_i Z_{i+1} X_{i+2}-h_1 \sum_{i=1}^N Z_i -h_2 \sum_{i=1}^{N-1} Z_i Z_{i+1}.
\end{equation}
When $h_2\ne 0$, the Hamiltonian cannot be transformed into a free fermionic system. Nonetheless, the ground state can still exhibit the SPT phase characterized by a nonzero string order parameter (SOP).

As an illustration, we choose the coupling operators to be single Pauli operators on the boundary $A_1=Y_1, A_2=Y_{\revvv{N}}$.  The parameters for $h_1=0.4, h_2=-0.4$ are set so that the ground state is in the SPT phase. We simulate the Lindblad dynamics starting from the maximally mixed state, which lies in the paramagnetic phase and exhibits a zero SOP value. As shown in \cref{fig:sop}, the energy converges to the ground state energy, and the SOP converges to approximately 1 during the evolution, indicating that the state transitions from paramagnetic phase to SPT phase. Our result shows that the Lindblad-based algorithm crosses the phase transition boundary during ground state preparation.

\begin{figure}
\centering
\includegraphics[width=0.4\textwidth]{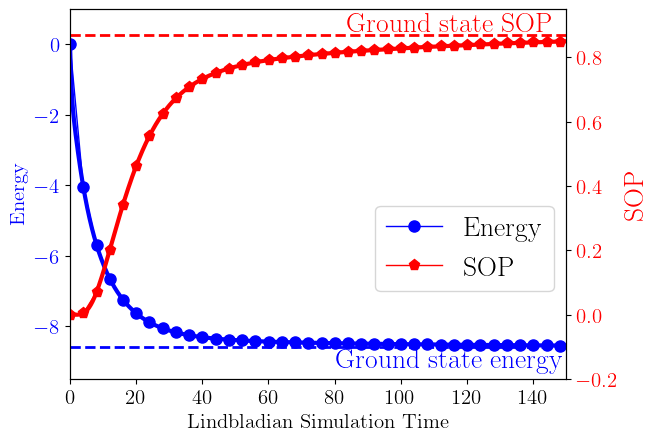}
\caption{\label{fig:sop}
Energy and string order parameter (SOP) during the evolution for the nonintegrable cluster state Hamiltonian \eqref{eq:NI-cluster-Ham} under boundary dissipation with system size $N=10$.
The plot shows the energy converging to the ground state energy and the SOP transitioning from the paramagnetic phase (SOP $= 0$) to the symmetry-protected topological (SPT) phase.}
\end{figure}

\subsection{\rev{Comparison with adiabatic state preparation protocols}}\label{sec:compare_adiabatic}

A natural question is whether dissipative state preparation protocols offer advantages over alternative approaches, such as adiabatic state preparation. Since both adiabatic state preparation~\cite{AlbashLidar2018} and dissipative state preparation~\cite{VerstraeteWolfIgnacioCirac2009} are known to be \BQP-complete, meaningful comparisons arise only within specific problem settings, where one can exploit structural properties of the target system and evaluate the relative performance of each method under those conditions.

In a typical protocol for adiabatic ground-state preparation, one considers a time-dependent Hamiltonian that interpolates between an initial noninteracting Hamiltonian whose ground is easily prepared and a final target Hamiltonian whose ground state is desired. This method can be effective if one can identify a gapped adiabatic path that connects the target ground state to the trivial ground state.  However, it is likely to fail if the path crosses a {first-order} phase transition where the gap vanishes and the quantum state must vary sharply to remain close to the instantaneous ground state along the path.


To compare the performance of dissipative ground-state preparation and adiabatic ground-state preparation, we consider as an example the  1D axial next-nearest-neighbor Ising (ANNNI) model \cite{Selke1988}:
\begin{equation}
H_{\mathrm{ANNNI}}
=\frac{J_1}{4} \sum_{i} Z_i Z_{i+1}
+\frac{J_2}{4} \sum_{i} Z_i Z_{i+2}
-\frac{\Gamma}{2} \sum_{i} X_i,
\label{eq:annni}
\end{equation}
where $J_1$ is the nearest-neighbor coupling, $J_2$ is the frustrating next-nearest-neighbor coupling,
and $\Gamma$ is the strength of the transverse field.

In the ANNNI model, there are three competing tendencies: (i) for $J_1>0$, the nearest-neighbor Ising term 
favors an antiferromagnetic N\'eel state such as $\ket{\uparrow\downarrow\uparrow\downarrow\cdots}$; (ii) for $J_2>0$, the next-nearest-neighbor term 
favors a period-4 
modulated structure such as $\ket{\uparrow\uparrow\downarrow\downarrow\cdots}$; and (iii) the transverse field 
provides quantum fluctuations and favors a state with spins polarized in the $x$ direction.

For adiabatic state preparation, we may use a linear interpolation between \revvv{a} simple initial Hamiltonian $H_{\mathrm{init}} = - \frac{h_0}{2} \sum_{i} Z_i$ and the ANNNI target s.t.
\begin{equation}
H(s(t)) = \bigl(1 - s(t)\bigr)\,H_{\mathrm{init}} + s(t)\,H_{\mathrm{ANNNI}},
\label{eq:asp_path}
\end{equation}
with a monotonic schedule satisfying $s(0)=0$, and $s(T)=1$ (for instance, $s(t)=t/T, t\in[0,T]$). The adiabatic evolution from the initial ground state $|\psi(0)\rangle$ is given by
\begin{equation}
i\,\partial_t |\psi(t)\rangle = H(s(t))\,|\psi(t)\rangle, \quad 0 \le t \le T.
\end{equation}
%
The initial state is polarized 
in the $z$ direction. As $s$ increases, this $z$-polarized phase competes with the ordered phase selected by $(J_1,J_2)$.  When $\Gamma=0$, along the adiabatic path from $H_{\mathrm{init}}$ to $H_{\mathrm{ANNNI}}$, a first-order transition is encountered at a point $s=s_c$ where the energy of the $z$-polarized state becomes equal to that of the competing ordered phase and the gap vanishes.  For small $\Gamma$, the location of the transition shifts, and the level crossing opens into an avoided crossing with a small gap.

We simulated this protocol for $J_1=2$, $J_2=0.6$, $\Gamma=0.2$, $h_0=1.0$, with periodic boundary conditions and lattice size $L=12$. For each instantaneous Hamiltonian $H(s)$ along the adiabatic path, we define the ground state manifold as the set of orthogonal eigenstates within $10^{-4}$ energy of the lowest eigenvalue. As shown in the inset of Fig.~\ref{fig:asp_annni}, the ground state manifold dimension varies non-monotonically along the path: starting from a unique ground state, it reaches a maximum of 3 around $s=0.4$, decreases to 1 around $s=0.5$, and stabilizes to 2 for $s>0.6$. We also define the effective gap as the energy difference between the lowest excited state above the ground state manifold and the ground state energy; notably, this gap nearly closes three times along the adiabatic path.

To assess the quality of state preparation, we monitor the overlap with the ground-state subspace of the target Hamiltonian $H_{\rm ANNNI}$, quantified by $\Tr[\rho(t)\sigma]$, where $\sigma = \ketbra{\psi_0}{\psi_0} + \ketbra{\psi_1}{\psi_1}$ \revv{denotes the (unnormalized) projector onto the two-fold degenerate ground-state manifold. This choice reflects the fact that the goal of the algorithm is to prepare the correct subspace rather than a specific pure state within it, which is a natural measure from an algorithmic perspective for degenerate systems.} In addition to fidelity, we track the convergence of the order parameters
\begin{equation}
m_1 = \frac{1}{4L}\sum_{i}\langle Z_i Z_{i+1}\rangle, \quad m_2 = \frac{1}{4L}\sum_{i}\langle Z_i Z_{i+2}\rangle
\end{equation}
to their target values averaged over the $H_{\rm ANNNI}$ ground-state manifold. \revv{Since all ground states within the manifold yield identical values of $m_1$ and $m_2$, these observables serve as stable and physically meaningful indicators of successful ground-state subspace preparation.}

For a total evolution time $T=1000$, we plot in \cref{fig:asp_annni} the overlap between the time-evolved state $\rho(t)=\ketbra{\psi(t)}{\psi(t)}$ and the ground state manifold of $H_{\rm ANNNI}$. This ground manifold overlap exhibits persistent oscillations and never achieves a high value, indicating that the adiabatic protocol fails to prepare the ground state. The initial state is $z$-polarized, and as the system evolves, it does not reach the correct ordered phase. This poor performance is due to the presence of multiple level crossings and small gaps along the adiabatic path, which induce diabatic transitions and prevent the system from remaining in the ground state manifold.

\begin{figure}[h]
\begin{center}
\includegraphics[width=0.45\textwidth]{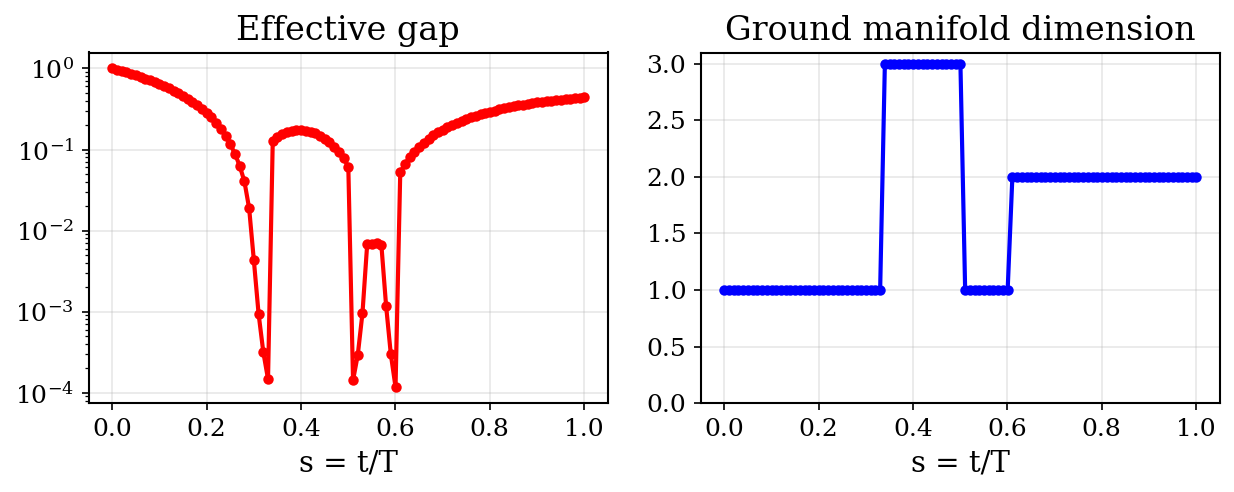}

\includegraphics[width=0.45\textwidth]{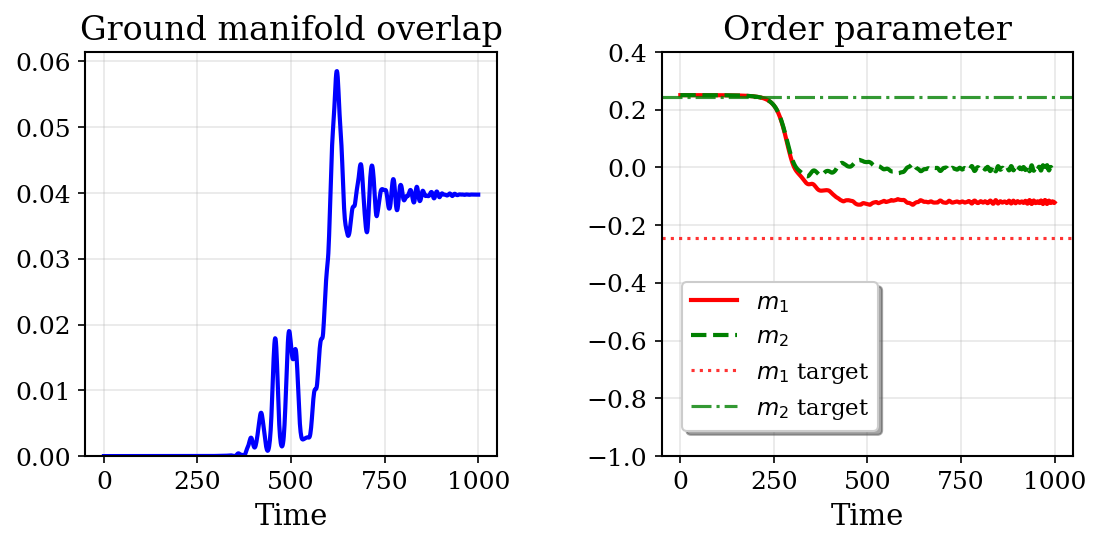}
\end{center}
\caption{Adiabatic state preparation for the ANNNI model at $L=12$. Top: Effective gap, and the ground manifold dimension of the instantaneous Hamiltonian $H(s)$ along the adiabatic path. Bottom: Evolution of the overlap between $\rho(t)$ and the ground manifold and the order parameters.}
\label{fig:asp_annni}
\end{figure}

In contrast, as shown in \cref{fig:dsp_annni}, dissipative state preparation with coupling operators $\mc{A}=\{X_i, Z_i\}_{i=1}^L$ successfully prepares a high-fidelity approximation to the ground state manifold. We set the spectral gap parameter $\Delta=0.2$ when constructing the jump operators, and this protocol does not distinguish between the two degenerate ground states of $H_{\rm ANNNI}$. The initial state is chosen to be the all-one state, the same as that used in adiabatic state preparation. Under dissipative state preparation, the observables $m_1$ and $m_2$ converge rapidly and monotonically to their target values, which are unaffected by the complex energy landscape or gap closings encountered along the adiabatic path.

\begin{figure}[h]
\begin{center}
\includegraphics[width=0.45\textwidth]{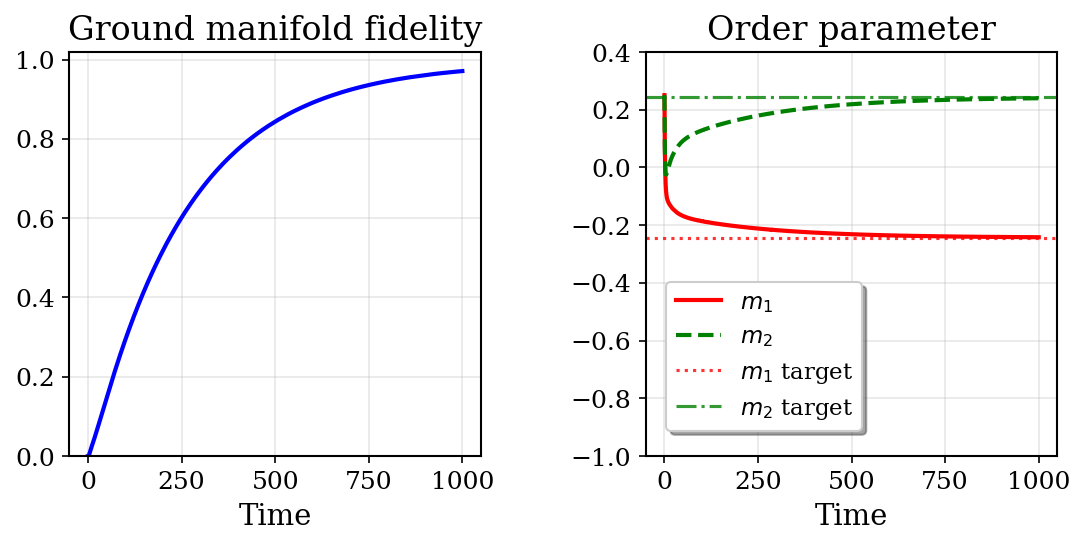}
\end{center}
\caption{Adiabatic state preparation for the ANNNI model at $L=12$. Evolution of the overlap between $\rho(t)$ and the ground manifold and the order parameters.}
\label{fig:dsp_annni}
\end{figure}

The results presented here do not imply that adiabatic state preparation cannot be adjusted to achieve ground state preparation. For example, we could introduce a carefully designed adiabatic path to avoid any gap closing. However, this approach may require detailed knowledge of the system, and it can be challenging to determine in advance whether a specific adiabatic path will be sufficient.  On the other hand, the design of dissipative protocols can be more agnostic to the specifics of the target Hamiltonian or ground state. Ultimately, the comparison between these methods will likely be system dependent, and more theoretical and numerical investigations will be needed in the future.

\section{Theoretical Results}
\subsection{Mixing time of quasi-free systems}
\label{sec:mixingtime_quasifree}

\rev{In this section, we establish a theory that provides an explicit bound on the convergence in trace distance to the ground state for general quasi-free dynamics.} When applied to the TFIM model~\rev{studied in~\cref{sec:numer_quasifree}}, we find that the convergence rate estimate matches with the numerical observation.


We provide some high-level ideas of our strategy below.
To prove the convergence in trace distance, we use the Fuchs--van de Graaf inequality \rev{(see Theorem \ref{thm:fuchs})}.
 The ground state of any Hamiltonian that is quadratic in Majorana operators can be written as a quasivacuum state $\sigma = \ketbra{\mathrm{vac}}{\mathrm{vac}}$, with
\begin{equation}
b_k \ket{\mathrm{vac}} = 0, \quad k = 1, \ldots, N,
\end{equation}
for a properly defined set of fermionic annihilation operators $\{b_k\}$ (\cref{sec:prelim_majorana}). Let $\hat{N} = \sum_k b_k^{\dag} b_k$ be the total number operator. Note that all states other than $\ket{\mathrm{vac}}$ \revvv{have} at least one particle. This gives the inequality
\begin{equation}
1 - \braket{\mathrm{vac}|\rho|\mathrm{vac}} =\Tr[\rho(I-\ketbra{\mathrm{vac}}{\mathrm{vac}})]\leq \Tr[\rho \hat{N} ].
\end{equation}
\rev{If we could prove an inequality of the form $\mathcal{L}^{\dagger}(\hat{N}) \preceq -c \hat{N}$ for some constant $c>0$, then it would immediately follow that $\Tr[\hat{N}\rho(t)] \leq \Tr[\hat{N}\rho(0)] e^{-ct}$, completing the proof. However, in many cases including boundary dissipation, such an inequality does not hold for any $c>0$. }
Our key innovation is to find a positive definite observable $O$, which is equivalent to the number operator in the sense that there exist constants $C_1, C_2$ such that $C_1 \hat{N} \preceq O \preceq  C_2 \hat{N}$. Then, for a proper choice of $O$, we prove \rev{the desired inequality $\mc{L}^{\dag}(O)\preceq -2\Delta O$ for some constant $\Delta>0$ and thus the exponential convergence property:}
\begin{equation}
\Tr[O \rho(t)] \le  \Tr[O \rho(0)] e^{-2\Delta t}.
\end{equation}
\rev{This relation in turn gives
\begin{equation}
\Tr[\hat{N} \rho(t)]\le \frac{C_2}{C_1}\Tr[\hat{N} \rho(0)] e^{-2\Delta t}.
\end{equation}
The key point is that the ratio $\frac{C_2}{C_1} > 1$ and contributes only a logarithmic additive term to the mixing time.}
When substituted into the Fuchs--van de Graaf inequality, this yields the desired exponential convergence \rev{in trace distance} with the explicit convergence rate $\Delta$.
It is worth noting that \rev{$\Tr[O \rho]$ may be viewed as a Lyapunov function of the Lindblad dynamics.}

We find that the convergence rate $\Delta$, as well as the constants $C_1, C_2$, are determined by a non-Hermitian quadratic Hamiltonian
\begin{equation}\label{eqn:nonhermitian_ham}
H_{\rm nh} = iH - \frac{1}{2} \sum_{a} K_a^{\dag} K_a =\sum_{p,q=1}^{2N} (h_{\mathrm{nh}})_{pq} w_p w_q,
\end{equation}
where $h_{\mathrm{nh}}$ is a non-Hermitian matrix in general. Assume $h_{\mathrm{nh}}$ is diagonalizable as $h_{\mathrm{nh}}=V D V^{-1}$, then $\Delta$ is given by the non-Hermitian gap \revvv{$(-\max _i \operatorname{Re} D_{i i})$}. For a quadratic observable $O$ in Majorana operators, we find that  $\mc{L}^{\dag}(O)$ is entirely determined by the non-Hermitian coefficient matrix $h_{\mathrm{nh}}$ (see \cref{eqn:Ldag_O_quasifree}).


We now present our main theorem for quasi-free systems in Theorem \ref{thm:sharpbound_quasifree}, with the proof provided in \cref{sec:proof_sharpbound_quasifree}.

\begin{thm}
Let $H$ be a gapped quadratic Majorana Hamiltonian with $2N$ modes,  $\{A_a\}$ be a set of coupling operators that are linear in Majorana operators, and $K_a$ be the corresponding jump operators defined via \cref{eqn:jump_time}. We consider the non-Hermitian Hamiltonian in \cref{eqn:nonhermitian_ham}, assume the coefficient matrix $h_{\mathrm{nh}}$ is diagonalizable with $h_{\mathrm{nh}}=V D V^{-1}$, and define the non-Hermitian gap $\Delta=-\max _i \operatorname{Re} D_{i i}$. We denote the condition number of $V$ by $\kappa(V)$.

If $\Delta>0, \kappa(V)<\infty$, then starting from any initial state $\rho_0$, the Lindblad dynamics \eqref{eq:lindblad} converges exponentially in trace distance to the quasivacuum state $\sigma = \ketbra{\mathrm{vac}}{\mathrm{vac}}$ with
\begin{equation}
D(\rho(t),\sigma)\leq \kappa(V) \sqrt{N} e^{-\Delta t}.
\end{equation}

\label{thm:sharpbound_quasifree}
\end{thm}

An immediate result from Theorem \ref{thm:sharpbound_quasifree} is that the mixing time defined in the trace distance in \cref{eqn:mixing_tracedistance} scales as
\begin{equation}
\tau_{\mathrm{mix}}(\eta) \leq \Delta^{-1} \log \left(\frac{\kappa(V) \sqrt{N}}{\eta}\right).
\end{equation}
Therefore as long as $\kappa(V)=\poly(N)$, the scaling of the mixing time is determined by the scaling of the non-Hermitian gap $\Delta$ with respect to $N$, up to a logarithmic factor.

\vspace{1em}
\paragraph*{Application to TFIM with boundary dissipation---}

The $\Or(N^3)$ scaling for boundary-dissipated 1D translationally invariant TFIM has been observed in previous studies \cite{ZhengWangChen2023}, where the jump operator is strictly applied to a single site on the boundary. Their proof maps the problem to a non-Hermitian Su-Schrieffer-Heeger (SSH) model, which enables an analytic computation of the rapidity spectrum. In contrast, our jump operators ${K_a}$ are quasilocal, rendering this technique inapplicable. Instead, we leverage the stronger result in Theorem \ref{thm:sharpbound_quasifree} to directly bound the convergence in trace distance.

For simplicity we only consider the case when the coupling operators are $X_1,Y_1$ on one end of the boundary.
First, following the proof of Theorem \ref{thm:sharpbound_quasifree} in \cref{sec:proof_sharpbound_quasifree}, the jump operators take the form
\begin{equation}
\begin{aligned}
K_{X_1}&=\int f(s) e^{i H s} X_1 e^{-i H s} \ud s=\sum_k \varphi_{k1} b_k,\\
K_{Y_1}&=\int f(s) e^{i H s} Y_1 e^{-i H s} \ud s=\sum_k \psi_{k1} b_k,
\end{aligned}
\end{equation}
for some coefficient vectors $\{\varphi_{k1}\},\{\psi_{k1}\}$, where the ground state is a quasi-\revvv{vacuum} state satisfying $b_k \ket{\mathrm{vac}}=0$. Let $\Lambda$ be a diagonal matrix encoding the eigenvalues of the TFIM Hamiltonian $H$, then the non-Hermitian Hamiltonian in \cref{eqn:nonhermitian_ham} can be written as
\begin{equation}
H_{\rm nh} = \vb^{\dag} h_{\mathrm{nh}}^f \vb, \quad
h^{f}_{\mathrm{nh}}=i\Lambda-\frac12 \varphi\varphi^{\dag}-\frac12 \psi \psi^{\dag}.
\end{equation}
The crucial role of the coherent term in convergence is now evident. Without the coherent $i\Lambda$, $h^{f}_{\mathrm{nh}}$ is merely a rank-2 matrix, resulting in a large kernel for $H_{\rm nh}$ and the non-Hermitian gap $\Delta = 0$. When the coherent term is present, the jump operators can shift the imaginary eigenvalues $i\Lambda$ away from the real axis, opening a positive spectral gap.

Using first-order perturbation theory, we estimate the spectral gap as $\Delta = \Theta(N^{-3})$. The cubic scaling is mainly due to long-wavelength modes, whose magnitude scales as $\Or(N^{-1.5})$ near the boundary. The square of this magnitude determines the spectral gap from the real axis, perturbing the eigenvalues from the imaginary axis by an amount proportional to $N^{-3}$.

Furthermore, $\kappa(V) = \Or(1)$, which gives the mixing time scaling as $\Or(N^3 \log N)$.
We provide the details in \cref{sec:tfim_boundary_proof}. We expect the analysis for the cluster state Hamiltonian may be derived with a similar argument and is omitted here.

\subsection{Rapid ground state preparation of weakly interacting spin systems}\label{sec:provable_rapid}

Our numerical results in \cref{sec:numer_TNmethod} strongly suggest that dissipative dynamics can achieve rapid mixing, or $\mathcal{O}(\log N)$ mixing time, for certain noncommuting Hamiltonians under bulk dissipation. However, as discussed in \cref{sec:intro}, despite significant progress in theoretical understanding of the effectiveness of finite-temperature quantum Gibbs samplers~\cite{TemmeKastoryanoRuskaiEtAl2010,KastoryanoTemme2013,BardetCapelGaoEtAl2023,rouz2024,DingLiLinZhang2024,kochanowski2024rapid,rouze2024optimal,tong2024fast}, many of these techniques cannot be used to characterize the convergence towards the ground state because it is not a \revvv{full-rank} state.
In this section, we develop a technique that provides the first rapid mixing result for a class of general noncommuting Hamiltonians.

We focus on the Lindbladian dynamics without coherent terms, i.e.
\begin{equation}\label{eq:lindblad_no_coherent}
  \frac{\mathrm{d} \rho}{\mathrm{d}t}=\mathcal{L}[\rho]=\sum_a \underbrace{K_a \rho K_a^{\dagger}-\frac{1}{2}\left\{K_a^{\dagger} K_a, \rho\right\}}_{:=\mathcal{L}_{a,\varepsilon}}.
\end{equation}
\rev{Although the coherent term is removed, the fixed point of the above Lindblad dynamics is still the ground state of the Hamiltonian $H$, since the jump operator still depends on the Hamiltonian.} For concreteness, consider a local Hamiltonian $H$ over a $D$-dimensional lattice of spin systems $\Lambda= [0,L]^D$ with the following form (the system size is $N=(L+1)^D$):

\begin{equation}\label{eqn:H}
H=H_0+H_1=-\sum_{i} Z_i+ \varepsilon \sum_{j} h_j, \quad \|h_j\|\leq 1.
\end{equation}
Here, $H_0 = -\sum_{i} Z_i$ is referred to as the noninteracting term because its indices do not overlap. The choice of $Z_i$ as the noninteracting term is made for convenience, and can be substituted with other simple, gapped local terms that also have non-overlapping indices. We assume the interacting term $H_1$ is an $(r_0,l)$-geometrically local Hamiltonian (see \cref{sec:notation}). A specific example of the Hamiltonian in \cref{eqn:H} is the $D$-dimensional TFIM model, which is a $(2,2D+1)$-local Hamiltonian. The parameter $\varepsilon$ is called the interaction strength.


For the Hamiltonian \eqref{eqn:H}, it is sufficient to choose $\{A_a\}=\{X_i\}_{i\in \Lambda}$ to be the set of all single Pauli $X$ matrices as coupling operators. This is because if the interaction strength $\varepsilon=0$, then the dissipative dynamics for the noninteracting problem is ergodic and the mixing time scales as  $\Or(\log N)$. Our main result is that there exists a critical interaction strength $\varepsilon^*$ independent of $L$ (and thus $N$), so that for all $\varepsilon<\varepsilon^*$, the scaling of the mixing time remains $\mathcal{O}(\log N)$.

\begin{thm}[Informal]
\label{thm:rapid_mixing_2D_TFIM}
Consider a gapped Hamiltonian $H$ in the form of \eqref{eqn:H} defined on a $D$-dimensional lattice $\Lambda= [0,L]^D$, and $N=(L+1)^D$ is the system size.  Let $\{A_a\}=\{X_i\}_{i\in\Lambda}$ be a set of coupling operators and $\{K_a\}$ be the corresponding jump operators defined via \cref{eqn:jump_time}. Consider the Lindblad operator without the coherent term~\eqref{eq:lindblad_no_coherent}. Then there exists a constant $\varepsilon^*$ independent of the system size such that when $\varepsilon<\varepsilon^*$, we have
\begin{equation}
\tau_{\operatorname{mix}}(\eta)=\Theta(\log(N/\eta))\,,
\end{equation}
where $\tau_{\operatorname{mix}}(\eta)$ is defined in~\eqref{eqn:mixing_tracedistance_rho_0}.
\end{thm}

The rigorous statement of Theorem \ref{thm:rapid_mixing_2D_TFIM} and its proof are given in \cref{proof_rapid_thm}.~\rev{Our proof is inspired by recent analyses of mixing times for quantum Gibbs samplers~\cite{rouze2024optimal}. The analysis in~\cite{rouze2024optimal} avoids relying on the invertibility of the fixed point and characterizes convergence through the decay of the so-called  ``local oscillator norm", a quantity that can be defined for any Lindblad dynamics with a unique fixed point. In our setting, we employ a modified local oscillator norm of observables (see~\cref{eqn:oscillator} in~\cref{proof_rapid_thm}). For the noninteracting Hamiltonian $H_0$, the exponential decay rate of the oscillator norm can be computed explicitly. Furthermore, both its evolution and its decay rate remain stable under local perturbations of the Hamiltonian, which can be rigorously established using the Lieb-Robinson bound.}


\subsection{Rapid ground state preparation of weakly interacting fermionic systems}\label{sec:provable_rapid_fermion}

In Theorem \ref{thm:rapid_mixing_2D_TFIM}, the noninteracting Hamiltonian is defined as a sum of single-site Pauli $Z$ operators. In this section, we extend this result to the fermionic setting. Since free fermionic systems can be exactly diagonalized, we introduce a more general noninteracting term that permits coupling between fermionic sites.

We consider a local fermionic Hamiltonian $H$ defined on a $D$-dimensional lattice of fermionic systems, $\Lambda = [0,L]^D$, given by
\begin{equation}\label{eqn:H_fermion}
H = H_0+ H_1= \sum_{ij} M_{i,j} c^\dagger_i c_j + \varepsilon \sum_j h_j, \quad \|h_j\|\leq 1,
\end{equation}
where $(M_{i,j})$ is a positive definite Hermitian matrix, and $c^\dagger_j$ and $c_j$ are the creation and annihilation operators at site $j$. The terms $\{h_j\}$ are local fermionic perturbations and are parity preserving, meaning that each $h_j$ contains an even number of creation and annihilation operators. We further assume that $H_0$ is $(1,l)$-geometrically local and $\sum_j h_j$ are $(r_0,l)$-geometrically local. Specifically, each term in $H_0$ is a product of fermionic operators acting on a set of sites whose Manhattan diameter is at most $1$, and each $h_j$ is a product of fermionic operators acting on a set of sites whose Manhattan diameter is at most $r_0$. In addition, each site $i$ appears in at most $l$ non-trivial $c^\dagger_i c_j$ and $h_j$ terms.

For~\cref{eqn:H_fermion}, we choose $\{A_a\}=\{c^\dagger_i,c_i\}_{i\in \Lambda}$ to be the set of all single fermionic operators as coupling operators. We show that the mixing time of the Lindblad dynamics for the fermionic system~\eqref{eqn:H_fermion} also scales logarithmically with the system size for sufficiently small $\varepsilon$. This is summarized in the following theorem:
\begin{thm}[Informal]
\label{thm:rapid_mixing_fermion}
Consider a gapped fermionic Hamiltonian $H$ in the form of \eqref{eqn:H_fermion} defined on a $D$-dimensional lattice $\Lambda= [0,L]^D$, and $N=(L+1)^D$ is the system size. Let $\{A_a\}=\{c^\dagger_i,c_i\}_{i\in\Lambda}$ be a set of coupling operators and $\{K_a\}$ be the corresponding jump operators defined via \cref{eqn:jump_time}. Consider the Lindblad operator without the coherent term~\eqref{eq:lindblad_no_coherent}. Then there exists a constant $\varepsilon^*$ independent of the system size  such that when $\varepsilon<\varepsilon^*$, we have
\begin{equation}
\tau_{\operatorname{mix}}(\eta)=\Theta(\log(N/\eta))\,,
\end{equation}
where $\tau_{\operatorname{mix}}(\eta)$ is defined in~\eqref{eqn:mixing_tracedistance}.
\end{thm}

\rev{Theorem~\ref{thm:rapid_mixing_fermion_rigorous}, the rigorous version of Theorem~\ref{thm:rapid_mixing_fermion}, is proven in Appendix~\ref{proof_rapid_thm_fermion}, where further technical details are provided.} \rev{Similar to Theorem~\ref{thm:rapid_mixing_2D_TFIM}, the proof of the above theorem uses the oscillator norm of observables. On the other hand, we note that the creation and annihilation operators are nonlocal in the spin basis after applying the Jordan--Wigner transformation. As a result, the oscillator norm defined in \eqref{eqn:oscillator} is not suitable for fermionic systems. A proper definition of the oscillator norm requires the notion of the fermionic partial trace that is compatible with the canonical anticommutation relation (CAR); see~\cref{proof_rapid_thm_fermion}~\eqref{eqn:fermion_partial_trace_simplified} and~\eqref{eqn:fermion_partial_trace} for details.}

\rev{Additionally, we note that although the perturbation terms in both Theorems~\ref{thm:rapid_mixing_2D_TFIM} and~\ref{thm:rapid_mixing_fermion} are assumed to be local for simplicity, the resulting perturbation in the jump operator $K$ is not strictly local. Instead, under a suitable choice of $f(t)$ (decaying rapidly in $|t|$), the perturbation in $K$ becomes quasilocal, which is sufficient to establish the stability of the evolution of the local oscillator norm. Consequently, our results extend naturally to quasilocal perturbations, as these also induce quasilocal perturbations to the jump operator $K$.}

\section{Discussion}

This work significantly strengthens the evidence that dissipative dynamics is a powerful method for preparing ground states for a wide class of noncommuting Hamiltonians. A variety of dissipative mechanisms exist, such as imaginary-time evolution (ITE) \cite{McArdleJonesEndoEtAl2019,MottaSunTanEtAl2020,HugginsOGormanBryanEtAl2022}.
However, implementing ITE via the operator $e^{-\tau H}$ does not readily yield a completely positive trace-preserving (CPTP) map. Existing approaches often rely on variational ansatze or tomography-based procedures. By contrast, Lindblad dynamics provides a nonvariational and inherently CPTP process that can be efficiently implemented on fault-tolerant quantum hardware.

We have shown in Theorem~\ref{thm:sharpbound_quasifree} that a carefully designed Lindblad dynamics succeeds in preparing the ground states of quadratic Majorana Hamiltonians, and have proven in Theorem~\ref{thm:rapid_mixing_2D_TFIM} its effectiveness for nonintegrable but weakly interacting spin Hamiltonians.
Our tensor-network simulations suggest that these methods can remain effective beyond the reach of our rigorous analysis. For fermionic systems, our Theorem~\ref{thm:rapid_mixing_fermion} generalizes the recent work of~\cite{tong2024fast,vsmid2025polynomial}, which establishes the spectral gap for Gibbs state preparation in perturbed fermionic systems. Our result extends this to the ground state (zero temperature). Our result enhances the spectral gap bound (also called fast mixing) at finite temperatures from Ref.~\cite{tong2024fast,vsmid2025polynomial} to the stronger notion of rapid mixing, and proves rapid mixing at zero temperature. We note that Theorem~\ref{thm:rapid_mixing_fermion} imposes certain restrictions on the choice of the noninteracting term. Removing these restrictions and extending our result to efficient low-temperature thermal state preparation remain interesting directions for future work.
Additionally, investigating spin systems with long-range interactions may provide further insights into the mixing properties of dissipative processes.

We also note that Lindblad dynamics with jump operators of the form \eqref{eqn:jump_time} are closely related to cooling and thermalization protocols based on weakly coupled system-bath interactions~\cite{MiMichailidisShabaniEtAl2024,LloydMichailidisMiEtAl2025,Meghana_2020,langbehn2024dilute}, including several that appeared after the initial submission of the present work~\cite{molpeceres2025,hahn2025provably,hagan2025thermodynamic,langbehn2025universal,lloyd2025quantumthermal,scandi2025thermal}. Theoretical justification of the end-to-end efficiency of such protocols requires mixing time analysis (see \cref{sec:Lindalgo}). Our work provides the first rigorous mixing time analysis for a number of physical, noncommuting Hamiltonians, and we expect that these results will inform the future development of dissipative ground-state preparation protocols.

Recent numerical results also show that dissipative state preparation can be more robust to decoherence than adiabatic state preparation~\cite{doi:10.1126/sciadv.aaw9268}. A similar phenomenon has been observed experimentally showing that the lifetime of dissipatively prepared states can be much longer than the coherence times of physical qubits~\cite{mi2024stable,DelReRostFoss-FeigEtAl2024}. A rigorous understanding of the source of such robustness would be an interesting direction for future research.

A key open question is whether quantum computers simulating dissipative dynamics can tackle classically hard ground-state problems.  Viewing ground-state preparation as a minimization problem, there exist instances where finding even a local minimum is classically hard, yet Lindblad dynamics can efficiently achieve this quantumly~\cite{ChenHuangPreskillEtAl2024}. In the case of~\cite{ChenHuangPreskillEtAl2024}, the local and global minima of the energy coincide, resulting in a single-phase ground state. Many challenging physical Hamiltonians involve resolving multiple phases with nearly degenerate energy levels, which typically lie outside the perturbative regime. A prominent example is the phase diagram of the two-dimensional Hubbard model. \revv{It would be instructive to compare dissipative algorithms with a broad class of classical approaches such as
those based on variational states or quantum Monte Carlo. While a detailed benchmarking study is beyond our current scope, we emphasize that dissipative dynamics can, in principle, explore low-energy subspaces efficiently even in beyond one-dimensional settings and in settings where classical projector Monte Carlo methods suffer from the sign problem, as in frustrated or fermionic systems. } Further theoretical analysis and
\revv{classical} simulations may be instrumental in quantifying the scope of the quantum advantage in these more complex scenarios.

\section*{Data availability}
The data that support this study are available upon request.

\section*{Code availability}
The codes that support this study are available on GitHub via \url{https://github.com/lin-lin/oneancillaground}.

\vspace{1em}
\emph{Acknowledgments--}
This material is based upon work supported by the U.S. Department of Energy, Office of Science, Accelerated Research in Quantum Computing Centers, Quantum Utility through Advanced Computational Quantum Algorithms, grant no. DE-SC0025572 (J.P., G.K.C., L.L.). Additional support is acknowledged from the U.S. Department of Energy, Office of Science, National Quantum Information Science Research Centers, Quantum Systems Accelerator (Y.Z., Z.D., J.P., L.L.) and the National Science Foundation, grant no. PHY-2317110 (Y.Z., J.P.). The Institute for Quantum Information and Matter is an NSF Physics Frontiers Center. L.L. is a Simons Investigator in Mathematics. This research used the Savio computational cluster resource provided by the Berkeley Research Computing program at the University of California, Berkeley. Z.D., J.H. and L.L. thank the Institute for Pure and Applied Mathematics (IPAM) for its hospitality in hosting them as long-term visitors during the semester-long program   ``Mathematical and Computational Challenges in Quantum Computing'' in Fall 2023, from which this collaboration started.  The authors thank Joao Basso, Paul Cazeaux, Anthony Chen, Soonwon Choi, Marius Junge, Michael Kastoryano, Jianfeng Lu, Christian Mendl, Gunhee Park, Cambyse Rouz{\'e}, Yu Tong and Lexing Ying for helpful discussions.

\vspace{1em}
\emph{Author contributions--} G.K.C. and L.L. conceived the original study.
Y.Z., Z.D., and L.L. carried out theoretical analysis to support the study.
Y.Z., J.H. and L.L. carried out numerical calculations to support the study. J.G. provided support in tensor network-based simulations. All authors, Y.Z., Z.D., J.H., J.G., J.P., G.K.C., and L.L. discussed the results of the manuscript and contributed to the writing of the manuscript.


\clearpage
\widetext

\appendix

\section{Notation}\label{sec:notation}
For a matrix $A\in\CC^{N\times N}$, let $A^*, A^T, A^{\dag}$ be the complex conjugation, transpose, and Hermitian transpose (or adjoint) of $A$, respectively.
Unless specified otherwise, $\norm{A}\equiv \norm{A}_\infty$ denotes the operator norm, and $\norm{A}_{1} = \mathrm{Tr}(\sqrt{A^\dagger A})$ denotes the $1$-norm or the trace norm. The trace distance between two states $\rho,\sigma$ is $D(\rho,\sigma):=\frac12 \norm{\rho-\sigma}_1$. We write $A \succeq 0$ (resp., $A \succ 0$) for a positive semidefinite (resp., definite) matrix,  $A \succeq B$ if $A-B\succeq 0$, and $A\preceq B$ if $B\succeq A$.

We adopt the following asymptotic notations \revvv{beside} the usual big $\Or$ one. We write $f=\Omega(g)$ if $g=\Or(f)$; $f=\Theta(g)$ if $f=\Or(g)$ and $g=\Or(f)$. The notations $\wt{\Or}$, $\wt{\Omega}$, $\wt{\Theta}$ are used to suppress subdominant polylogarithmic factors. Specifically, $f = \wt{\Or}(g)$ if $f = \Or(g\operatorname{polylog}(g))$; $f = \wt{\Omega}(g)$ if $f = \Omega(g\operatorname{polylog}(g))$; $f = \wt{\Theta}(g)$ if $f = \Theta(g\operatorname{polylog}(g))$. Note that these tilde notations do not remove or suppress dominant polylogarithmic factors. For instance, if $f=\Or(\log g \log\log g)$, then we write $f=\wt{\Or}(\log g)$ instead of $f=\wt{\Or}(1)$.

In this paper, we consider spin systems on a $D$-dimensional lattice $\Lambda= [0,L]^D$ for some integer $L>0$. The total number of lattice sites is $N=(L+1)^D$. We measure the distance between $i,j\in \Lambda$ using the Manhattan distance (with or without the periodic boundary condition). For $j\in \Lambda$, let $\mathcal{B}_{j}(r)$ be the set of indices in $\Lambda$ with a Manhattan distance at most $r$ to the site $j$. If an operator $O\in \mathbb{C}^{2^N\times 2^N}$ can be decomposed as $O=\sum_{j\in \Lambda} O_j$, where each $O_j$ is supported on $\mathcal{B}_{j}(r)$, then $O$ is called an $r$-geometrically local Hamiltonian. If each site $i\in \Lambda$ also appears in at most $l$ non-trivial $O_j$ terms, then $O$ is called an $(r,l)$-geometrically local Hamiltonian. Given $C,\mu>0$, if $O$ can be decomposed as $O=\sum_{r\ge 1} O_r$, where each $O_r=\sum_{j\in \Lambda} O_{r,j}$ is $r$-geometrically local and  satisfies
  \begin{equation}
    \max_{j\in \Lambda}\|O_{r,j}\|\leq C\exp(-\mu r)\,,
  \end{equation}
  then $O$ is called a $(C,\mu)$-quasilocal operator.

  The definition above can be directly generalized to fermionic operators on a lattice $\Lambda$. We refer readers to e.g., ~\cite{Hastings_2019} and~\cite[Definition 7]{tong2024fast}.

\section{\rev{Comparison of mixing time metrics}}\label{sec:compare_mixing_time}

 The following inequality, originally due to Fuchs and van de Graaf ~\cite{FuchsVanDeGraaf2002}, plays a central role in our analysis of converging to ground state.

\begin{thm}[Fuchs--van de Graaf {\cite{FuchsVanDeGraaf2002}\cite[Section 9.2]{NielsenChuang2000}}]
For two density matrices $\rho,\sigma$, let $F(\rho, \sigma)= \operatorname{Tr}\left[\sqrt{\rho^{\frac{1}{2}} \sigma \rho^{\frac{1}{2}}}\right]$ be the fidelity and $D(\rho,\sigma)$ be the trace distance. Then
\begin{equation}\label{eqn:Fuchs}
1-F(\rho, \sigma) \leq D(\rho, \sigma) \leq \sqrt{1-F(\rho, \sigma)^2}.
\end{equation}
\label{thm:fuchs}
\end{thm}
Note that the fidelity of quantum states is sometimes defined as $F(\rho,\sigma)^2$. When $\sigma = \ketbra{\psi_0}{\psi_0}$ is a pure state, Theorem~\ref{thm:fuchs} can be used to establish a relation between the trace distance, the energy error, and the infidelity.

\begin{prop}\label{prop:gap_controls_trace_distance}
Let $\{(\lambda_i,\ket{\psi_i})\}$ be the eigenpairs of the Hamiltonian $H$, ordered such that $\lambda_0<\lambda_1=\lambda_0+\Delta\le \lambda_2\le\cdots$, where $\Delta>0$ is the spectral gap and $\sigma=\ketbra{\psi_0}{\psi_0}$ is the unique ground state. For any density matrix $\rho$, we have
\begin{equation}\label{eq:trace_energy}
\frac{\operatorname{Tr}[H\rho]-\lambda_0}{4\norm{H}}\le
\frac12(1-F(\rho, \sigma)^2) \le D(\rho,\sigma)\ \le \sqrt{1-F^2(\rho,\sigma)}\ \le\ \sqrt{\frac{\operatorname{Tr}[H\rho]-\lambda_0}{\Delta}}.
\end{equation}
\end{prop}

\begin{proof}
Using the spectral decomposition of $H$ we have
\begin{equation}
H-\lambda_0 I\ =\ \sum_{i\ge 1} (\lambda_i-\lambda_0)\ket{\psi_i}\bra{\psi_i}\ \ge\ \Delta\sum_{i\ge 1} \ket{\psi_i}\bra{\psi_i}\ =\ \Delta\,(I-\sigma).
\end{equation}
When $\sigma=\lvert\psi_{0}\rangle\langle\psi_{0}\rvert$ is a pure state, the fidelity simplifies to $F(\rho,\sigma)=\sqrt{\Tr[\braket{\psi_0|\rho|\psi_0}]}$.
Taking expectation in $\rho$ yields
\begin{equation}
\operatorname{Tr}[(H-\lambda_0 I)\rho]\ \ge\ \Delta\,\operatorname{Tr}[(I-\sigma)\rho]\ =\ \Delta\,(1-F^{2}(\rho,\sigma)).
\end{equation}
Applying the Fuchs-van de Graaf inequality yields the upper bound for $D(\rho,\sigma)$.

Note that
\begin{equation}
1-F^2 = (1-F)(1+F)\le 2(1-F),
\end{equation}
and
\begin{equation}
\operatorname{Tr}[(H-\lambda_0 I)\rho]\le 2\norm{H} \Tr\left[ \rho (I-\sigma)\right] =2\norm{H}(1-F(\rho,\sigma)^2).
\end{equation}
Applying the Fuchs-van de Graaf inequality again yields the lower bound for $D(\rho,\sigma)$.

\end{proof}

We can also define the fidelity-based and energy-based mixing times, as given in \cref{eqn:tau_E,eqn:tau_F}, to be independent of the initial state:
\begin{equation}\label{eqn:mixing_tracedistance}
\tau_{\operatorname{mix}}^E(\eta) = \max_{\rho_0} \tau_{\operatorname{mix}}^E(\eta;\rho_0),\quad \tau_{\operatorname{mix}}^F(\eta) = \max_{\rho_0} \tau_{\operatorname{mix}}^F(\eta;\rho_0)
\end{equation}
From Proposition \ref{prop:gap_controls_trace_distance}, we immediately obtain
\begin{equation}\label{eqn:mixingtime_relations}
\tau^E_{\operatorname{mix}}(4\norm{H}\eta) \le \tau^F_{\operatorname{mix}}(2\eta ) \le \tau_{\operatorname{mix}}(\eta) \le \tau^F_{\operatorname{mix}}(\eta^2) \le \tau^E_{\operatorname{mix}}(  \Delta\eta^2).
\end{equation}
Thus, when maximizing over all initial states $\rho_0$, the mixing times defined via fidelity or energy provide both upper and lower bounds for the mixing time defined via trace distance.

\section{Quasi-free systems}\label{sec:prelim_majorana}

\paragraph*{Jordan--Wigner transformation---}

\revvv{We introduce the Jordan-Wigner transformation for fermionic annihilation and creation operators following the convention in \cite{Pfeuty1970},}
\begin{equation}
c_j = \left( \prod_{k=1}^{j-1} Z_k \right) X_j^-=X_j^- \left( \prod_{k=1}^{j-1} Z_k \right) , \quad  c_j^\dagger = X_j^+ \left( \prod_{k=1}^{j-1} Z_k \right)= \left( \prod_{k=1}^{j-1} Z_k \right) X_j^+,
\end{equation}
with 
\begin{equation}
X_j^- = \frac{1}{2}(X_j - iY_j), \quad
X_j^+ = \frac{1}{2}(X_j + iY_j), \quad \revvv{Z_j = 2 c_j^\dagger c_j - I.}
\end{equation}

After the Jordan-Wigner transformation,
\begin{equation}
X_j X_{j+1}=X_j Z_j  \left( \prod_{k=1}^{j-1} Z_k \right) (c_{j+1}^{\dag}+c_{j+1})=-iY_{j} \left( \prod_{k=1}^{j-1} Z_k \right) (c_{j+1}^{\dag}+c_{j+1})=(c_j-c^{\dag}_j)(c_{j+1}^{\dag}+c_{j+1}).
\end{equation}
Then the 1D TFIM Hamiltonian in \eqref{eqn:H_TFIM}
$$
    H=-g \sum_{i=1}^N Z_i-J \sum_{i=1}^{N-1} X_i X_{i+1}
$$
can be expressed as
\begin{equation}\label{eqn:H_TFIM_fermion}
H=-J\sum_{j=1}^{N-1} (c_j-c^{\dag}_j)(c_{j+1}^{\dag}+c_{j+1}) -2g\sum_{j=1}^N c_j^{\dag} c_j + gN.
\end{equation}
We now perform a unitary rotation
\begin{equation}\label{eqn:fermion_majorana}
\binom{c_j}{c_j^{\dagger}}=\frac{1}{\sqrt{2}}\left(\begin{array}{cc}
1 & -i \\
1 & i
\end{array}\right)\binom{w_{j}}{w_{j+N}}, \quad j=1,\ldots,N.
\end{equation}
This defines a set of $2N$ Majorana operators, $\left\{w_p\right\}_{p=1}^{2 N}$, which are Hermitian operators satisfying the anticommutation relation
\begin{equation}
\left\{w_p, w_q\right\}:=w_p w_q+w_qw_p=\delta_{p q}, \quad p,q=1,\ldots, 2N.
\end{equation}
This gives rise to the Majorana form of the Hamiltonian in \cref{eqn:H_TFIM_majorana}:
$$
H =2i J \sum_{j=1}^{N-1} w_{j+N}w_{j+1} + 2i g \sum_{j=1}^N  w_{j}w_{j+N}.
$$

\paragraph*{Quadratic Majorana systems---}

The general form of a quadratic Majorana Hamiltonian with $2N$ modes is:
\begin{equation}
H=2 \sum_{1 \leq p<q \leq 2 N} h_{p q} w_p w_q= \sum_{p, q=1}^{2 N} h_{p q} w_p w_q,
\end{equation}
The coefficient matrix $h$ is Hermitian and purely imaginary. In other words,   we may write $h=-i \mathsf{A}$, where $\mathsf{A}$ is a real antisymmetric matrix. The eigenvalues of the coefficient matrix $h$ are thus real and symmetric with respect to $0$. Let $\{\lambda_k\}_{k=1}^N$ be the \emph{non-negative} eigenvalues of $h$. If $\Delta=\min_{k} \lambda_k>0$, then $H$ is called a gapped Hamiltonian and $\Delta$ is referred to as the spectral gap. Then after a unitary transformation, we may write
\begin{equation}\label{eqn:canonical_quadratic}
H=2\sum_{k=1}^{N}\lambda_k b_k^{\dag} b_k+\text{constant}.
\end{equation}
Here $\{b_k,b_k^{\dag}\}$ is a set of fermionic annihilation and creation operators satisfying the canonical anticommutation relation (CAR), and is linear in the fermionic operators $\{c_j,c_j^{\dag}\}$ in \cref{eqn:fermion_majorana}. The ground state is the quasivacuum state satisfying (see e.g. \cite[Chapter 3.3]{BlaizotRipka1985})
\begin{equation}
b_k\ket{\mathrm{vac}}=0, \quad k=1,\ldots,N.
\end{equation}

\paragraph*{Quasi-free dynamics---}

We consider a general noninteracting quadratic Hamiltonian $H=\sum_{i,j=1}^{2N} h_{ij} w_i w_j$. According to the Thouless theorem \cite{Thouless1960}, \begin{equation}
e^{\mathrm{i}H s} w_a e^{-\mathrm{i}H s}=\sum_p w_p \left(e^{-2\mathrm{i}h \revvv{s}}\right)_{ap}.
\end{equation}
As a result, the jump operator associated with a coupling  operator $w_a$ is a linear combination of Majorana operators
\begin{equation}
\begin{aligned}
K_a  &= \int_{\RR} f(s) e^{\mathrm{i}Hs}w_{a} e^{-\mathrm{i}Hs}\mathrm{d} s\\
&= \sum_p \int_{\RR} f(s)\left(e^{-\mathrm{i} 2 h s}\right)_{ap} w_p  \mathrm{d} s \\
&=\sum_p [\hat{f}(-2h)]_{ap} w_p.
\end{aligned}
\end{equation}
Then if the set of coupling operators is $\{A_a=w_a\}_{a\in\mc{I}}$ where $\mc{I}$ is some index set, a closed-form equation for the covariance matrix $\Gamma_{p q}:=\frac{i}{2}\left\langle w_p w_q - w_q w_p\right\rangle$ can be derived as \cite[Proposition 1]{BarthelZhang2022}:
\begin{equation}
\label{covmat}
\partial_t \Gamma = X\Gamma +\Gamma X^T + Y, \quad X=-2ih-B_\text{real}, \quad Y=B_\text{imag},
\end{equation}
Here the coefficient matrix
\begin{equation}
\label{matB}
 B_{pq}=\sum_{a\in \mc{I}} \left[\hat{f}(-2h)\right]_{a p}\left[\hat{f}(-2h)\right]_{aq}^*.
\end{equation}
Here, $B$ is the sum over all the coefficients of the jump operators, with $B_\text{real}$, $B_\text{imag}$ denoting the (entry-wise) real and imaginary parts of $B$, respectively. Since the filter function $\hat{f}(\omega)$ is supported only on the negative real axis, the Lindbladian dynamics filters out all positive eigenmodes of $h$ while simultaneously populating the negative modes, which contributes to the ground state of the quadratic Hamiltonian $H$.

\section{\rev{Additional results for quasi-free systems}}

\vspace{1em}

\paragraph*{Quasilocality of jump operators---}
We plot the heat map of the coupling operator $X_1$ in the computational basis, and the corresponding jump operator $K_{X_1}$ in the energy basis in Fig.~\ref{fig:TFIMheat}. We find that although  $X_1$ is very sparse in the computational basis, $K_{X_1}$ has \revvv{significantly more} nonzero elements in the energy basis enabling transitions from high energy components to low energy ones. The filter function forbids transitions from low to high energy components. Therefore the jump operator is always an upper triangular matrix in the energy eigenbasis. Furthermore, the magnitude of the coefficients $|\zeta_{j}|^2+\abs{\zeta_{j+N}}^2$ decays exponentially as $j$ increases ($1\le j<N$), which implies that $K_a$ is quasilocal in Majorana operators (see~\cref{fig:TFIM_quasi_locality}).

\begin{figure}[h!]
    \centering
    \includegraphics[width=0.5\linewidth]{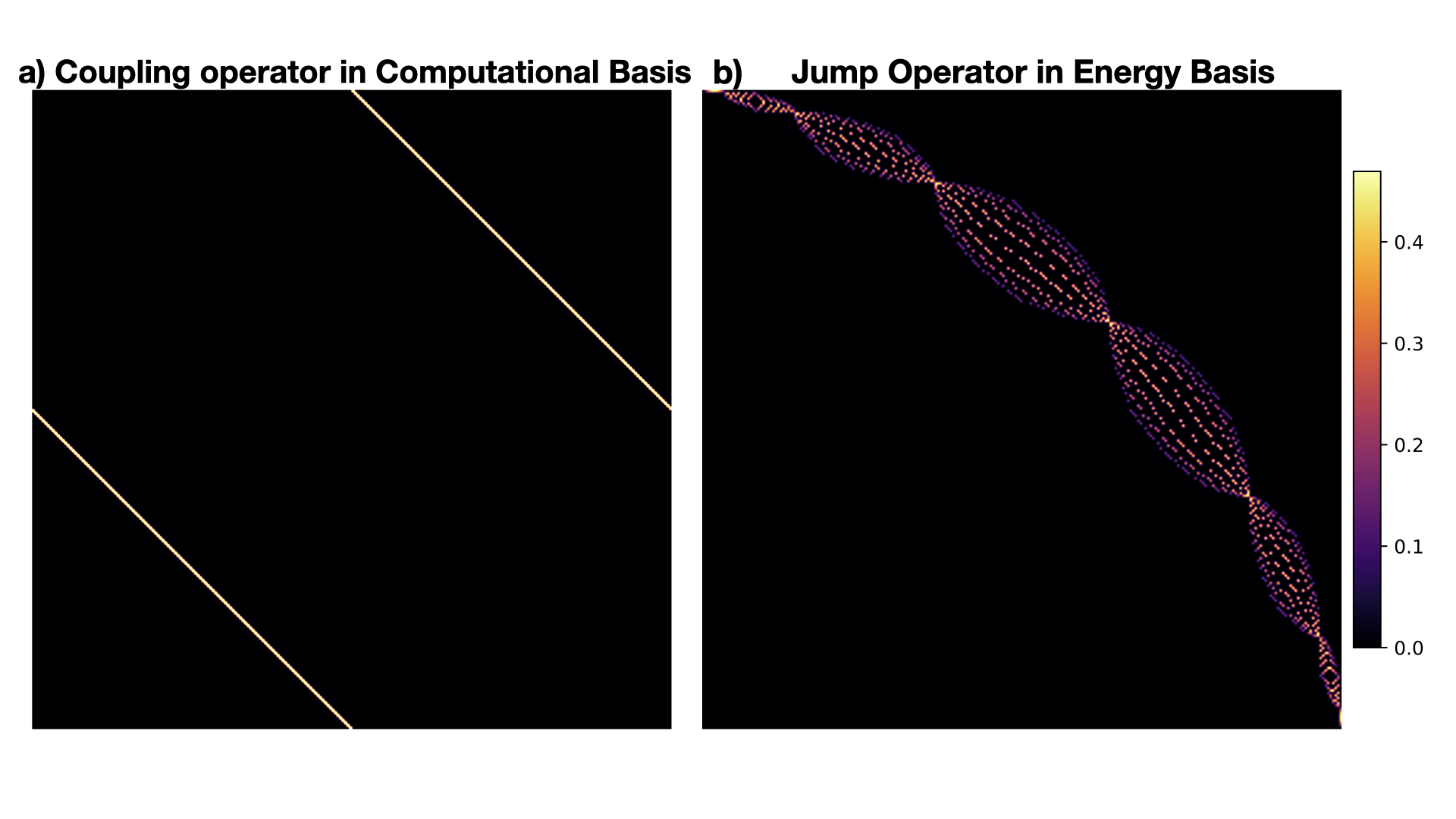}
    \caption{(a). Matrix elements of a local Pauli operator $X_1$ in the computational basis. (b). Jump operator $K_{X_1}$ associated with coupling operator $X_1$ in the energy basis of  the TFIM Hamiltonian with $N=8$ sites. The lower triangular part vanishes due to the filter function.}
    \label{fig:TFIMheat}
\end{figure}
\begin{figure}[h!]
    \centering
    \includegraphics[width=0.4\linewidth]{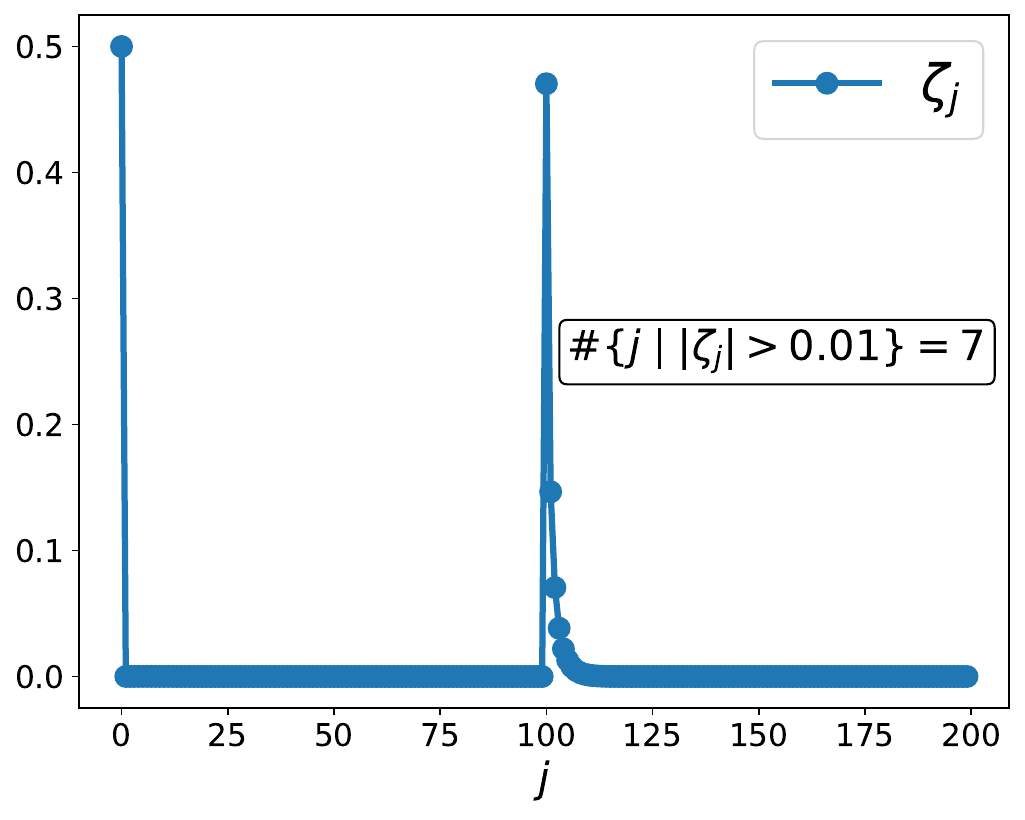}
    \caption{Coefficients of the jump operator under Majorana basis. $|\zeta_j|$ is calculated from~\eqref{eqn:K_a} with $A_a=X_1$. System size is $N=100$.}
    \label{fig:TFIM_quasi_locality}
\end{figure}

\paragraph*{Importance of the coherent term---}

In the case of boundary dissipation with coupling operators $X_1, X_N, Y_1, Y_N$, the coherent term $-i[H,\cdot]$ plays a critical role for the system to converge to the ground state, as illustrated in \cref{fig:coherent}. Physically, since the jump operator is localized near the boundary, dissipation primarily occurs there. The coherent term induces an energy flux from the bulk to the boundary, which effectively reduces the energy.

Mathematically, without the coherent term, the Lindblad dynamics lacks a unique fixed point. The role of the coherent term is to lift this large degeneracy, and place eigenvalues on the imaginary axis. The dissipative term then slightly perturbs these eigenvalues away from the imaginary axis, creating a spectral gap that leads to convergence. This will be rigorously justified in \cref{sec:mixingtime_quasifree}.

\begin{figure}[htbp]
  \centering
  \includegraphics[width=0.35\textwidth]{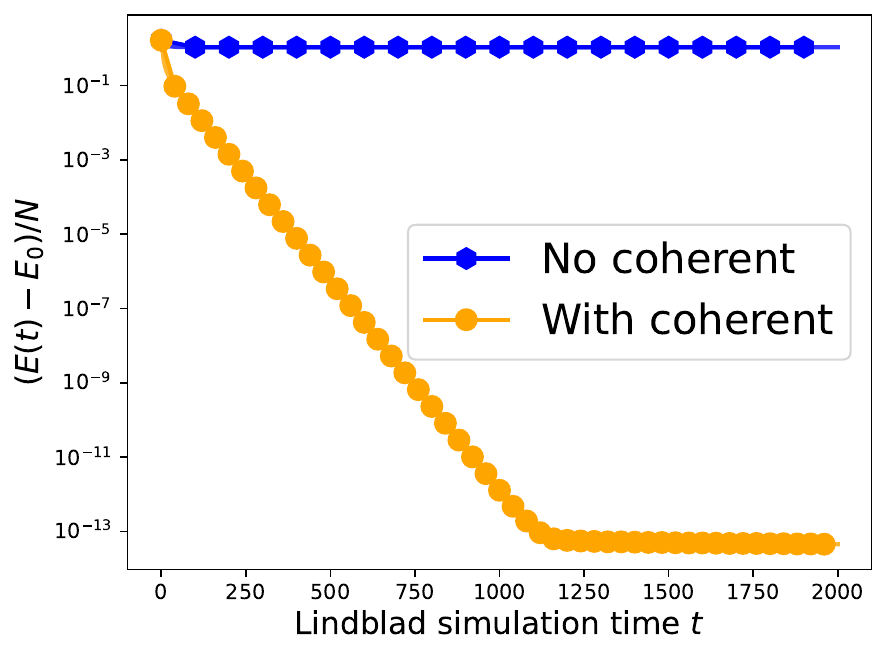}
  \caption{Convergence of energy in TFIM  using $\{X_1, X_N, Y_1, Y_N\}$ as coupling operators, with and without the coherent term $-i[H,\cdot]$ in the Lindbladian.}
  \label{fig:coherent}
\end{figure}

\vspace{1em}
\paragraph*{\rev{Convergence starting from different initial states}---}

We choose an \revvv{all-ones} initial state $\ket{1^N}$ as the initial state for boundary-dissipated TFIM. Other parameter settings are the same in \cref{sec:Results: Fig: tfimdynamics}. The results are presented in \cref{Fig: tfimdynamics_diff_ini}. We also observe that $\Delta_{\mc{L}}=\Theta(N^{-3})$, which matches the scaling of the energy-based mixing time $\tau^E_{\operatorname{mix}}=\Theta(N^3)$ for fixed $\eta$.

\section{Additional numerical results on tensor network simulation}\label{sec:additional_numer_tensor}

\begin{figure}[ht]
  \centering
    \begin{subfigure}[b]{0.35\textwidth}
        \includegraphics[width=\textwidth]{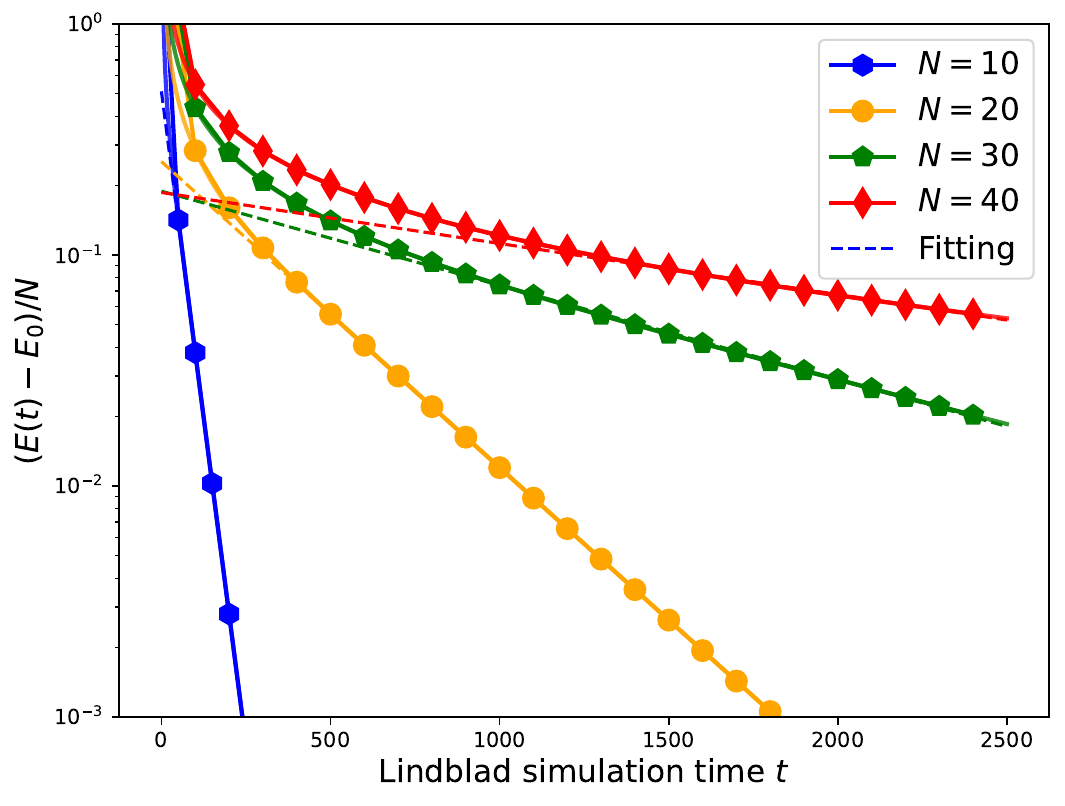}
        \caption{ \centering }
    \end{subfigure}
    \begin{subfigure}[b]{0.35\textwidth}
        \includegraphics[width=\textwidth]{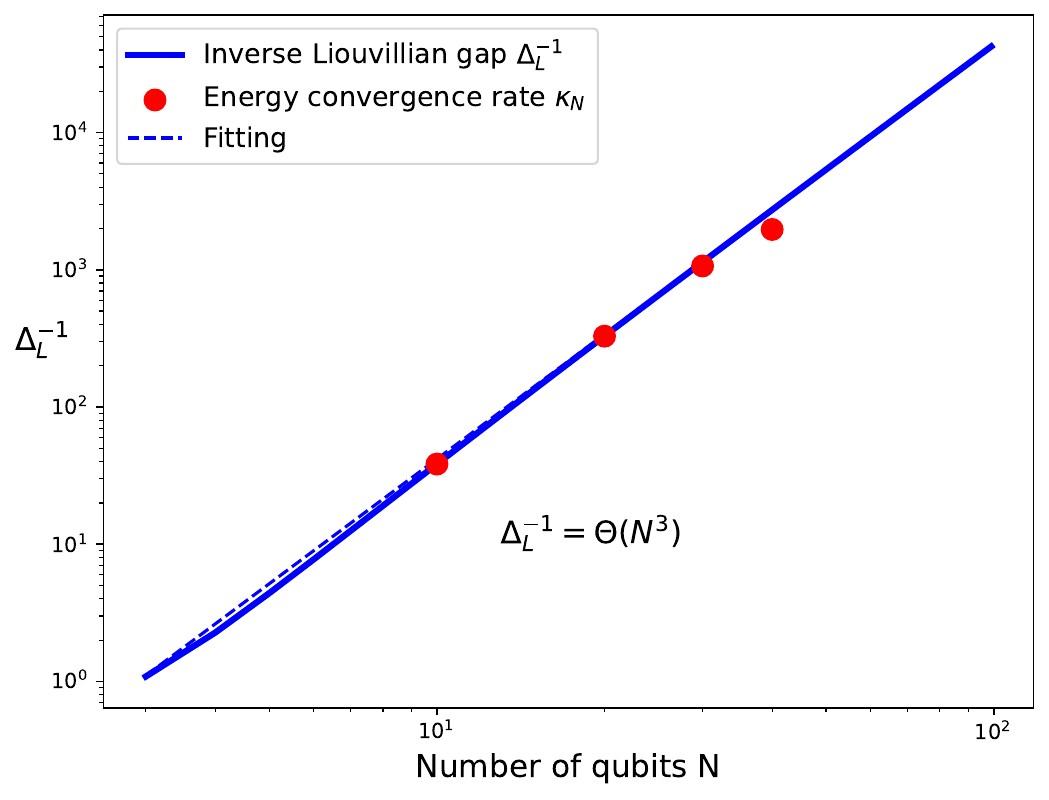}
        \caption{ \centering }
    \end{subfigure}
  \caption{
  \rev{Numerical results for the 1D TFIM \eqref{eqn:H_TFIM} with the initial state $\ket{1^N}$. (a) Convergence of the energy during Lindblad dynamics. (b) Scaling of the inverse Liouvillian gap $\Delta_{\mathcal{L}}^{-1}$ as a function of the system size $N$. Red points indicate the fitted energy convergence rates $\kappa_N$ extracted from panel (a).
}}
  \label{Fig: tfimdynamics_diff_ini}
\end{figure}

A key part of simulating the Lindbladian evolution is the compression of a sum of $M$ triple MPO products to a single MPO of bond dimension $D$ on $N$ sites, as described in \cref{sec:TNmethod}.
A \emph{direct method} contracts the three tensors per site for each term, then explicitly sums the $M$ terms, yielding a single MPO with bond dimension $MD^3$ which can be compressed using a sweep of QR decompositions and singular value truncations. The cost of this direct method scales as $\mathcal{O}(N M^3 D^9)$.

An alternative option would be to interleave compressions and contractions, that is, immediately compress back to bond dimension $D$ after every pairwise MPO multiplication or addition. Such an approach has a better scaling of $\mathcal{O}(N M D^6)$. However, this approach can introduce a significantly larger error, as the number of compressions performed is $\Or(NM)$ rather than $\Or(N)$ in the previous case.

To resolve this problem, we employ the \emph{fitting method}~\cite{VerstraeteCirac2004}, which iteratively constructs the optimal (in terms of Frobenius norm) 1D approximation of a tensor network using only the overlap between the ansatz and the target network. Since the target is a linear sum of terms, each overlap can be calculated separately.
This leads to a scaling of $\mathcal{O}(K M N D^5)$ with $K$ the number of sweeps required to converge the fitting procedure (typically $<20$). The library \textsf{quimb}~\cite{Gray2018} enables the fitting of a sum of such MPO product terms to a single MPO.  It also supports the use of GPUs, which can greatly speed up the computations dominated by linear algebra operations such as this fitting routine. We report the results in \cref{fig:compression_time}.

To demonstrate the accuracy of the fitting method, we calculate $\Tr(H K_a \rho_{t=0.1} K_a^{\dagger})$ from the bulk-dissipated 1D-TFIM model with $J=1$ and $g=1.5$, and compress the triple MPO product $K_a \rho K_a^{\dagger}$ using the direct method and the fitting method. Here $K_a$ is the jump operator with coupling operator $X$ at the fifth site and $\rho_{t=0.1}$ is the density matrix at $t=0.1$ obtained from the Lindblad dynamics using a bond dimension of $50$. Then, $K_a$ and $\rho_{t=0.1}$ are compressed into MPOs with a reduced bond dimension $D$ ranging from 6 to 14. For each given $D$, the absolute difference between $\Tr(H K_a \rho_{t=0.1} K_a^{\dagger})$ computed using the direct method and the fitting method is very small (see \cref{fig:fitting_error}). We note that for $D > 14$, the direct method becomes too expensive to use for comparison, while the fitting method remains efficient.

\begin{figure}
\begin{center}
\includegraphics[width=0.45\textwidth]{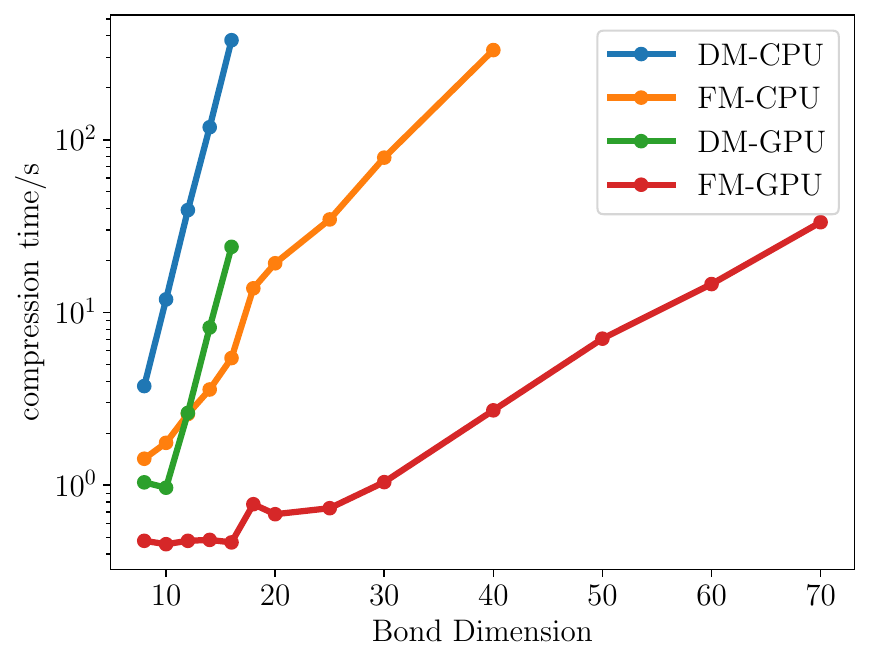}
\end{center}
  \caption{Compression time of multiplying three $N=20$ random tensor networks with the direct method (DM) and fitting method(FM). The direct method is not applicable for bond dimension $D>14$ due to memory constraints.}
\label{fig:compression_time}
\end{figure}

\begin{figure}
\begin{center}
\includegraphics[width=0.45\textwidth]{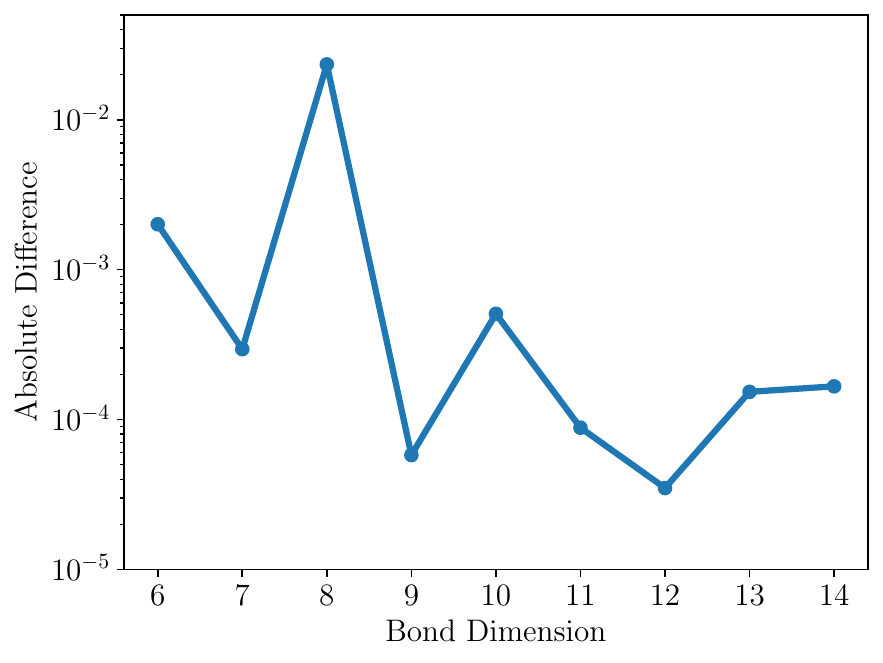}
\end{center}
  \caption{Absolute difference of $\Tr(H K_a \rho_{t=0.1} K_a^{\dagger})$ calculated using the direct method (DM) and the fitting method (FM) as the bond dimension $D$ increases. }
\label{fig:fitting_error}
\end{figure}

Then, we numerically validate that the jump operator can be compressed into an MPO with relatively low bond dimension. To this end, we evaluate $\norm{K_a\ket{\psi_g}}^2$ for the 1D-TFIM model with $J=1$ and $g=1.5$, using a fixed bond dimension $D=50$, with the coupling operator $X$ positioned at the center of the chain. We then investigate how $\norm{K_a\ket{\psi_g}}^2$ scales with increasing system size $N$.
 Ideally, this quantity should remain small since $\ket{\psi_g}$ belongs to the kernel of the jump operator. \cref{fig:jump_check} shows that  $\norm{K_a\ket{\psi_g}}^2$ remains approximately $10^{-3}$ as the system size grows. This result confirms that a bond dimension of $D=50$ can be sufficiently accurate for representing the ground state and the jump operators.

\begin{figure}[ht]
    \includegraphics[width=0.45\textwidth]{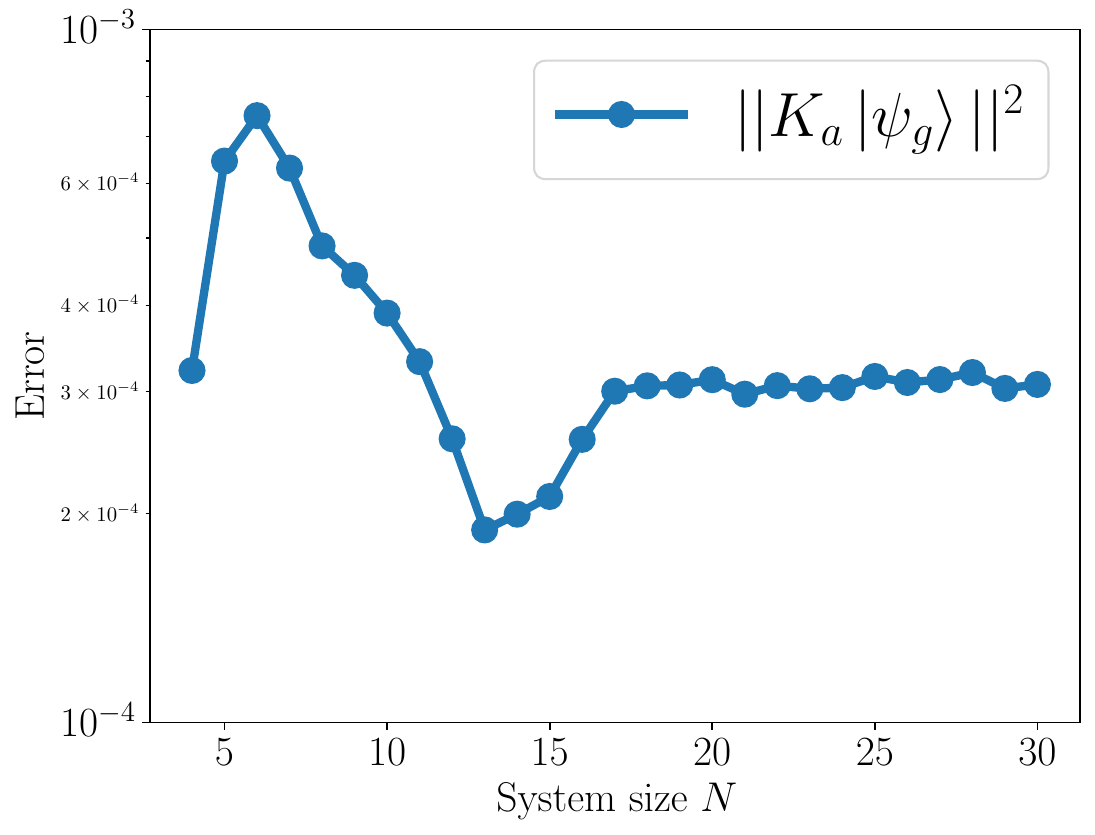}
  \caption{
$\norm{K_a\ket{\psi_g}}^2$ as a function of system size $N$ for the 1D-TFIM model with $a=\lfloor L/2\rfloor $. The bond dimension is set to $D = 50$. }
  \label{fig:jump_check}
\end{figure}
\begin{figure}[ht]
    \includegraphics[width=0.45\textwidth]{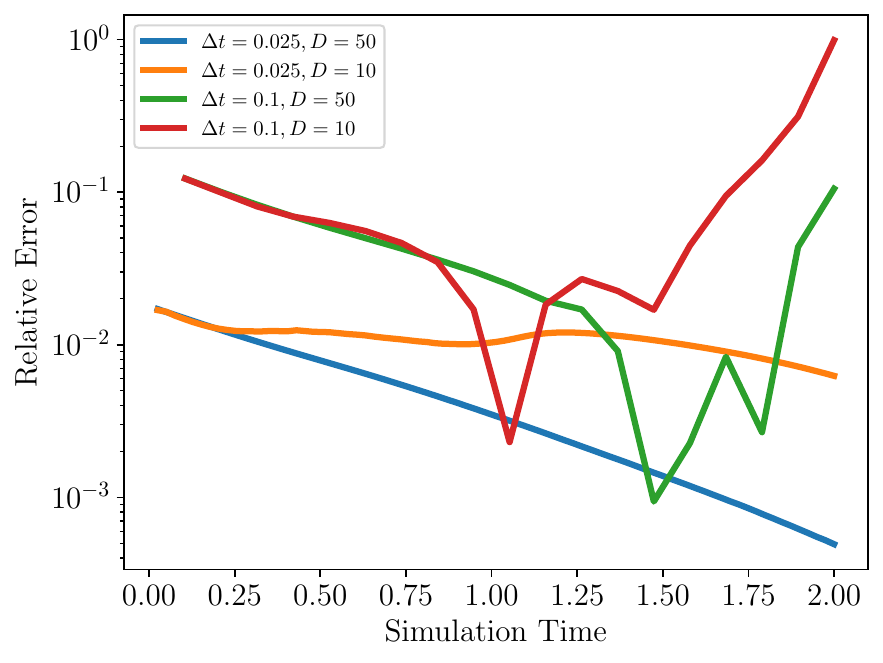}
  \caption{Relative error of energy with different time step sizes and bond dimensions for 1D TFIM model. }
  \label{fig:TNerror}
\end{figure}
Finally, in~\cref{fig:TNerror}, we demonstrate the reliability of the tensor network results by comparing simulations with different bond dimensions. We consider the bulk dissipation of an $N=10$ one-dimensional TFIM model using bond dimensions $D=10, 50$ and a time step size of $\Delta t = 0.025$. The relative error is computed using $D=75$ and $\Delta t = 0.0125$ as the reference (see \cref{fig:TNerror}). For small simulation times, the dominant source of error is the time discretization. A large time step size may also lead to unstable behavior as the simulation time increases. Furthermore, as the system evolves towards a fixed point, the error is no longer primarily dictated by time discretization, and the bond dimension plays a more significant role.

\section{Mixing time of quasi-free systems and proof of Theorem \ref{thm:sharpbound_quasifree}}\label{sec:proof_sharpbound_quasifree}

In order to characterize the convergence rate of Lindblad dynamics, one standard approach is to evaluate the Liouvillian gap. When the dissipative dynamics satisfies the quantum detailed balance condition (DBC), the Lindbladian may be transformed into a Hamiltonian under a similarity transformation, and the Liouvillian gap can be bounded using techniques for bounding spectral gaps for quantum many-body Hamiltonians~\cite{rouz2024,tong2024fast}. This strategy, however, cannot be applied to Lindblad dynamics with a coherent term, which breaks the DBC.

A more general formulation for bounding the spectral gap may be captured by the hypocoercivity theory, which was originally formulated in the context of classical kinetic theory~\cite{Villani2007} and has recently been applied in the context of open quantum systems described by Lindblad equations~\cite{fang2024mixingtimeopenquantum}. However, the formulation in~\cite{fang2024mixingtimeopenquantum} can only be used to prove the existence of a spectral gap. For quasi-free dynamics, the spectral gap can also be derived explicitly from the rapidity spectrum~\cite{Prosen2010} related to the equation of motion for the covariance matrix.

Even with spectral gap estimates, the problem remains how to bound convergence of the density matrix in trace distance. This is because spectral gap estimates only imply convergence in $\chi^2$-distance~\cite{TemmeKastoryanoRuskaiEtAl2010}, and the conversion from convergence in $\chi^2$-distance to the trace distance involves a factor that blows up exponentially as temperature decreases, making it inapplicable for ground state preparation when the temperature is zero. A special technique called hypercontractivity can be applied to quasi-free quantum groups~\cite{TemmePastawskiKastoryano2014} but also requires the stationary state to be invertible.

\rev{Theorem \ref{thm:sharpbound_quasifree} provides an explicit bound on the convergence in trace distance to the ground state for general quasi-free dynamics, and its proof is given below.}

After the canonical unitary transformation and the particle hole transformation, we express the quadratic Hamiltonian $H$ in the canonical form of \eqref{eqn:canonical_quadratic}. 
Let $\hat{n}_k=b_k^{\dag} b_k$ be the number operator of the $k$-th mode, and $\hat{N}=\sum_{k} \hat{n}_k=\vb^{\dag} \vb$ be the total number operator.
Note that all creation operators $b_k^{\dag}$ increase the energy, while all annihilation operators $b_k$ decrease the energy. As a result, the jump operator must be a linear combination of annihilation operators alone:
\begin{equation}
K_a=\sum_{p=1}^{N} \Phi^*_{pa} b_p
\end{equation}
for some coefficient matrix $\Phi$, so that $K_a \ket{\mathrm{vac}}=0$.

Define $O=\sum_{k,l=1}^{N} \Xi_{kl} b^{\dag}_{k} b_l$ for some positive definite matrix $\Xi$ to be determined. Then
\begin{equation}
i[H,O]=2i\sum_{k} \lambda_k [n_k, O]=2i \sum_{kl} (\lambda_k-\lambda_l) \Xi_{kl} b_k^{\dag} b_l.
\end{equation}
We may also directly compute
\begin{equation}
\begin{split}
&\sum_a K_a^{\dag} O K_a - \frac{1}{2} \{K^{\dag}_a K_a, O\} = \frac{1}{2} \sum_a (K^{\dag}_a [O, K_a] - [O, K^{\dag}_a] K_a)\\
=& -\frac12 \sum_a \left(\sum_{p} \Phi_{pa} b^{\dag}_p \sum_{kl} \Xi_{kl} \Phi_{ka}^* b_l + \sum_{kl} \Xi_{kl} \Phi_{la} b_k^{\dag} \sum_{p} \Phi_{pa}^* b_p\right)\\
=& -\frac12 \sum_{kl} b_k^{\dag} \left(\sum_{ap} \Phi_{ka} \Phi_{pa}^* \Xi_{pl} + \Xi_{kp} \Phi_{pa} \Phi_{la}^*\right) b_l.
\end{split}
\end{equation}
Let $h=\diag(\{2 \lambda_k\})$, and define a non-Hermitian Hamiltonian (the superscript $f$ means the non-Hermitian matrix is defined with respect to fermionic creation and annihilation operators instead of Majorana operators):
\begin{equation}
H_{\mathrm{nh}}=iH-\frac12 \sum_a  K_a^{\dag}K_a=\vb^{\dag} h^f_{\mathrm{nh}} \vb, \quad h^{f}_{\mathrm{nh}}=ih-\frac12 \Phi\Phi^{\dag},
\end{equation}
then
\begin{equation}\label{eqn:Ldag_O_quasifree}
\mc{L}^{\dag}[O]=\vb^{\dag}\left[\left(ih-\frac12 \Phi\Phi^{\dag}\right)\Xi+\Xi\left(-ih-\frac12 \Phi\Phi^{\dag}\right)\right]\vb=\vb^{\dag}(h^{f}_{\mathrm{nh}}\Xi+\Xi (h^{f}_{\mathrm{nh}})^{\dag})\vb.
\end{equation}
Here the coefficient matrix $h^{f}_{\mathrm{nh}}$ is related to that in \cref{eqn:nonhermitian_ham} by a similarity transformation.

Under the assumption that $h^{f}_{\mathrm{nh}}$ can be diagonalized as $VDV^{-1}$ for some invertible matrix $V$ and diagonal matrix $D$, we now make a choice of the Hermitian matrix $\Xi=VV^{\dag}\succ 0$. Let $\Delta=-\max _i \operatorname{Re} D_{i i}$ be the non-Hermitian gap. Then
\begin{equation}
\mc{L}^{\dag}[O]=\vb^{\dag}V(D+D^*)V^{\dag}\vb\preceq -2\Delta \vb^{\dag} V V^{\dag} \vb= -2\Delta O.
\end{equation}
This immediately yields the exponential convergence of the observable $O$ as
\begin{equation}
\Tr[O \rho(t)]\le \Tr[O \rho(0)] e^{-2\Delta t}.
\end{equation}

From the definition of $\Xi$ we have
\begin{equation}
\lambda_{\min}(VV^{\dag}) \hat{N} \preceq O\preceq
\lambda_{\max}(VV^{\dag}) \hat{N}.
\end{equation}
Then the infidelity can be bounded as
\begin{equation}
\begin{split}
1-\braket{\mathrm{vac}|\rho|\mathrm{vac}}\le \Tr[\hat{N}\rho]\le& \frac{1}{\lambda_{\min}(VV^{\dag})} \Tr[O \rho]\le \frac{1}{\lambda_{\min}(VV^{\dag})} \Tr[O \rho(0)]e^{-2\Delta t}\\
\le & \frac{\lambda_{\max}(VV^{\dag})}{\lambda_{\min}(VV^{\dag})} \Tr[\revvv{\vb^{\dag}  \vb}\rho(0)]e^{-2\Delta t}=\kappa^2(V) \Tr[\hat{N} \rho(0)] e^{-2\Delta t}.
\end{split}
\end{equation}
Finally, by the Fuchs--van de Graaf inequality,
\begin{equation}
D(\rho, \ketbra{\mathrm{vac}}{\mathrm{vac}})\le \sqrt{1-\braket{\mathrm{vac}|\rho|\mathrm{vac}}}\le \kappa(V) \sqrt{\Tr[\hat{N} \rho(0)]} e^{-\Delta t}.
\end{equation}
Finally, we use $\Tr[\hat{N} \rho(0)]\le N$ for any initial state $\rho(0)$ and finish the proof of Theorem \ref{thm:sharpbound_quasifree}.

\section{Mixing time of 1D TFIM with boundary dissipation}\label{sec:tfim_boundary_proof}

To simplify the analysis we adopt the ``c-cyclic'' approximation in \cite{LiebSchultzMattis1961,Pfeuty1970},
and consider the following periodized version of the TFIM Hamiltonian expressed in fermionic operators
\begin{equation}\label{eqn:H_TFIM_fermion_per}
H_{\mathrm{per}}=-J\sum_{j=1}^{N-1} (c_j-c^{\dag}_j)(c_{j+1}^{\dag}+c_{j+1}) -J\left(c_N-c_N^{\dagger}\right)\left(c_1^{\dagger}+c_1\right)-2g\sum_{j=1}^N c_j^{\dag} c_j+gN.
\end{equation}
Compared to \cref{eqn:H_TFIM_fermion}, this Hamiltonian introduces a periodic term. Define the ratio $\xi=J/g$. When $\xi\ne 1$, the modified Hamiltonian is gapped, and can be exactly diagonalized as
\begin{equation}
H_{\rm per,c}=2g \sum_k \Lambda_k b_k^{\dag} b_k-g \sum_k \Lambda_k.
\end{equation}

Here for convenience we assume \revvv{$g=1$,} $N$ is even, and we label the eigenvalues
from $-\frac{N}{2}$ to $\frac{N}{2}-1$, with
\begin{equation}
\Lambda_k=\sqrt{1+\xi^2+2 \xi \cos \left(\frac{2\pi k}{N}\right)}, \quad k=-\frac{N}{2},\ldots,\frac{N}{2}-1,
\end{equation}
which is even with respect to the index $k$.
The annihilation operators $\{b_k\}$ are given in the form of a Nambu spinor \begin{equation}\label{eqn:nambu_spinor}
b_k=\sum_j\left\{\left(\frac{\varphi_{k j}+\psi_{k j}}{2}\right) c_j+\left(\frac{\varphi_{k j}-\psi_{k j}}{2}\right) c_j^{\dag}\right\},
\end{equation}
with coefficients
\begin{equation}
\varphi_{k j}=\begin{cases}
(2 / N)^{1 / 2} \sin \left(\frac{2\pi jk}{N}\right) & k=1,\ldots,\frac{N}{2}-1 \\
(2 / N)^{1 / 2} \cos \left(\frac{2\pi jk}{N}\right) & k=-\frac{N}{2},\ldots,0
\end{cases}
\end{equation}
and
\begin{equation}
\psi_{k j}=-\Lambda_k^{-1}\left[\left(1+\xi \cos \left(\frac{2\pi k}{N}\right)\right) \varphi_{k j}+\xi \sin\left(\frac{2\pi k}{N}\right) \varphi_{-k, j}\right].
\end{equation}
Direct calculation shows
\begin{equation}
X_1=c_1+c_1^{\dag}=\sum_k \varphi_{k1} (b_k+b_k^{\dag}), \quad Y_1=i(c_1-c_1^{\dag})=\sum_k \psi_{k\revvv{1}} i(b_k-b_k^{\dag}).
\end{equation}
For ground state preparation, the corresponding jump operators simply filter out the energy-increasing $b_k^{\dag}$ components:
\begin{equation}
K_{X_1}=\int f(s) e^{i H s} X_1 e^{-i H s} \ud s=\sum_k \varphi_{k1} b_k, \quad K_{Y_1}=\int f(s) e^{i H s} Y_1 e^{-i H s} \ud s=\sum_k \psi_{k1} b_k.
\end{equation}
The non-Hermitian Hamiltonian in Theorem \ref{thm:sharpbound_quasifree} can be written as
\begin{equation}
H_{\rm nh} = \vb^{\dag} h_{\mathrm{nh}}^f \vb, \quad
h^{f}_{\mathrm{nh}}=i\Lambda-\frac12 \varphi\varphi^{\dag}-\frac12 \psi \psi^{\dag}.
\end{equation}
Note that the magnitudes of $\abs{\varphi_{k j}},\abs{\psi_{k j}}$ vanish at least as $\Or(N^{-\frac12})$ for large system sizes. This allows us to use  first order perturbation theory to estimate the non-Hermitian gap. For $k=0$, we have $\lambda_0=1+\xi$,
\begin{equation}
\frac12 \left(\varphi_{0,1}^2+\psi_{0,1}^2\right)=\frac{2}{N},
\end{equation}
which is  large compared to $\Or(N^{-3})$. Due to the modification to periodic boundary conditions, every eigenvalue $\Lambda_k$ with $k> 0$ is doubly degenerate with $\Lambda_k=\Lambda_{-k}$. This relation also approximately holds for the original TFIM problem with open boundary conditions. So we apply the perturbation theory to each two dimensional space spanned by the eigenvectors corresponding to the eigenvalue $\Lambda_k=\Lambda_{-k}$.
Then
\begin{equation}
\Delta \approx \frac12 \min_{k>0} \lambda_{\min}(M_k), \quad M_k=\begin{pmatrix}
\varphi_{k,1}^2+\psi_{k,1}^2 & \varphi_{k,1} \varphi_{-k,1} + \psi_{k,1} \psi_{-k,1}\\
\varphi_{k,1} \varphi_{-k,1} + \psi_{k,1} \psi_{-k,1} & \varphi_{-k,1}^2+\psi_{-k,1}^2
\end{pmatrix}.
\end{equation}
For each $0<k<N/2$, we have
\begin{equation}
\det M_k=(\varphi_{k,1}\psi_{-k,1}-\varphi_{-k,1}\psi_{k,1})^2=\frac{\xi^2 \sin^2\left(\frac{2\pi k}{N}\right)}{\Lambda_k^2}\left(\varphi_{k,1}^2+\varphi^2_{-k, 1}\right)^2=\frac{4\xi^2 \sin^2\left(\frac{2\pi k}{N}\right)}{N^2\Lambda_k^2}=\Omega(N^{-4}).
\end{equation}
Therefore each $M_k$ is invertible and the spectral gap is positive. In particular, when $k=1$, $\det M_k=\Theta(N^{-4})$, and the magnitude of the entries are
\begin{equation}
M_1=\begin{pmatrix}
\Theta(N^{-3}) & \Theta(N^{-2})\\
\Theta(N^{-2}) & \Theta(N^{-1}).
\end{pmatrix}
\end{equation}
Therefore one of the eigenvalues must be $\Theta(N^{-1})$. To obtain $\det M_k=\Theta(N^{-4})$, the other eigenvalue must be $\Theta(N^{-3})$.

\section{Proof of rapid ground state preparation of weakly interacting spin systems}\label{proof_rapid_thm}
In this section, we show the rapid ground state preparation for the perturbed Hamiltonian. First, we introduce the assumptions of the filter function $f$, which is similar to that in~\cite[Assumption 12]{DingChenLin2024}.
Our analysis employs the Gevrey function, a subclass of smooth functions characterized by well-controlled decay of the Fourier coefficients. This characteristic plays a crucial role in the quadrature analysis.
\begin{defn}[Gevrey function] \label{def:gevrey}
Let $\Omega\subseteq \RR^d$ be a domain. A complex-valued $C^\infty$ function $h: \Omega\to \CC$ is a \emph{Gevrey function} of order $s\ge 0$, if there exist constants $C_1,C_2>0$ such that for every $d$-tuple of nonnegative integers $\alpha = (\alpha_1,\alpha_2,\ldots,\alpha_d)$,
\begin{equation}
\left\|\partial^\alpha h\right\|_{L^\infty(\Omega)}\leq C_1C^{|\alpha|}_2|\alpha|^{|\alpha|s}\,,
\end{equation}
where $|\alpha|=\sum^d_{i=1} |\alpha_i|$. For fixed constants $C_1,C_2,s$, the set of Gevrey functions is denoted by $\mathcal{G}^{s}_{C_1,C_2}(\Omega)$. Furthermore, $\mathcal{G}^s=\bigcup_{C_1,C_2>0}\mathcal{G}^s_{C_1,C_2}$.
\end{defn}
We refer readers to \cite{AdwanHoepfnerRaich2017,Gus_2019} for background on the Gevrey class. We also note that the use of Gevrey class functions is mainly for simplifying the discretization error analysis and not essential for the design of the Lindbladian.
\begin{assumption}[Filter function in the frequency domain]\label{assum:f_freq}
Assume $\hat{f}$ in the Fourier domain takes the form:
  \begin{equation}\label{eqn:hat_f}
  \hat f(\omega)=\hat{u}(\omega/8)\hat{v}(2\omega)\,.
  \end{equation}
  Here, $\hat{u}$ is a positive function and belongs to Gevrey class $\mathcal{G}^{\alpha}_{A_{1,u},A_{2,u}}(\mathbb{R})$ for some $A_{1,u},A_{2,u}>0$ and $\alpha>1$, meaning that
      \[
      \sup_{\omega\in\mathbb{R}}\left|\frac{\mathrm{d}^n}{\mathrm{d}\omega^n}\left(\hat{u}(\omega)\right)\right|\leq A_{1,u} A^n_{2,u}n^{n\alpha}
      \]
      for any $n\in\mathbb{N}$.
  Also, $\mathrm{supp}(\hat{u})\subset[-1,1]$,  $\hat{u}(\omega)=\Omega(1)$ when $\omega\in [-1/2,1/2]$, and $\hat{u}(-1/4)=1$. In addition, we assume $\hat{v}\in \mathcal{G}^{\alpha}_{A_{1,v},A_{2,v}}(\mathbb{R})$, $\left\|\frac{\mathrm{d}}{\mathrm{d}\omega}\hat{v}\right\|_{L^1}=\mathcal{O}(1)$, $\mathrm{supp}(\hat{v})\subset (-\infty,0]$, and $\hat{v}(-4)=1$.
\end{assumption}

We define the perturbed Hamiltonian as
\begin{equation}\label{eqn:H_eps}
  H_{\varepsilon}=-\sum_{i} Z_i+ \varepsilon \sum_{j} h_j\,.
\end{equation}

Compared to~\cite[Assumption 12]{DingChenLin2024}, the above assumption sets $\Delta=1/2$ and $S_\omega=4$, which is sufficient to ensure~\eqref{eqn:K_0}.
 First, setting $\Delta=1/2$ is adequate for our analysis since the spectral gap of $H_0$ is one, implying that the spectral gap of $H_\varepsilon$ remains greater than $1/2$ when $\varepsilon$ is sufficiently small. Second, choosing $S_\omega=4$ is sufficient because the coupling operator $A_i = X_i$ modifies the energy of $H_0$ by at most one, ensuring that the energy decay of $H_\varepsilon$ can also be bounded by 4 when $\varepsilon$ is sufficiently small.

Now, we are ready to present the rigorous version of Theorem \ref{thm:rapid_mixing_2D_TFIM}:
\begin{thm}[Rigorous version of~Theorem \ref{thm:rapid_mixing_2D_TFIM}]\label{thm:rapid_mixing_2D_TFIM_rigo}
Assume $H$ is a $(r_0,l)$-local Hamiltonian that takes the form of \eqref{eqn:H}, choose the coupling operators $\{A_a\}=\{X_i\}_{i\in \Lambda}$, and the filter function $f$ satisfies Assumption \ref{assum:f_freq}. There exists a constant $\varepsilon^*$ only depends on $k,l,D$ such that when $\varepsilon<\varepsilon^*$, we have
\[
\left\|\exp(\mathcal{L}t)\rho-\ket{\psi_0}\bra{\psi_0}\right\|_1\leq \eta,\quad \forall t=\Omega(\log(N/\eta))\,,\ \rho
\]
where $N=(L+1)^D$ is the system size. Here, $\varepsilon^*=\widetilde{\mathcal{O}}\left(\left(r_0l\right)^{-\Theta(D)}\right)$.
\end{thm}

\rev{The proof of Theorem \ref{thm:rapid_mixing_2D_TFIM_rigo} is based on the analysis of the convergence of observables in the Heisenberg picture, which is inspired by \cite{rouze2024optimal}.} Specifically, the evolution of any observable $O$ in the Heisenberg picture follows the dynamics $O(t)=e^{\mathcal{L}^\dagger t}(O)$. Since $\mc{L}^{\dag}(I)=0$ for any Lindbladian, if the Lindblad dynamics $\rho(t)=\exp(\mathcal{L}t)\rho$ has a unique fixed point, then the identity operator $I$ is also the unique fixed point of the dynamics $e^{\mathcal{L}^\dagger t}$. In other words, $\lim_{t\to \infty} O(t)= \chi_O I$ for some constant $\chi_O$. For a index set $\mc{A}\subseteq \Lambda$, define the \emph{local oscillation operator}
\begin{equation}\label{eqn:delta_X}
  \delta_{\mc{A}}(O):=O-\frac{1}{2^{|\mc{A}|}}I_{\mc{A}}\otimes \Tr_\mc{A}(O)\,,
\end{equation}
where $\Tr_\mc{A}(O)$ is the partial trace of $O$ with respect to the indices in $\mc{A}$. Then $\delta_{\mc{A}}(O)$ measures the local deviation of $O$ from the identity. Furthermore, we expect that $\lim_{t\to \infty} \delta_{\mc{A}}(O(t))=0$ for any nonempty $\mc{A}$ and observable $O$.
For a given index $i\in \Lambda$, for simplicity we identify $i$ with its singleton set $\{i\}$. Then Ref.~\cite{rouze2024optimal} quantifies the convergence of the Lindblad dynamics by means of the convergence of the oscillator norm $\sum_{i\in \Lambda} \left\|\delta_i(O)\right\|$.

For ground state preparation, we first modify the definition of the oscillator norm as follows:
\begin{equation}\label{eqn:oscillator}
\tnorm{O}:=\sum_{i\in \Lambda} \left\|\delta_i\circ P_i(O)\right\|+\left\|\delta_i\circ Q_i(O)\right\|\,.
\end{equation}
Here
\begin{equation}\label{eqn:X_1}
  P_i(O)=\ket{0_i}\bra{0_i}\bra{0_i}O\ket{0_i}+\ket{1_i}\bra{1_i}\bra{1_i}O\ket{1_i}\,,
  \end{equation}
  and
  \begin{equation}\label{eqn:X_2}
  Q_i(O)=\ket{0_i}\bra{1_i}\bra{0_i}O\ket{1_i}+\ket{1_i}\bra{0_i}\bra{1_i}O\ket{0_i}\,,
\end{equation}
which will be used to measure the progress of the dynamics along the diagonal and off-diagonal directions, respectively.

Then using the characterization of the trace distance via observables, the trace distance between $\rho(t)$ and $\sigma=|\psi_0\rangle \langle\psi_0|$ can be bounded as
\begin{equation}\label{eqn:observation_convergence}
\begin{aligned}
  &\left\|\rho(t)-\sigma\right\|_1=\sup_{\|O\|\leq 1} \Tr\left(O\left(\rho(t)-\sigma\right)\right)\\
  \leq  &\sup_{\|O\|\leq 1}\left\|O(t)-\Tr(O(t))/2^N\right\|\left\|\rho(0)-\sigma\right\|_1\\
  \leq &\sup_{\|O\|\leq 1}\tnorm{O(t)}\left\|\rho(0)-\sigma\right\|_1.
\end{aligned}
\end{equation}

Now, to prove Theorem \ref{thm:rapid_mixing_2D_TFIM_rigo}, it suffices to prove the following proposition. We will prove this proposition after giving the proof of Theorem \ref{thm:rapid_mixing_2D_TFIM_rigo}.
\begin{prop}\label{prop:rapid_mixing}
Under the conditions of Theorem \ref{thm:rapid_mixing_2D_TFIM_rigo}, for any observable $O$ such that $\|O\|\leq 1$, we have
  \begin{equation}\label{eqn:O}
  \tnorm{O(t)}\leq \tnorm{O(0)}\exp(-t/4).
  \end{equation}
\end{prop}

\vspace{1em}

\begin{proof}[Proof of Theorem \ref{thm:rapid_mixing_2D_TFIM_rigo}]
Using the relation between the $1$-norm and the trace with observables,
\begin{equation}
\left\|\rho(t)-\sigma\right\|_1=\sup_{\|O\|\leq 1} \Tr\left(O\left(\rho(t)-\sigma\right)\right)\leq  \sup_{\|O\|\leq 1}\left\|O(t)-\Tr(O(t))/2^N\right\|\left\|\rho(0)-\sigma\right\|_1\,.
\end{equation}
We then notice
\[
O(t)-\Tr(O(t))/2^N=\delta_{1}(O(t))+\sum^N_{i=2} \delta_{i}\circ \left(\frac{I_{\{1,\cdots,i-1\}}}{2^{i-1}}\otimes \mathrm{Tr}_{\{1,\cdots,i-1\}}(O(t))\right)\,.
\]
Thus, we have
\[
\begin{aligned}
&\left\|O(t)-\Tr(O(t))/2^N\right\|\\
\leq &\left\|\delta_{1}(O(t))\right\|+\sum^N_{i=2} \left\|\delta_{i}\circ \left(\frac{I_{\{1,\cdots,i-1\}}}{2^{i-1}}\otimes \mathrm{Tr}_{\{1,\cdots,i-1\}}(O(t))\right)\right\|\\
=&\left\|\delta_{1}(O(t))\right\|+\sum^N_{i=2} \left\|\left(\frac{I_{\{1,\cdots,i-1\}}}{2^{i-1}}\otimes \mathrm{Tr}_{\{1,\cdots,i-1\}}\right)\circ \delta_{i}(O(t))\right\|\\
\leq &\sum^N_{i=1}\left\|\delta_{i}(O(t))\right\|\leq \sum^N_{i=1}\left\|\delta_{i}\circ P_i(O(t))\right\|+\left\|\delta_{i}\circ Q_i(O(t))\right\|=\tnorm{O(t)}\,.
\end{aligned}
\]
This provides a proof of \eqref{eqn:observation_convergence}. Next, according to Proposition \ref{prop:rapid_mixing} and $\left\|\rho-\ket{\psi_0}\bra{\psi_0}\right\|_1\leq \revvv{2}$, we have
\[
\left\|\exp(\mathcal{L}t)\rho-\ket{\psi_0}\bra{\psi_0}\right\|_1\leq \revvv{2}\tnorm{O(0)}\exp(-t/4)\leq 8\revvv{N}\exp(-t/4)\,,
\]
where we use $\tnorm{O(0)}\leq 4N$ in the last inequality. This concludes the proof.
\end{proof}

In the following part of the section, we focus on the proof of Proposition \ref{prop:rapid_mixing}. First, the jump operator is
\begin{equation}\label{eqn:K_j_eps}
K_{j,\varepsilon}=\int^\infty_{-\infty}f(t)\exp(iH_\varepsilon t)A_j\exp(-iH_\varepsilon t)\mathrm{d}t=\int^\infty_{-\infty}f(t)\exp(iH_\varepsilon t)X_j\exp(-iH_\varepsilon t)\mathrm{d}t\,,
\end{equation}
and the corresponding dissipative term in the Lindbladian is
\begin{equation}\label{eqn:L_j_eps}
\mathcal{L}_{j,\varepsilon}(\rho)=K_{j,\varepsilon}\rho
K^\dagger_{j,\varepsilon}-\frac{1}{2}\left\{K^\dagger_{j,\varepsilon}K_{j,\varepsilon},\rho\right\}\,.
\end{equation}

When $\varepsilon=0$,
\begin{equation}\label{eqn:K_0}
K_{j,0}=\int^\infty_{-\infty}f(t)\exp(iH_0 t)X_j\exp(-iH_0 t)\mathrm{d}t=\hat{f}(-2)\ket{0_j}\bra{1_j}=\ket{0_j}\bra{1_j}\,.
\end{equation}

We start with the evolution of the observable $P_i(O(t))$:
\[
\begin{aligned}
\partial_t P_i(O(t))=P_i\left(\mathcal{L}^\dagger_{i,\varepsilon}(O(t))\right)+P_i\left(\sum_{j\neq i}\mathcal{L}^\dagger_{j,\varepsilon}(O(t))\right)
\end{aligned}
\]
which implies
\[
\begin{aligned}
&\partial_t \delta_{i}\circ P_i(O(t))=\delta_i\circ P_i \left(\mathcal{L}^\dagger_{i,\varepsilon}(O(t))\right)+\delta_i\circ P_i \left(\sum_{j\neq i}\mathcal{L}^\dagger_{j,\varepsilon}(O(t))\right)\\
=&\delta_i\circ P_i \left(\mathcal{L}^\dagger_{i,0}(O(t))\right)+\delta_i\circ P_i \left(\mathcal{L}^\dagger_{i,\varepsilon}(O(t))-\mathcal{L}^\dagger_{i,0}(O(t))\right)+\sum_{j\neq i}\mathcal{L}^\dagger_{j,\varepsilon}\left(\delta_i\circ P_i \left(O(t)\right)\right)+\left[\delta_i\circ P_i,\sum_{j\neq i}\mathcal{L}^\dagger_{j,\varepsilon}\right]\left(O(t)\right)\\
=&\underbrace{-\delta_i\circ
P_i(O(t))}_{\text{decaying part}}+\underbrace{\sum_{j\neq i}\mathcal{L}^\dagger_{j,\varepsilon}\left(\delta_i\circ P_i \left(O(t)\right)\right)}_{\text{contractive part}}+\delta_i\circ P_i \left(\mathcal{L}^\dagger_{i,\varepsilon}(O(t))-\mathcal{L}^\dagger_{i,0}(O(t))\right)+\left[\delta_i\circ P_i,\sum_{j\neq i}\mathcal{L}^\dagger_{j,\varepsilon}\right]\left(O(t)\right)\,.
\end{aligned}
\]

Here, the second term is contractive in the sense that $\left\|\exp\left(\sum_{j\neq i}\mathcal{L}^\dagger_{j,\varepsilon} t\right)\right\|_{\infty\to\infty}\leq 1$ for any $t>0$, which ensures that the perturbation error does not grow exponentially with time $t$. In the last equality, we use the fact that $\delta_i\circ P_i \left(\mathcal{L}^\dagger_{i,0}(O(t))\right)=-\delta_i\circ P_i(O(t))$ by direct calculation. Following the similar calculations in \cite[Appendix A. Section I (A6-A8)]{rouze2024optimal}, we obtain
\begin{equation}\label{eqn:P_bound}
\begin{aligned}
&\left\|\delta_{i}\circ P_i(O(t))\right\|\\
&\leq \exp(-t)\left\|\delta_{i}\circ P_i(O)\right\|+\int^t_{0}\exp(s-t)\left\|\delta_i\circ P_i \left(\mathcal{L}^\dagger_{i,\varepsilon}(O(s))-\mathcal{L}^\dagger_{i,0}(O(s))\right)+\left[\delta_i\circ P_i,\sum_{j\neq i}\mathcal{L}^\dagger_{j,\varepsilon}\right]\left(O(s)\right)\right\|_{\infty}\mathrm{d}s
\end{aligned}
\end{equation}
Similar to the above calculation, for $Q_i(O(t))$, we also have
\[
\begin{aligned}
&\partial_t \delta_{i}\circ Q_i(O(t))=\delta_i\circ Q_i \left(\mathcal{L}^\dagger_{i,\varepsilon}(O(t))\right)+\delta_i\circ Q_i \left(\sum_{j\neq i}\mathcal{L}^\dagger_{j,\varepsilon}(O(t))\right)\\
=&\delta_i\circ Q_i \left(\mathcal{L}^\dagger_{i,0}(O(t))\right)+\delta_i\circ Q_i \left(\mathcal{L}^\dagger_{i,\varepsilon}(O(t))-\mathcal{L}^\dagger_{i,0}(O(t))\right)+\sum_{j\neq i}\mathcal{L}^\dagger_{j,\varepsilon}\left(\delta_i\circ Q_i \left(O(t)\right)\right)+\left[\delta_i\circ Q_i,\sum_{j\neq i}\mathcal{L}^\dagger_{j,\varepsilon}\right]\left(O(t)\right)\\
=&\underbrace{-\frac{1}{2}\delta_i\left(Q_i(O(t))\right)}_{\text{decaying part}}+\underbrace{\sum_{j\neq i}\mathcal{L}^\dagger_{j,\varepsilon}\left(\delta_i\circ Q_i \left(O(t)\right)\right)}_{\text{contractive part}}+\delta_i\circ Q_i \left(\mathcal{L}^\dagger_{i,\varepsilon}(O(t))-\mathcal{L}^\dagger_{i,0}(O(t))\right)+\left[\delta_i\circ Q_i,\sum_{j\neq i}\mathcal{L}^\dagger_{j,\varepsilon}\right]\left(O(t)\right)\,.
\end{aligned}
\]
This implies
\begin{equation}\label{eqn:Q_bound}
\begin{aligned}
&\left\|\delta_i\left(Q_i(O(t))\right)\right\|\\
&\leq \exp(-\revvv{t/2})\left\|\delta_i\left(Q_i(O)\right)\right\|+\int^t_{0}\exp(\revvv{(s-t)/2})\left\|\delta_i\circ Q_i \left(\mathcal{L}^\dagger_{i,\varepsilon}(O(s))-\mathcal{L}^\dagger_{i,0}(O(s))\right)+\left[\delta_i\circ Q_i,\sum_{j\neq i}\mathcal{L}^\dagger_{j,\varepsilon}\right]\left(O(s)\right)\right\|_{\infty}\mathrm{d}s\,.
\end{aligned}
\end{equation}
Combining \eqref{eqn:P_bound} and \eqref{eqn:Q_bound}, we obtain 
\begin{equation}\label{eqn:PQ_bound}
\begin{aligned}
&\left\|\delta_i\left(P_i(O(t))\right)\right\|+\left\|\delta_i\left(Q_i(O(t))\right)\right\|\\
\leq &\exp(-t/2)\left(\left\|\delta_i\left(P_i(O)\right)\right\|+\left\|\delta_i\left(Q_i(O)\right)\right\|\right)\\
&+\int^t_{0}\exp(s-t)\left\|\delta_i\circ P_i \left(\mathcal{L}^\dagger_{i,\varepsilon}(O(s))-\mathcal{L}^\dagger_{i,0}(O(s))\right)+\left[\delta_i\circ P_i,\sum_{j\neq i}\mathcal{L}^\dagger_{j,\varepsilon}\right]\left(O(s)\right)\right\|_{\infty}\mathrm{d}s\\
&+\int^t_{0}\exp(\revvv{(s-t)/2})\left\|\delta_i\circ Q_i \left(\mathcal{L}^\dagger_{i,\varepsilon}(O(s))-\mathcal{L}^\dagger_{i,0}(O(s))\right)+\left[\delta_i\circ Q_i,\sum_{j\neq i}\mathcal{L}^\dagger_{j,\varepsilon}\right]\left(O(s)\right)\right\|_{\infty}\mathrm{d}s.
\end{aligned}
\end{equation}

To bound the last two terms in \cref{eqn:PQ_bound}. We introduce a lemma to bound the second and third terms in \eqref{eqn:PQ_bound}. First, given lattice $i$ and radius $r>0$, we define $H^{(i,r)}_{\varepsilon}$
as the Hamiltonian that consists of
the Hamiltonian terms of $H_{\varepsilon}$ on a ball of radius $r$ centered at lattice $i$. $K^{r}_{i,\varepsilon}$, $\mathcal{L}^{r}_{i,\varepsilon}$ are defined according to \eqref{eqn:K_j_eps} and \eqref{eqn:L_j_eps} with $H^{(i,r)}_{\varepsilon}$. Then we have the following lemma:
\begin{lem}\label{lem:difference}
Under conditions of Theorem \ref{thm:rapid_mixing_2D_TFIM_rigo} and let $J=r_0^Dl$, for any $r\geq 1$, we have
\begin{equation}
\begin{aligned}
& \left\|\mathcal{L}_{i,\varepsilon}^{r \dagger}-\mathcal{L}_{i,\varepsilon}^{r-1 \dagger}\right\|_{\infty \rightarrow \infty} \leq \xi(r)=\mathcal{O}\left(\frac{1}{2^r}+\exp(-C_{2,f} (r/(4Je))^{1/\alpha}/2)\right), \\
& \left\|\mathcal{L}_{i,\varepsilon}^{\dagger}-\mathcal{L}_{i,0}^{\dagger}\right\|_{\infty \rightarrow \infty} \leq \eta(\varepsilon)=\Or\left(\varepsilon (J\log^\alpha(1/\varepsilon)+1)^{D}l\right)\,.
\end{aligned}
\end{equation}
\end{lem}
With Lemma \ref{lem:difference}, we are ready to provide the proof of Proposition \ref{prop:rapid_mixing}.
\begin{proof}[Proof of Proposition \ref{prop:rapid_mixing}]
The proof follows a similar strategy to that in \cite[Appendix A. Section I]{rouze2024optimal}. We first claim that, there exists $\kappa^c_{i}$ and $\gamma^c_i$ such that
\begin{equation}\label{L_eps_commute_difference}
  \left\|\left[\delta_i\circ P_i,\mathcal{L}^\dagger_{j,\varepsilon}\right]\left(O\right)\right\|,\quad \left\|\left[\delta_i\circ Q_i,\mathcal{L}^\dagger_{j,\varepsilon}\right]\left(O\right)\right\|\leq\sum_k \kappa^k_{i,j}\left(\|\delta_k\circ P_k(O)\|+\|\delta_k\circ Q_k(O)\|\right)
  \end{equation}
  and
\begin{equation}\label{L_eps_difference}
\left\|\delta_i\circ P_i \left(\mathcal{L}^\dagger_{i,\varepsilon}(O)-\mathcal{L}^\dagger_{i,0}(O)\right)\right\|, \left\|\delta_i\circ Q_i \left(\mathcal{L}^\dagger_{i,\varepsilon}(O)-\mathcal{L}^\dagger_{i,0}(O)\right)\right\|\leq \sum_k \gamma^k_{i}\left(\|\delta_k\circ P_k(O)\|+\|\delta_k\circ Q_k(O)\|\right)
\end{equation}
with $\sum_{i,j\neq i}\kappa^k_{i,j}+\sum_i\gamma^k_i$ is smaller than a constant that is independent of the system size.

Denote by $d(i,j)$ the Manhattan distance  between sites $i,j\in \Lambda$ and
$\Gamma(r_0)=\sum_{r\geq r_0}\xi(r)=\mathcal{O}\left(2^{-r_0}\right)$ when $r_0\geq 4kl$. We first show~\eqref{L_eps_commute_difference} and calculate $\kappa$. Following the calculations in \cite[Appendix B]{rouze2024optimal} and letting $B_i(r)=\{j|d(i,j)\leq r\}$, we get
\begin{equation}
\begin{aligned}
&\left\|\left[\delta_i\circ P_i, \mathcal{L}_{j,\varepsilon}^{ \dagger}\right](O)\right\|_{\infty}  =\left\|\left[\delta_i\circ P_i,\left(\mathcal{L}_{j,\varepsilon}^{\dagger}-\mathcal{L}_{j,\varepsilon}^{d(i,j) \dagger}\right)\right](O)\right\|_{\infty} \\
 \leq &\left\|\mathcal{L}_{j,\varepsilon}^{\dagger}-\mathcal{L}_{j,\varepsilon}^{d(i,j) \dagger}\right\|_{\infty \rightarrow \infty}\left\|\delta_i\circ P_i(O))\right\|_{\infty}+2\left\|\left(\mathcal{L}_{j,\varepsilon}^{\dagger}-\mathcal{L}_{j,\varepsilon}^{d(i,j) \dagger}\right)(O)\right\|_{\infty} \\
 \leq &\sum_{r>d(i,j)} \xi(r)\left\|\delta_i\circ P_i(O))\right\|_{\infty}+2 \sum_{r>d(i,j)}\left\|\left((\mathcal{L}^r_{j,\varepsilon})^\dagger-(\mathcal{L}^{r-1}_{j,\varepsilon})^\dagger\right) \delta_{B_j(r)}(O)\right\|_{\infty} \\
 \leq &\Gamma(d(i,j))\left\|\delta_i\circ P_i(O))\right\|_{\infty}+2 \sum_{r>d(i,j)} \zeta(r) \sum_{d(j,k) \leq r}(\left\|\delta_k\circ P_k(O))\right\|_{\infty}+\left\|\delta_k\circ Q_k(O))\right\|_{\infty}) \\
 =&\Gamma(d(i,j))\left\|\delta_i\circ P_i(O))\right\|_{\infty}+2 \sum_k(\left\|\delta_k\circ P_k(O))\right\|_{\infty}+\left\|\delta_k\circ Q_k(O))\right\|_{\infty})\Gamma(\max (d(i,j), d(j,k))) .
\end{aligned}
\end{equation}
Let $r_0>0$, to be determined later. For $d(i,j)\geq r_0$, we have
\begin{equation}\label{eqn:k_1}
\kappa^k_{i,j}=\left\{
  \begin{aligned}
   &\Gamma(d(i,j)),\quad k=i\\
    &2\Gamma(\max (d(i,j), d(j,k))),\quad k\neq i.
  \end{aligned}
\right.
\end{equation}
For $d(i,j)<r_0$, we can bound the commutator as follows:
\begin{equation}
\begin{aligned}
& \left\|\left[\delta_i\circ P_i, \mathcal{L}_{j,\varepsilon}^{ \dagger}\right](O)\right\|_{\infty} \leq\left\|\left[\delta_i\circ P_i,\left(\mathcal{L}_{j,\varepsilon}^{ \dagger}-\mathcal{L}_{j,\varepsilon}^{r_0,\dagger}\right)\right](O)\right\|_{\infty}+\left\|\left[\delta_i\circ P_i,\left(\mathcal{L}_{j,\varepsilon}^{r_0,\dagger}-\mathcal{L}_{j,0}^{r_0,\dagger}\right)\right](O)\right\|_{\infty} \\
\quad \leq&\left(\eta(\eps)+\Gamma\left(r_0\right)\right)\left\|\delta_i\circ P_i(O))\right\|_{\infty}+2 \sum_{r>r_0} \zeta(r) \sum_{d(k, j) \leq r}(\left\|\delta_k\circ P_k(O))\right\|_{\infty}+\left\|\delta_k\circ Q_k(O))\right\|_{\infty})\\
&+2 \eta(\eps) \sum_{d(k, j) \leq r_0}(\left\|\delta_k\circ P_k(O))\right\|_{\infty}+\left\|\delta_k\circ Q_k(O))\right\|_{\infty}) \\
\quad \leq&\left(\eta(\eps)+\Gamma\left(r_0\right)\right)\left\|\delta_i\circ P_i(O))\right\|_{\infty}+2 \sum_k \Gamma\left(\max \left(r_0, d(k, j)\right)\right)(\left\|\delta_k\circ P_k(O))\right\|_{\infty}+\left\|\delta_k\circ Q_k(O))\right\|_{\infty})\\
&+2 \eta(\eps) \sum_{d(k, j) \leq r_0}(\left\|\delta_k\circ P_k(O))\right\|_{\infty}+\left\|\delta_k\circ Q_k(O))\right\|_{\infty})\,.
\end{aligned}
\end{equation}
Therefore, for $d(i,j)<r_0$ we have
\begin{equation}\label{eqn:k_2}
\kappa^k_{i,j}=\left\{
  \begin{aligned}
   &3\eta(\eps)+3\Gamma\left(r_0\right),\quad d(k,j)\leq r_0\\
    &2\Gamma\left(d(k,j)\right),\quad d(k,j)> r_0
\end{aligned}\right.\,.
\end{equation}

Similarly, we can also show~\eqref{L_eps_difference} and calculate $\gamma$. We have
\begin{equation}
\begin{aligned}
\left\|\delta_i\circ P_i \left(\mathcal{L}^\dagger_{i,\varepsilon}-\mathcal{L}^\dagger_{i,0}\right)(O)\right\|_{\infty}& \leq \left\|\delta_i\circ P_i \left(\mathcal{L}^\dagger_{i,\varepsilon}-\mathcal{L}^{r_0,\dagger}_{i,\varepsilon}\right)(O)\right\|_{\infty}+\left\|\delta_i\circ P_i \left(\mathcal{L}^{r_0,\dagger}_{i,\varepsilon}-\mathcal{L}^{\dagger}_{i,0}\right)(O)\right\|_{\infty} \\
& \leq \eta(\varepsilon) \sum_{d(i, k) \leq r_0}\left\|\delta_k\circ P_k(O))\right\|_{\infty}+2 \sum_{r \geq r_0} \zeta(r) \sum_{d(i,k) \leq r} (\left\|\delta_k\circ P_k(O))\right\|_{\infty}+\left\|\delta_k\circ Q_k(O))\right\|_{\infty}) \\
& \leq \eta(\varepsilon) \sum_{d(i,k) \leq r_0}\left\|\delta_k\circ P_k(O))\right\|_{\infty}\\
&+2 \sum_j \Gamma\left(\max \left(r_0, d(i,k)\right)\right)(\left\|\delta_k\circ P_k(O))\right\|_{\infty}+\left\|\delta_k\circ Q_k(O))\right\|_{\infty}).
\end{aligned}\,.
\end{equation}
This implies
\begin{equation}
  \gamma^k_{i}=\left\{
    \begin{aligned}
     &\eta(\eps)+2\Gamma\left(r_0\right),\quad d(i,k)\leq r_0\\
      &2\Gamma\left(d(i,k)\right),\quad d(i,k)> r_0.
  \end{aligned}\right.
  \end{equation}
Similar to the calculations in \cite[Appendix A.2]{rouze2024optimal}, we get
\begin{equation}
\kappa = \sum_{i,j\neq i} \kappa_{i, j}^k+\sum_i \gamma_i^k \leq 4\left(2 r_0+1\right)^{2 D} \eta(\eps)+20\sum_{m'\geq r_0}\sum_{m\geq m'}(2m+1)^{2D+1}\Gamma(m)\,.
\end{equation}
Choosing $r_0=\Theta(\max\{J,D^2\})$ sufficiently large so that the second term is smaller than $1/8$, we then set $\varepsilon$ small enough so that
\[
\varepsilon=\mathcal{O}\left((2r_0+1)^{-2D}(J\log^\alpha(1/\varepsilon)+1)^{-D}l^{-1}\right)\,,
\]
we have $\kappa<1/4$.

Plugging \eqref{L_eps_commute_difference} and \eqref{L_eps_difference} into \eqref{eqn:PQ_bound}, we have
\begin{equation}
\begin{aligned}
&\left\|\delta_i\left(P_i(O(t))\right)\right\|+\left\|\delta_i\left(Q_i(O(t))\right)\right\|\\
&\leq \exp(-t/2)(\left\|\delta_i\left(P_i(O)\right)\right\|+\left\|\delta_i\left(Q_i(O)\right)\right\|)\\
&+\kappa\int^t_{0}\exp(\revvv{(s-t)/2})(\left\|\delta_i\left(P_i(O(s))\right)\right\|+\left\|\delta_i\left(Q_i(O(s))\right)\right\|)\mathrm{d}s
\end{aligned}
\end{equation}
Because $\kappa<1/4$, we can bound $\left\|\delta_i\left(P_i(O(t))\right)\right\|+\left\|\delta_i\left(Q_i(O(t))\right)\right\|$:
\begin{equation}
\begin{aligned}
\tnorm{O(t)}=\sum_i  \left(\left\|\delta_i\left(P_i(O(t))\right)\right\|+\left\|\delta_i\left(Q_i(O(t))\right)\right\| \right)\leq e^{-t/4}\sum_i\left(\left\|\delta_i\left(P_i(O)\right)\right\|+\left\|\delta_i\left(Q_i(O)\right)\right\|\right)
\end{aligned}
\end{equation}
This concludes the proof.
\end{proof}

Finally, we provide the proof of Lemma \ref{lem:difference}.
\begin{proof}[Proof of Lemma \ref{lem:difference}]
  According to \cite[Lemma 16]{DingChenLin2024} with $\Delta=1/2$ and $S_\omega=4$ (see the detailed explanation about this choice under Assumption~\ref{assum:f_freq}), the filter function in the time domain satisfies $|f(s)|=\mathcal{O}( C_{1,f}\exp(-C_{2,f}|s|^{1/\alpha})$, where $C_{1,f},C_{2,f}$ are constants that only depend on $A_{1,u},A_{2,u},A_{1,v},A_{2,v}$ in Assumption~\ref{assum:f_freq}. This directly implies
   \begin{equation}\label{eqn:L_1_f_bound}
  \|(|s|+1)f(s)\|_{L^1}\leq C\,.
  \end{equation}
  where $C$ only depends on $A_{1,u},A_{2,u},A_{1,v},A_{2,v}$, and $\alpha$.

  We first show the bound of $\left\|\mathcal{L}^\dagger_{i,\eps}-\mathcal{L}^\dagger_{i,0}\right\|_{\infty\rightarrow\infty}$. Note
  \begin{equation}\label{eqn:L_diff}
  \left\|\mathcal{L}^\dagger_{i,\eps}-\mathcal{L}^\dagger_{i,0}\right\|_{\infty\rightarrow\infty}\leq 2\left(\left\|K_{j,\varepsilon}\right\|+\left\|K_{j,0}\right\|_{\infty}\right)\left\|K_{j,\varepsilon}-K_{j,0}\right\|_{\infty}\leq 4\|f(t)\|_{L^1}\left\|K_{j,\varepsilon}-K_{j,0}\right\|\,.
  \end{equation}
  Let $J=r_0^D l$. Next,
  \[
  \begin{aligned}
  &\left\|K_{j,\varepsilon}-K_{j,0}\right\|_{\infty}=\left\|\int^\infty_{-\infty} f(t)\left[\exp(iH_\eps t)X_j\exp(-iH_\eps t)-\exp(iH_0 t)X_j\exp(-iH_0 t)\right]\mathrm{d} t\right\|_{\infty}\\
  =&\left\|\int^\infty_{-\infty} f(t)\underbrace{\left[\exp(iH_\eps t)X_j\exp(-iH_\eps t)-\exp(iH^{(j,r)}_\varepsilon t)X_j\exp(-iH^{(j,r)}_\eps t)\right]}_{\left\|\cdot\right\|\leq \|X_j\|_\infty\min\left\{\frac{(2J|t|)^r}{r!},2\right\}}\mathrm{d} t\right.\\
  &+\int^\infty_{-\infty} f(t)\left[\exp(iH^{(j,r)}_\varepsilon t)X_j\exp(-iH^{(j,r)}_\eps t)-\exp(iH^{(j,r)}_0 t)X_j\exp(-iH^{(j,r)}_0 t)\right]\mathrm{d} t\\
  &\left.+\int^\infty_{-\infty} f(t)\underbrace{\left[\exp(iH^{(j,r)}_0 t)X_j\exp(-iH^{(j,r)}_0 t)-\exp(iH_0 t)X_j\exp(-iH_0 t)\right]}_{=0\ \forall r>1}\mathrm{d} t\right\|_{\infty}\\
  \leq &\int^\infty_{-\infty}f(t)\left(\min\left\{\frac{(2J|t|)^r}{r!},2\right\}+\left\|\exp(iH^{(j,r)}_\varepsilon t)-\exp(iH^{(j,r)}_0 t)\right\|\right)\mathrm{d}t\\
  \leq &\int^\infty_{-\infty}f(t)\left(\min\left\{\frac{(2J|t|)^r}{r!},2\right\}+\left\|H^{(j,r)}_\varepsilon-H^{(j,r)}_0 \right\||t|\right)\mathrm{d}t\\
  \leq & \int^\infty_{-\infty}f(t)\left(\min\left\{\frac{(2J|t|)^r}{r!},2\right\}+\varepsilon (r+1)^Dl|t|\right)\mathrm{d}t\\
  = &\mathcal{O}\left(\int_{|t|\leq \frac{r}{4Je}}\frac{(2J)^r|t|^r}{r!}|f(t)| \mathrm{d}t+\int_{|t|>\frac{r}{4Je}} |f(t)|\mathrm{d}t+\varepsilon (r+1)^Dl\right)\\
  = &\mathcal{O}\left(\int_{|t|\leq \frac{r}{4Je}}\frac{(2J)^r|t|^{r-1}}{r!}\mathrm{d}t+\int_{|t|>\frac{r}{4Je}} \exp(-C_{2,f}|t|^{1/\alpha})\mathrm{d}t+\varepsilon (r+1)^Dl\right)\\
  = &\mathcal{O}\left(\frac{1}{2^r}+\int_{|t|>\frac{r}{4Je}} \exp(-C_{2,f}|t|^{1/\alpha})\mathrm{d}t+\varepsilon (r+1)^Dl\right)
  \end{aligned}
  \]
  In the last third equality, we use \eqref{eqn:L_1_f_bound}.

  Next, to bound the second term, we let $C_{\alpha,C_{2,f}}$ be the constant that depends on $\alpha$ and $C_{2,f}$ such that $\alpha\left|u^{\alpha-1}\right|<\exp(C_{2,f}u/2)$ for any $|u|>C_{\alpha,C_{2,f}}$. When $r>4Je C_{\alpha,C_{2,f}}^\alpha$,
  \[
\begin{aligned}
    &\int_{|t|>\frac{r}{4Je}} \exp(-C_{2,f}|t|^{1/\alpha})\mathrm{d}t\\
    =&2\int_{u>(\frac{r}{4Je})^{1/\alpha}} \alpha u^{\alpha-1}\exp(-C_{2,f}|u|)\mathrm{d}u\\
    \leq &2\int_{u>(\frac{r}{4Je})^{1/\alpha}} \exp(-C_{2,f}|u|/2)\mathrm{d}u= \frac{4}{C_{2,f}}\exp\left(-C_{2,f}\left(\frac{r}{4Je}\right)^{1/\alpha}/2\right)
\end{aligned}
  \]
where we let $u=t^{1/\alpha}$ in the first equality. In the inequality, we use $\alpha \left|u\right|^{\alpha-1}<\exp(C_{2,f}u/2)$ for $|u|>C_{\alpha,C_{2,f}}$.

In conclusion, we can set $r=\Theta\left(\max\left\{J\log^\alpha(1/\varepsilon),J C_{\alpha,C_{2,f}}^\alpha\right\}\right)$ to obtain
  \[
  \begin{aligned}
    &\left\|K_{j,\varepsilon}-K_{j,0}\right\|_{\infty}\\
    =&\mathcal{O}\left(\frac{1}{2^r}+\int_{|t|>\frac{r}{4Je}} \exp(-C_{2,f}|s|^{1/\alpha})\mathrm{d}t+\varepsilon (r+1)^Dl\right)\\
    =&\mathcal{O}\left(\frac{1}{2^r}+\exp(-C_{2,f} (r/(4Je))^{1/\alpha}/2)+\varepsilon (r+1)^Dl\right)\,.
  \end{aligned}
  \]
  Plugging this into \eqref{eqn:L_diff}, we obtain
  \begin{equation}\label{eqn:ker_L_bound}
  \eta(\varepsilon)=O\left(\varepsilon (J\log^\alpha(1/\varepsilon)+1)^{D}l\right)\,.
  \end{equation}

  Next we calculate the function $f(r)$
  \[
  \left\|\mathcal{L}^{r,\dagger}_{i,\eps}-\mathcal{L}^{r-1,\dagger}_{i,\eps}\right\|_{\infty\rightarrow\infty}\leq 2\left(\left\|K^r_{j,\varepsilon}\right\|_\infty+\left\|K^{r-1}_{j,\eps}\right\|_{\infty}\right)\left\|K^{r}_{j,\varepsilon}-K^{r-1}_{j,\eps}\right\|_{\infty}\leq 4\|f\|_{L^1}\left\|K^r_{j,\varepsilon}-K_{j,\eps}^{r-1}\right\|_\infty\,.
  \]
  The term $\left\|K^r_{j,\varepsilon}-K_{j,\eps}^{r-1}\right\|$ can be calculated using Lieb-Robinson bound:
  \begin{equation}
  \begin{aligned}
  \left\|K^r_{j,\varepsilon}-K^{r-1}_{j,\varepsilon}\right\|_{\infty}&=\left\|\int^\infty_{-\infty} f(t)\left[\exp(iH_\eps^{(j,r)} t)X_j\exp(-iH_\eps^{(j,r)} t)-\exp(iH_{\eps}^{(j,r-1)} t)X_j\exp(-iH_{\eps}^{(j,r-1)} t)\right]\mathrm{d} t\right\|_{\infty}\\
  &\leq\int^\infty_{-\infty} f(t)\left\|\exp(iH_\eps^{(j,r)} t)X_j\exp(-iH_\eps^{(j,r)} t)-\exp(iH_{\eps}^{(j,r-1)} t)X_j\exp(-iH_{\eps}^{(j,r-1)} t)\right\|_{\infty}\mathrm{d} t\\
  &\leq\int^\infty_{-\infty} f(t) \min\left\{\frac{(2J|t|)^r}{r!},2\right\} \mathrm{d} t\\
  &=\mathcal{O}\left(\frac{1}{2^r}+\int_{|t|>\frac{r}{4Je}} \exp(-C_{2,f}|s|^{1/\alpha})\mathrm{d}t\right)\,.
  \end{aligned}
  \end{equation}

  As a result, we get
  \[
  \xi(r)=\mathcal{O}\left(\frac{1}{2^r}+\exp(-C_{2,f} (r/(4Je))^{1/\alpha}/2)\right)\,.
  \]
  This concludes the proof.
\end{proof}

\section{Proof of rapid ground state preparation of weakly interacting fermionic systems}\label{proof_rapid_thm_fermion}

We first present the rigorous version of Theorem~\ref{thm:rapid_mixing_fermion} for weakly interacting fermionic systems.
\begin{thm}[Rigorous version of Theorem~\ref{thm:rapid_mixing_fermion}]
\label{thm:rapid_mixing_fermion_rigorous}
Consider a gapped fermionic Hamiltonian $H$ in the form of \eqref{eqn:H_fermion} defined on a $D$-dimensional lattice $\Lambda= [0,L]^D$, and $N=(L+1)^D$ is the system size. Let $\{A_a\}=\{c^\dagger_i,c_i\}_{i\in\Lambda}$ be a set of coupling operators and $\{K_a\}$ be the corresponding jump operators defined via \cref{eqn:jump_time}. Assume the following conditions hold:
\begin{itemize}
    \item $M\succ \Delta$ for some $\Delta>0$.
    \item The filter functions $f$ is chosen properly to satisfy
\begin{equation}\label{eqn:L_1_bound}
\begin{aligned}
  &\hat{f}(\nu)=1,\quad \forall \nu\in [-\|M\|,-\Delta],\quad  \text{and}\quad \hat{f}(\nu)=0,\quad \forall \nu>0\,,\\
  &\max\{|f(t)|\}\leq C_1\|M\|\exp(-C_2|t/\Delta|^{1/\alpha})\,,
\end{aligned}
\end{equation}
where $C_1,C_2,\alpha>0$ are  constants independent of $N$.
\end{itemize}
 Consider the Lindblad dynamics in~\cref{eq:lindblad_no_coherent}. Then there exists a constant $\varepsilon^*$ independent of $L$ and $N$ such that when $\varepsilon<\varepsilon^*$, we have
\begin{equation}
\tau_{\operatorname{mix}}(\eta)=\Theta(\log(N/\eta))\,,
\end{equation}
where $\tau_{\operatorname{mix}}(\eta)$ is defined in~\eqref{eqn:mixing_tracedistance}. Here, $\varepsilon^*=\widetilde{\mathcal{O}}\left(\left(\left(\max_{i,j}|M_{i,j}|\right)r_0l/\Delta\right)^{-\Theta(D)}\|M\|^{-1}\right)$.
\end{thm}

The existence of $f$ follows from~\cite{AdwanHoepfnerRaich2017}. According to the proof of~\cite[Corollary 2.8]{AdwanHoepfnerRaich2017}, there exists a function $g$ belongs to $\mathcal{G}^{\alpha}_{A_{1,u},A_{2,u}}(\mathbb{R})$ for some $A_{1,u},A_{2,u}>0$ such that $g(x)=1$ when $x\geq 1$, $0\leq g(x)\leq 1$ when $x\in [1/2,1]$, and $g(x)=0$ when $x\leq 1/2$. According to~\cite[Lemmas 14, 15]{DingChenLin2024}, $\hat{f}(\nu)=g(-\nu/\Delta)g((\|M\|+1+\nu)/\Delta)$ satisfies the condition~\eqref{eqn:L_1_bound}.

\rev{Before proving Theorem~\ref{thm:rapid_mixing_fermion_rigorous}, we first introduce the fermionic partial trace and the fermionic local oscillation operator, which are essential for analyzing the fermionic systems.}

We introduce a shorthand notation $\ket{1_i} = c^\dagger_i \ket{\mathrm{vac}}$ and $\ket{0_i} = c_i c_i^{\dag} \ket{\mathrm{vac}}$ for all $i$. Then with some abuse of notation, we require that $c_j,c_j^{\dag}$ ``anticommutes'' with $\ket{1_i},\bra{1_i}$ for all $i\ne j$, and ``commutes'' with $\ket{0_i},\bra{0_i}$ for all $i\ne j$.
For instance, consider $O=c_1^{\dag} c_2^{\dag} c_2 c_3$, then
\begin{equation}
\braket{1_2|O|1_2}=(-1)^2 c_1^{\dag} \braket{1_2|c_2^{\dag} c_2|1_2}c_{3}=c_1^{\dag} c_3.
\end{equation}
Here the factor $(-1)^2$ is due to the convention that $\bra{1_2}$ is required to anticommute with $c_1^{\dag}$, and $\ket{1_2}$ is required to anticommute with $c_3$. This is analogous to operations in spin systems, where commuting a state past spin operators acting on different sites follows a similar rule. For example, in a two-site system, we can rewrite
$X_1 X_2 \ket{1_1} = X_1 \ket{1_1} X_2$,
treating the operators sequentially while preserving their site-specific action.

More generally, for any observable $O$, the fermionic partial trace on the $i$-th lattice, denoted by $\mathrm{Tr}^f_i$, can be defined as:
\begin{equation}\label{eqn:fermion_partial_trace_simplified}
\mathrm{Tr}^f_i(O) = \bra{0_i} O \ket{0_i} + \bra{1_i} O \ket{1_i}.
\end{equation}
The tensor product with the identity matrix in~\eqref{eqn:delta_X} now takes the form
\begin{equation}\label{eqn:T_i}
\mathsf{T}_i\left(O\right)=\frac{c^\dagger_ic_i+c_ic^\dagger_i}{2}\mathrm{Tr}^f_i(O)\,,
\end{equation}
where we have used that $c^\dagger_ic_i+c_ic^\dagger_i=1$ is even in the fermionic operators and commutes with all $c_j,c_j^{\dag}$ when $i\ne j$. The fermionic local oscillation operator is then defined as
\begin{equation}\label{eqn:fermionic_oscillator}
  \delta^f_i(O)=O-\mathsf{T}_{i}(O)\,.
\end{equation}
\begin{rem}\label{rem:fermion} \rev{To simplify the calculation and definiton of the oscillator norm, we also introduce a more explicit definition of the fermionic partial trace that does not depend on the convention of commuting states past operators. This definition is equivalent to that in ~\cref{eqn:T_i} for operators of even parity.}

For any observable $O$, we express it in fermionic form with increasing order, meaning
\begin{equation}\label{eqn:O_form}
  O=\sum_{\substack{\bvec{a},\bvec{b}}} g_{\bvec{a},\bvec{b}} (c^\dagger_1)^{a_1}(c_1)^{b_1}\cdots (c^\dagger_n)^{a_n}(c_n)^{b_n}\,.
\end{equation}
where $\bvec{a},\bvec{b}\in \{0,1\}^N$ and $c_i,c^\dagger_i$ are creation and annihilation operators on $i$-th site.

We will only consider operators of even parity.
Define the fermionic partial trace as
  \begin{equation}\label{eqn:fermion_partial_trace}
    \begin{aligned}
\mathrm{Tr}^f_j\left(O\right)=&\sum_{\substack{\sum_i a_i+b_i\,\mathrm{mod}\,2=0\\a_j=b_j=1}} g_{\bvec{a},\bvec{b}} (c^\dagger_1)^{a_1}(c_1)^{b_1}\cdots (c^\dagger_jc_j+c_jc^\dagger_j)\cdots(c^\dagger_n)^{a_n}(c_n)^{b_n}\\
  +&2\sum_{\substack{\sum_i a_i+b_i\,\mathrm{mod}\,2=0\\a_j=b_j=0}} g_{\bvec{a},\bvec{b}} (c^\dagger_1)^{a_1}(c_1)^{b_1}\cdots (c^\dagger_jc_j+c_jc^\dagger_j)\cdots(c^\dagger_n)^{a_n}(c_n)^{b_n}.
  \end{aligned}
\end{equation}
Here we note $c^\dagger_jc_j+c_jc^\dagger_j=1=(c^\dagger_j)^0(c_j)^0$, and partial trace does not change the parity of the operator.

Using the notational convention $\ket{1_j}=c^\dagger_j\ket{\mathrm{vac}}$, $c_k\ket{1_j}=-\ket{1_j}c_k$ and $c^\dagger_k\ket{1_j}=-\ket{1_j}c^\dagger_k$, we may check that
\begin{equation}\label{eqn:fermionic_tr}
\mathrm{Tr}^f_j(O)=\sum_{\substack{\bvec{a},\bvec{b}}} g_{\bvec{a},\bvec{b}} (c^\dagger_1)^{a_1}(c_1)^{b_1}\cdots\left(\bra{0_j}(c^\dagger_j)^{a_j}(c_j)^{b_j}\ket{0_j}+(-1)^{\sum_{i\neq j} a_i+b_i}\bra{1_j}(c^\dagger_j)^{a_j}(c_j)^{b_j}\ket{1_j}\right)\cdots (c^\dagger_n)^{a_n}(c_n)^{b_n}.
\end{equation}
When $O$ has even parity, $(-1)^{\sum_{i\neq j} a_i+b_i}=1$ for all nonvanishing terms. This proves the equivalence between the partial trace defined in \cref{eqn:fermion_partial_trace} and~\eqref{eqn:T_i}.
\end{rem}

\rev{Following the fermionic notation convention \rev{and Remark~\ref{rem:fermion}} above, we may generalize the definition of $P_i,Q_i$ in \eqref{eqn:X_1} and~\eqref{eqn:X_2} to fermionic systems.} We first rewrite $O$ as $O = \sum_{a_i, b_i=0}^{1} \left(c^\dagger_i\right)^{a_i} \left(c_i\right)^{b_i} \otimes_F O^i_{a,b}$, where $\otimes_F$ denotes the fermionic tensor product, which follows the increasing order convention as in~\eqref{eqn:O_form}. Specifically, $\{O^i_{a,b}\}$ are defined in the following
\begin{equation}\label{eqn:O_ab}
\begin{aligned}
&c^\dagger_i c_i \otimes_F O^i_{1,1}=c^\dagger_i c_i O c^\dagger_i c_i-c_i^\dagger O c_i,\\
&(c^\dagger_i)^0(c_i)^0 \otimes_F O^i_{0,0}=c_i c^\dagger_i O c_i c^\dagger_i+c^\dagger_i O c_i\\
&c_i \otimes_F O^i_{0,1}=c_i c^\dagger_i O c^\dagger_i c_i,\ c^\dagger_i \otimes_F O^i_{1,0}=c^\dagger_i c_i O c_i c^\dagger_i\,.
\end{aligned}
\end{equation}
Then we have two projection operators can be expressed as follows:
  \begin{equation}\label{eqn:O_1_fermion}
  P_i(O)=\sum_{a_i=b_i} \left(c^\dagger_i\right)^{a_i}\left(c_i\right)^{b_i}\otimes_F O^i_{a,b}\,,
  \end{equation}
  and
  \begin{equation}\label{eqn:O_2_fermion}
  Q_i(O)=\sum_{a_i\neq b_i} \left(c^\dagger_i\right)^{a_i}\left(c_i\right)^{b_i}\otimes_F O^i_{a,b}\,.
  \end{equation}

Because of the fact the even parity of the density operator $\rho$, which implies that $\mathrm{Tr}(O\rho)=0$ if $O$ has odd parity, we only consider even parity observables in our analysis.
We can define the \emph{fermionic oscillator norm} as
\begin{equation}
\tnorm{O}:=\sum^N_{i=1}\left\|\delta^f_i\circ P_{i}(O)\right\|_\infty+\left\|\delta^f_i\circ Q_{i}(O)\right\|_\infty\,.
\end{equation}

We note that several alternative expressions for the fermionic partial trace exist in the literature (see, e.g., \cite{PhysRevB.105.035121,PhysRevA.104.032411}). However, we emphasize that the fermionic partial trace operation is uniquely defined once the ordering of the sites is fixed and the fermionic states satisfy the parity superselection rule (SSR), i.e., a fermionic state should not involve a coherent superposition between states with an even and an odd number of particles.

Next, we summarize the properties of the fermionic partial trace~\eqref{eqn:T_i} (or~\eqref{eqn:fermion_partial_trace}) and the fermionic local oscillation operator~\eqref{eqn:fermionic_oscillator} in the following.

\begin{enumerate}
  \item The fermionic partial trace operators $\{\mathsf{T}_i\}$ commute with each other, meaning $
    \mathsf{T}_i\circ \mathsf{T}_j\left(O\right)=\mathsf{T}_j\circ \mathsf{T}_i\left(O\right)$.

\item The fermionic partial trace is contractive in operator norm, meaning $\|\mathsf{T}_i\|_{\infty\rightarrow\infty}\leq 1$.

  \item The fermionic local oscillation operator can control the convergence of observables.
  \begin{lem}\label{lem:bound_of_O_norm}
    \leftskip=0.9cm
    For any observable $O$ that takes the form of~\eqref{eqn:O_form}, we have
  \[
  \left\|O-I/2^N\mathrm{Tr}(O)\right\|_\infty\leq \sum^N_{i} \|\delta^f_i(O)\|_{\infty}\,.
  \]
  \end{lem}

\item The fermionic partial trace and local fermionic oscillation operator  commute with operators that act on different sites:
  \begin{lem}\label{lem:commute}
    \leftskip=0.9cm
    Given any superoperator $\mathcal{F}$:
    \begin{equation}\label{eqn:F_form}
    \mathcal{F}(O)=p_1p_2\cdots p_lOq_1q_2\cdots q_r\,,
    \end{equation}
    where $p_i,q_i\in \{c_i,c^\dagger_i\}_{i\in I}$. If $l+r$ is an even number, $j\notin I$, and $O$ has even parity, $[\mathcal{F},\mathsf{T}_j](O)=0$ and $\left[\mathcal{F},\delta^f_j\right](O)=0$.
  \end{lem}
  According to the above lemma, it is straightforward to see that $\mathcal{L}^\dagger_{i,\varepsilon}$ commutes with $\delta^f_j$ if the site \(j\) is not within the fermionic support of $\mathcal{L}^\dagger_{i,\varepsilon}$.

  \item The fermionic partial trace generates the local fixed point for local Lindbladian operators:
  \begin{lem}\label{lem:local_fix}
    \leftskip=0.9cm
    Assume $\mathcal{L}^\dagger$ can be written into the summations of ~\eqref{eqn:F_form} such that every term satisfies the conditions of Lemma~\ref{lem:commute}. Given a subset $J\in \{1,\dots,N\}$ and an observable $O$ that has even parity, if $I\subset J$, we have $\mathcal{L}(\mathsf{T}_J(O))=0$, where $\mathsf{T}_J(O)=\Pi_{i\in J}\mathsf{T}_{i}(O)$.
  \end{lem}
\end{enumerate}

\rev{In the following part of this section, we will first prove Theorem~\ref{thm:rapid_mixing_fermion_rigorous} using the above properties. In~\cref{sec:new_oscilla}, we first handle the  noninteracting case using the new fermionic oscillator norm. Then, in~\cref{sec:pf_thm_fermion_perturb}, we extend the proof to the perturbative regime. The proof of above properties of the fermionic partial trace and the fermionic local oscillation operator will be given in \cref{sec:pf_fermion_properties} for completeness.}



\subsection{\texorpdfstring{Noninteracting case}{}}\label{sec:new_oscilla}

In this section, we consider the noninteracting case with $\varepsilon=0$ and $H=H_0$. We first calculate $K_i$:
\begin{itemize}
\item When $A_i=c_i$,
\[
\begin{aligned}
  K_{i}=&\int^\infty_{-\infty}f(t)\exp(iH_0 t)c_i\exp(-iH_0 t)\mathrm{d}t=\sum_q c_q\int^\infty_{-\infty}f(s)(e^{-iMs})_{iq}\mathrm{d}s\\
  =&\sum_q c_q\left(\hat{f}(-M)\right)_{iq}=c_i\,.
\end{aligned}
\]
where we use $\hat{f}(-M)$ as an identity according to the conditions of $f$.

\item When $A_i=c^\dagger_i$,
\[
\begin{aligned}
  K_{i}=&\int^\infty_{-\infty}f(t)\exp(iH_0 t)c^\dagger_i\exp(-iH_0 t)\mathrm{d}t=\sum_q c^\dagger_q\int^\infty_{-\infty}f(s)(e^{-iMs})_{qi}\mathrm{d}s\\
  =&\sum_q c^\dagger_q\left(\hat{f}(M)\right)_{qi}=0\,,
\end{aligned}
\]
where we use $\hat{f}(M)=0$ according to the conditions of $f$.
\end{itemize}
Thus, we have the Lindbladian dynamics:
\[
  \frac{\mathrm{d} \rho}{\mathrm{d}t}=\mathcal{L}[\rho]=\sum_i \underbrace{c_i \rho c_i^{\dagger}-\frac{1}{2}\left\{c_i^{\dagger} c_i, \rho\right\}}_{:=\mathcal{L}^\dagger_i(\rho)}\,.
\]

Now, we first prove Theorem~\ref{thm:rapid_mixing_fermion_rigorous} for the simplest case $\varepsilon=0$.
\begin{proof}[Proof of Theorem~\ref{thm:rapid_mixing_fermion_rigorous} when $\varepsilon=0$]

  Let $O=\sum^1_{a_i,b_i=0} \left(c^\dagger_i\right)^{a_i}\left(c_i\right)^{b_i}\otimes_F O^i_{a,b}$, where $\otimes_F$ is the fermionic tensor product that follows the increasing order as in the form of~\eqref{eqn:O_form}. Here, \[
    c^\dagger_i c_i \otimes_F O^i_{1,1}=c^\dagger_i c_i O c^\dagger_i c_i-c_i^\dagger O c_i,\ c^\dagger_i \otimes_F O^i_{1,0}=c^\dagger_i c_i O c_i c^\dagger_i,\ c_i \otimes_F O^i_{0,1}=c_i c^\dagger_i O c^\dagger_i c_i,\  (c^\dagger_i)^0(c_i)^0 \otimes_F O^i_{0,0}=c_i c^\dagger_i O c_i c^\dagger_i+c^\dagger_i O c_i\,.
    \]

  We notice
  \[
  c^\dagger_i O c_i-\frac{1}{2}\left\{c^\dagger_ic_i,O\right\}=-c^\dagger_ic_i\otimes_F O^i_{1,1}-\frac{1}{2}c_i\otimes_F O^i_{0,1}-\frac{1}{2}c^\dagger_i\otimes_F O^i_{1,0}\,.
  \]
  Here, we did not generate parity sign in the first term because $O$ has even parity. This implies
  \[
  \begin{aligned}
  \delta^f_i\left(c^\dagger_iO c_i-\frac{1}{2}\left\{c^\dagger_ic_i,O\right\}\right)=\frac{1}{2}(c_ic^\dagger_i-c^\dagger_ic_i)\otimes_F O^i_{1,1}-\frac{1}{2}c_i\otimes_F O^i_{0,1}-\frac{1}{2}c^\dagger_i\otimes_F O^i_{1,0}
  \end{aligned}
  \]
  At the same time, we notice
  \[
  \delta^f_i(O)=\frac{1}{2}(c^\dagger_ic_i-c_ic^\dagger_i)\otimes_F O^i_{1,1}+c_i\otimes_F O^i_{0,1} +c^\dagger_i\otimes_F O^i_{1,0}\,.
  \]
  According to the commuting property in Lemma \ref{lem:commute}, let
  \[
  O_{1,i}=\sum_{a_i=b_i} \left(c^\dagger_i\right)^{a_i}\left(c_i\right)^{b_i}\otimes_F O^i_{a,b},\quad O_{2,i}=\sum_{a_i\neq b_i} \left(c^\dagger_i\right)^{a_i}\left(c_i\right)^{b_i}\otimes_F O^i_{a,b}\,,
  \]
  we have
  \begin{equation}\label{eqn:decay_O_1}
  \frac{\mathrm{d}\delta^f_i(O_{1,i})}{\mathrm{d}t }=-\delta^f_i(O_{1,i})+\sum_{j\neq i}\delta^f_i(\mathcal{L}^\dagger_j(O_{1,i}))=-\delta^f_i(O_{1,i})+\sum_{j\neq i}\mathcal{L}^\dagger_j(\delta^f_i(O_{1,i}))
  \end{equation}
  and
  \begin{equation}\label{eqn:decay_O_2}
  \frac{\mathrm{d}\delta^f_i(O_{2,i})}{\mathrm{d}t }=-\frac{1}{2}\delta^f_i(O_{2,i})+\sum_{j\neq i}\delta^f_i(\mathcal{L}^\dagger_j(O_{2,i}))=-\frac{1}{2}\delta^f_i(O_{2,i})+\sum_{j\neq i}\mathcal{L}^\dagger_j(\delta^f_i(O_{2,i})).
  \end{equation}

  Similar to the case of spin systems, from~\eqref{eqn:decay_O_1} and~\eqref{eqn:decay_O_2}, we obtain
  \[
    \left\|\delta^f_i\circ P_{i}(O(t))\right\|_\infty+\left\|\delta^f_i\circ Q_{i}(O(t))\right\|_\infty\leq \exp\left(-t/2\right) \left(\left\|\delta^f_i\circ P_{i}(O(0))\right\|_\infty+\left\|\delta^f_i\circ Q_{i}(O(0))\right\|_\infty\right).
  \]
  Thus, we have
  \[
    |||O|||\leq \exp\left(-t/2\right) |||O(0)|||\leq 4N\exp\left(-t/2\right)\,.
  \]
  for any $\|O(0)\|=1$. The remaining step is the same as the proof of Theorem \ref{thm:rapid_mixing_2D_TFIM_rigo} so we omit it.
\end{proof}

\subsection{Proof of Theorem~\ref{thm:rapid_mixing_fermion_rigorous}}\label{sec:pf_thm_fermion_perturb}

The following part of the proof follows a similar strategy to that used in proving Theorem~\ref{thm:rapid_mixing_2D_TFIM} in Appendix~\ref{proof_rapid_thm}. For completeness, we still write down the initial steps in the following. It suffices to show the decay rate of the following quantity that is similar to that in Proposition~\ref{prop:rapid_mixing}:
\[
\tnorm{O(t)}=\sum^N_{i=1}\left\|\delta^f_i\circ P_{i}(O(t))\right\|_\infty+\left\|\delta^f_i\circ Q_{i}(O(t))\right\|_\infty\,.
\]
Here $P_i,Q_i$ are defined in~\eqref{eqn:O_1_fermion} and~\eqref{eqn:O_2_fermion}. We assume the observable $O(0)$ takes the form of~\eqref{eqn:O_form}, we will bound each term separately as follows.
\begin{itemize}
\item $\delta^f_i\circ P_i\left(O(t)\right)$: We notice
\[
\begin{aligned}
&\partial_t \delta^f_i\circ P_i\left(O(t)\right)=\delta^f_i\circ P_i \left(\mathcal{L}^\dagger_{i,\varepsilon}(O(t))\right)+\delta^f_i\circ P_i \left(\sum_{j\neq i}\mathcal{L}^\dagger_{j,\varepsilon}(O(t))\right)\\
=&\delta^f_i\circ P_i \left(\mathcal{L}^\dagger_{i,0}(O(t))\right)+\delta^f_i\circ P_i \left(\mathcal{L}^\dagger_{i,\varepsilon}(O(t))-\mathcal{L}^\dagger_{i,0}(O(t))\right)\\
&+\sum_{j\neq i}\mathcal{L}^\dagger_{j,\varepsilon}\left(\delta^f_i\circ P_i \left(O(t)\right)\right)+\left[\delta^f_i\circ P_i,\sum_{j\neq i}\mathcal{L}^\dagger_{j,\varepsilon}\right]\left(O(t)\right)\\
=&\underbrace{-\delta^f_i\circ P_i(O(t))}_{\text{decaying part}}+\underbrace{\sum_{j\neq i}\mathcal{L}^\dagger_{j,\varepsilon}\left(\delta^f_i\circ P_i \left(O(t)\right)\right)}_{\text{contractive part}}\\
&+\delta^f_i\circ P_i \left(\mathcal{L}^\dagger_{i,\varepsilon}(O(t))-\mathcal{L}^\dagger_{i,0}(O(t))\right)+\left[\delta^f_i\circ P_i,\sum_{j\neq i}\mathcal{L}^\dagger_{j,\varepsilon}\right]\left(O(t)\right),
\end{aligned}
\]
where we use the calculation in the above section to derive the last equality.

\item $\delta^f_i\circ Q_i\left(O(t)\right)$: We notice
\[
\begin{aligned}
&\partial_t \delta^f_i\circ Q_i\left(O(t)\right)=\delta^f_i\circ Q_i \left(\mathcal{L}^\dagger_{i,\varepsilon}(O(t))\right)+\delta^f_i\circ Q_i \left(\sum_{j\neq i}\mathcal{L}^\dagger_{j,\varepsilon}(O(t))\right)\\
=&\delta^f_i\circ Q_i \left(\mathcal{L}^\dagger_{i,0}(O(t))\right)+\delta^f_i\circ Q_i \left(\mathcal{L}^\dagger_{i,\varepsilon}(O(t))-\mathcal{L}^\dagger_{i,0}(O(t))\right)\\
&+\sum_{j\neq i}\mathcal{L}^\dagger_{j,\varepsilon}\left(\delta^f_i\circ Q_i \left(O(t)\right)\right)+\left[\delta^f_i\circ Q_i,\sum_{j\neq i}\mathcal{L}^\dagger_{j,\varepsilon}\right]\left(O(t)\right)\\
=&\underbrace{-\frac{1}{2}\delta^f_i\circ Q_i(O(t))}_{\text{decaying part}}+\underbrace{\sum_{j\neq i}\mathcal{L}^\dagger_{j,\varepsilon}\left(\delta^f_i\circ Q_i \left(O(t)\right)\right)}_{\text{contractive part}}\\
&+\delta^f_i\circ Q_i \left(\mathcal{L}^\dagger_{i,\varepsilon}(O(t))-\mathcal{L}^\dagger_{i,0}(O(t))\right)+\left[\delta^f_i\circ Q_i,\sum_{j\neq i}\mathcal{L}^\dagger_{j,\varepsilon}\right]\left(O(t)\right).
\end{aligned}
\]
\end{itemize}

Similar to~\eqref{eqn:P_bound} and~\eqref{eqn:Q_bound}, according to the above calculation, we obtain
\begin{equation}\label{eqn:P_bound_fermion}
\begin{aligned}
&\left\|\delta^f_i\circ P_i\left(O(t)\right)\right\|_\infty\\
\leq &\exp(-t)\left\|\delta^f_i\circ P_i\left(O(0)\right)\right\|_\infty+\int^t_{0}\exp(s-t)\left\|\delta^f_i\circ P_i \left(\mathcal{L}^\dagger_{i,\varepsilon}(O(s))-\mathcal{L}^\dagger_{i,0}(O(s))\right)+\left[\delta^f_i\circ P_i,\sum_{j\neq i}\mathcal{L}^\dagger_{j,\varepsilon}\right]\left(O(s)\right)\right\|_{\infty}\mathrm{d}s,\\
\end{aligned}
\end{equation}
\begin{equation}\label{eqn:Q_bound_fermion}
\begin{aligned}
&\left\|\delta^f_i\circ Q_i\left(O(t)\right)\right\|_\infty\\
\leq &\exp(-t/2)\left\|\delta^f_i\circ Q_i\left(O(0)\right)\right\|_\infty+\int^t_{0}\exp((s-t)/2)\left\|\delta^f_i\circ Q_i \left(\mathcal{L}^\dagger_{i,\varepsilon}(O(s))-\mathcal{L}^\dagger_{i,0}(O(s))\right)+\left[\delta^f_i\circ Q_i,\sum_{j\neq i}\mathcal{L}^\dagger_{j,\varepsilon}\right]\left(O(s)\right)\right\|_{\infty}\mathrm{d}s\\
\end{aligned}
\end{equation}
These inequalities imply
\begin{equation}\label{eqn:PQ_bound_fermion}
  \begin{aligned}
  &\left\|\delta_i^f\left(P_i(O(t))\right)\right\|+\left\|\delta_i^f\left(Q_i(O(t))\right)\right\|\\
  \leq &\exp(-t/2)\left(\left\|\delta_i^f\left(P_i(O(0))\right)\right\|+\left\|\delta_i^f\left(Q_i(O(0))\right)\right\|\right)\\
  &+\int^t_{0}\exp(s-t)\left\|\delta_i^f\circ P_i \left(\mathcal{L}^\dagger_{i,\varepsilon}(O(s))-\mathcal{L}^\dagger_{i,0}(O(s))\right)+\left[\delta_i^f\circ P_i,\sum_{j\neq i}\mathcal{L}^\dagger_{j,\varepsilon}\right]\left(O(s)\right)\right\|_{\infty}\mathrm{d}s\\
  &+\int^t_{0}\exp((s-t)/2)\left\|\delta_i^f\circ Q_i \left(\mathcal{L}^\dagger_{i,\varepsilon}(O(s))-\mathcal{L}^\dagger_{i,0}(O(s))\right)+\left[\delta_i^f\circ Q_i,\sum_{j\neq i}\mathcal{L}^\dagger_{j,\varepsilon}\right]\left(O(s)\right)\right\|_{\infty}\mathrm{d}s
  \end{aligned}\,,
  \end{equation}
which is the same as~\eqref{eqn:PQ_bound}. Following the idea of proving Proposition~\ref{prop:rapid_mixing}, the next step is to show the inequalities
\begin{equation}\label{L_eps_commute_difference_fermion}
\left\|\left[\delta^f_i\circ P_i,\sum_{j\neq i}\mathcal{L}^\dagger_{j,\varepsilon}\right]\left(O\right)\right\|_\infty,\quad \left\|\left[\delta^f_i\circ Q_i,\sum_{j\neq i}\mathcal{L}^\dagger_{j,\varepsilon}\right]\left(O\right)\right\|_\infty\leq \sum_c \kappa^c_{i}\left(\|\delta^f_c\circ P_c(O)\|_{\infty}+\|\delta^f_c\circ Q_c(O)\|_{\infty}\right)
\end{equation}
and
\begin{equation}\label{L_eps_difference_fermion}
\begin{aligned}
\left\|\delta^f_i\circ P_i \left(\mathcal{L}^\dagger_{i,\varepsilon}(O)-\mathcal{L}^\dagger_{i,0}(O)\right)\right\|_\infty,\quad \left\|\delta^f_i\circ Q_i \left(\mathcal{L}^\dagger_{i,\varepsilon}(O)-\mathcal{L}^\dagger_{i,0}(O)\right)\right\|_\infty\leq \sum_c \gamma^c_{i}\left(\|\delta^f_c\circ P_c(O)\|_{\infty}+\|\delta^f_c\circ Q_c(O)\|_{\infty}\right)
\end{aligned}
\end{equation}
with $\sum_{i}\kappa^c_{i}+\sum_{i}\gamma^c_i$  smaller than a constant that is independent of the system size. The value of $\kappa^c_{i}$ and $\gamma^c_{i}$ can be directly calculated by the following lemma.
\begin{lem}\label{lem:difference_fermion}
Define $J=\left(\max_{i,j}|M_{i,j}|\right)r^D_0l$. Under conditions of Theorem~\ref{thm:rapid_mixing_fermion_rigorous}, for any $r\geq 1$, we have
\begin{equation}
\begin{aligned}
& \left\|\mathcal{L}_{i,\varepsilon}^{r \dagger}-\mathcal{L}_{i,\varepsilon}^{r-1 \dagger}\right\|_{\infty \rightarrow \infty} \leq \xi(r)=\mathcal{O}\left(\left\|M\right\|\left(\frac{1}{2^r}+\exp(-C_{2,f} (r\Delta/(4Je))^{1/\alpha}/2)\right)\right), \\
& \left\|\mathcal{L}_{i,\varepsilon}^{\dagger}-\mathcal{L}_{i,0}^{\dagger}\right\|_{\infty \rightarrow \infty} \leq \eta(\varepsilon)=\Or\left(\left\|M\right\|\left(\varepsilon (J\log^\alpha(1/\varepsilon)/\Delta+1)^{D}l\right)\right)\,.
\end{aligned}
\end{equation}
\end{lem}
\begin{proof}
According to~\cite[Lemma 3]{tong2024fast} and recall $\|h_j\|\leq 1$, $J$ represents the Lieb-Robinson velocity for a fermionic system. The proof of this lemma follows the same argument as that of Lemma~\ref{lem:difference}, and thus, we omit it.
\end{proof}

Finally, similar to the proof of Proposition~\ref{prop:rapid_mixing}, letting $r^*=\Theta(\max\{J/\Delta,D^2,\log(\|M\|^{-1})\})$, we can show
\[
  \kappa = \sum_{i} \kappa_{i}^k+\sum_i \gamma_i^k \leq 4\left(2 r^*+1\right)^{2 D} \eta(\eps)+20\sum_{m'\geq r^*}\sum_{m\geq m'}(2m+1)^{2D+1}\Gamma(m)\,,
\]
where $\Gamma(r)=\sum_{r\geq r_0}\xi(r)=\mathcal{O}\left(\|M\|2^{-r}\right)$  when $r=\Omega(J/\Delta)$. Because $r^*=\Theta(\max\{J/\Delta,D^2,\log(\|M\|^{-1})\})$, the second term is smaller than $1/8$. Finally, we set
\[
\varepsilon=\mathcal{O}\left((2r^*+1)^{-2D}\left(J\log^{\alpha}(1/\varepsilon)/\Delta+1\right)^{-D}l^{-1}\|M\|^{-1}\right)\,
\]
to ensure $\kappa<1/4$ and conclude the proof.

\subsection{Proof of properties of fermionic partial trace and local oscillation operator}\label{sec:pf_fermion_properties}

\begin{proof}[Proof of $\|\mathsf{T}_i\|_{\infty\rightarrow\infty}\leq 1$]
Here, we only consider even parity observable. Different from the definition of $\{O^i_{a,b}\}$ in \cref{eqn:O_ab}, we rewrite it as  \[
  O=c^\dagger_i c_i \otimes_F O^i_{1,1}+c^\dagger_i \otimes_F O^i_{1,0}+c_i \otimes_F O^i_{0,1}+c_ic^\dagger_i \otimes_F O^i_{-1,-1}.
  \]
where

$$
\begin{aligned}
& c_i^{\dagger} c_i \otimes_F O_{1,1}^i=c_i^{\dagger} c_i O c_i^{\dagger} c_i \\
& c_i c_i^{\dagger} \otimes_F O_{-1,-1}^i=c_i c_i^{\dagger} O c_i c_i^{\dagger} \\
& c_i \otimes_F O_{0,1}^i=c_i c_i^{\dagger} O c_i^{\dagger} c_i, c_i^{\dagger} \otimes_F O_{1,0}^i=c_i^{\dagger} c_i O c_i c_i^{\dagger}
\end{aligned}
$$
  Given any vector $\ket{\psi}=\ket{0_i}\ket{\phi_{0,i}}+\ket{1_i}\ket{\phi_{1,i}}$, we have
\[
\begin{aligned}
\bra{\psi}O\ket{\psi}=&\bra{\phi_{1,i}}\bra{1_i}c^\dagger_i c_i \otimes_F O^i_{1,1}\ket{1_i}\ket{\phi_{1,i}}+\bra{\phi_{0,i}}\bra{0_i}c_i c^\dagger_i \otimes_F O^i_{-1,-1}\ket{0_i}\ket{\phi_{0,i}}\\
&+\bra{\phi_{1,i}}\bra{1_i}c^\dagger_i \otimes_F O^i_{1,0}\ket{0_i}\ket{\phi_{0,i}}+\bra{\phi_{0,i}}\bra{0_i}c_i \otimes_F O^i_{0,1}\ket{1_i}\ket{\phi_{1,i}}
\end{aligned}\,.
\]
This implies
\[
\begin{aligned}
  \left\|O\right\|_{\infty}\geq &\max\left\{\sup_{\left\|\ket{1_i}\ket{\phi_{1,i}}\right\|_2=1}\bra{\phi_{1,i}}\bra{1_i}c^\dagger_i c_i \otimes_F O^i_{1,1}\ket{1_i}\ket{\phi_{1,i}},\sup_{\left\|\ket{0_i}\ket{\phi_{0,i}}\right\|_2=1}\bra{\phi_{0,i}}\bra{0_i}c_i c^\dagger_i \otimes_F O^i_{-1,-1}\ket{0_i}\ket{\phi_{0,i}}\right\}\\
  =&\max\left\{\left\|c^\dagger_i c_i \otimes_F O^i_{1,1}\right\|_\infty,\left\|c_i c^\dagger_i \otimes_F O^i_{-1,-1}\right\|_\infty\right\}
\end{aligned}\,.
\]
We consider the first term
\[
  \bra{\phi_{1,i}}\bra{1_i}c^\dagger_i c_i \otimes_F O^i_{1,1}\ket{1_i}\ket{\phi_{1,i}}=\bra{\phi_{1,i}}\bra{1_i}O^i_{1,1} \ket{1_i}\ket{\phi_{1,i}}\,.
\]

Now, we try to get rid of $\ket{1_i}$ in the above equality. We rewrite $O^i_{1,1}=O^{\rm odd,i}_{1,1}+O^{\rm even,i}_{1,1}$, where $O^{\rm odd,i}_{1,1}$ contains terms satisfying $\sum_{j>i} a_j+b_j\,\rm mod\, 2=1$ and $O^{\rm even,i}_{1,1}$ contains terms satisfying $\sum_{j>i} a_j+b_j\,\rm mod\,2=0$. Then,
\[
  \bra{\phi_{1,i}}\bra{1_i}O^i_{1,1} \ket{1_i}\ket{\phi_{1,i}}=-\bra{\phi_{1,i}}O^{\rm odd,i}_{1,1} \ket{\phi_{1,i}}+\bra{\phi_{1,i}}O^{\rm even,i}_{1,1} \ket{\phi_{1,i}}\,,
\]
where we abuse the notation and let $\ket{\phi_{1,i}}\in \mathbb{C}^{2^{n-1}}$ and $O^{\rm odd,i}_{1,1},O^{\rm even,i}_{1,1}\in \mathbb{C}^{2^{n-1}\times 2^{n-1}}$ act on qubits $1,2,\cdots,i-1,i+1\cdots,n$. Next, we write
\[
\ket{\phi_{1,i}}=\sum_{\bvec{a}\in \{0,1\}^{N-1}}c_{\bvec{a}}\ket{a_1}\ket{a_2}\cdots\ket{a_{i-1}}\ket{a_{i+1}}\cdots\ket{a_n}\,,
\]
where $\bvec{a}=(a_1,a_2,\cdots,a_{i-1},a_{i+1},\cdots,a_n)\in \{0,1\}^{N-1}$. We then define
\[
  \ket{\phi^{\rm odd,i}_{1,i}}=\sum_{\sum_{j>i}a_{j}\,\rm mod\, 2=1}c_{\bvec{a}}\ket{a_1}\ket{a_2}\cdots\ket{a_{i-1}}\ket{a_{i+1}}\cdots\ket{a_n}\,,
\]
\[
  \ket{\phi^{\rm even,i}_{1,i}}=\sum_{\sum_{j>i}a_{j}\,\rm mod\, 2=0}c_{\bvec{a}}\ket{a_1}\ket{a_2}\cdots\ket{a_{i-1}}\ket{a_{i+1}}\cdots\ket{a_n}\,.
\]
We have $\left\langle \phi^{\rm even,i}_{1,i}\right.\ket{\phi^{\rm odd,i}_{1,i}}=0$ and
\[
\begin{aligned}
  &-\bra{\phi_{1,i}}O^{\rm odd,i}_{1,1} \ket{\phi_{1,i}}+\bra{\phi_{1,i}}O^{\rm even,i}_{1,1} \ket{\phi_{1,i}}\\
  =&-\bra{\phi^{\rm odd,i}_{1,i}}O^{\rm odd,i}_{1,1} \ket{\phi^{\rm even,i}_{1,i}}-\bra{\phi^{\rm even,i}_{1,i}}O^{\rm odd,i}_{1,1} \ket{\phi^{\rm odd,i}_{1,i}}\\
  &+\bra{\phi^{\rm even,i}_{1,i}}O^{\rm even,i}_{1,1} \ket{\phi^{\rm even,i}_{1,i}}+\bra{\phi^{\rm odd,i}_{1,i}}O^{\rm even,i}_{1,1} \ket{\phi^{\rm odd,i}_{1,i}}\\
  =  &\bra{\phi^{\rm odd,i}_{1,i}}O^{\rm odd,i}_{1,1} \ket{-\phi^{\rm even,i}_{1,i}}+\bra{-\phi^{\rm even,i}_{1,i}}O^{\rm odd,i}_{1,1} \ket{\phi^{\rm odd,i}_{1,i}}\\
  &+\bra{-\phi^{\rm even,i}_{1,i}}O^{\rm even,i}_{1,1} \ket{-\phi^{\rm even,i}_{1,i}}+\bra{\phi^{\rm odd,i}_{1,i}}O^{\rm even,i}_{1,1} \ket{\phi^{\rm odd,i}_{1,i}}\\
  = &(\bra{\phi^{\rm odd,i}_{1,i}}-\bra{\phi^{\rm even,i}_{1,i}})(O^{\rm odd,i}_{1,1}+O^{\rm even,i}_{1,1})(\ket{\phi^{\rm odd,i}_{1,i}}-\ket{\phi^{\rm even,i}_{1,i}})\,.
\end{aligned}
\]
This implies
\[
  \left\|c^\dagger_i c_i \otimes_F O^i_{1,1}\right\|_\infty=\sup_{\left\|\ket{1_i}\ket{\phi_{1,i}}\right\|_2=1}\left|\bra{\phi_{1,i}}\bra{1_i}O^i_{1,1} \ket{1_i}\ket{\phi_{1,i}}\right|=\left\|O^{\rm odd,i}_{1,1}+O^{\rm even,i}_{1,1}\right\|_\infty=\left\|O^{i}_{1,1}\right\|_\infty\,.
\]
Similarly, we also have
\[
  \left\|c_i c^\dagger_i \otimes_F O^i_{-1,-1}\right\|_\infty=\left\|O^i_{-1,-1}\right\|_\infty,\quad   \left\|O\right\|_{\infty}\geq \max\left\{\left\|O^i_{1,1}\right\|_\infty,\left\|O^i_{-1,-1}\right\|_\infty\right\}
\]

Finally, we notice
\[
  \mathsf{T}_i(O)=\frac{1}{2}\left(c^\dagger_i c_i+c_ic^\dagger_i\right) \otimes_F (O^i_{1,1}+O^{i}_{-1,-1})\,.
\]
Similar to the above calculation, we have
\[
  \left\|\mathsf{T}_i(O)\right\|_{\infty}=\left\|\frac{O^i_{1,1}+O^i_{-1,-1}}{2}\right\|_\infty\leq \max\left\{\left\|O^i_{1,1}\right\|_\infty,\left\|O^i_{-1,-1}\right\|_\infty\right\}\leq\left\|O\right\|_{\infty}\,.
\]
This concludes the proof.
\end{proof}

\begin{proof}[Proof of Lemma~\ref{lem:bound_of_O_norm}]
  We first notice
  \[
  O=\delta^f_1(O)+\mathsf{T}_1(O)=\delta^f_1(O)+\delta^f_2(\mathsf{T}_1(O))+\mathsf{T}_2\circ \mathsf{T}_1(O)\,.
  \]
  Applying the above equality iteratively, we have
  \[
  O-\underbrace{\mathsf{T}_{n}\circ \dots\circ \dots \mathsf{T}_{1}(O)}_{I/2^N\mathrm{Tr}(O)}=\sum^n_{i=1}\delta^f_i\left(\mathsf{T}_{i-1}\circ \dots\circ \dots \mathsf{T}_{1}(O)\right)\,.
  \]
  Next, using the fact that $[\delta^f_i,\mathsf{T}_j]=0$ and $\|\mathsf{T}_i\|_{\infty\rightarrow\infty}\leq 1$, we have
  \[
    \left\|O-\underbrace{\mathsf{T}_{n}\circ \dots\circ \dots \mathsf{T}_{1}(O)}_{I/2^N\mathrm{Tr}(O)}\right\|\leq \sum^n_{i=1}\left\|\delta^f_i\left(O\right)\right\|
  \]
  This concludes the proof.
\end{proof}

\begin{proof}[Proof of Lemma~\ref{lem:commute}] Decompose $O$ as~\eqref{eqn:O_form}, it suffices to prove $[\mathcal{F},\mathsf{T}_j]=0$. There are three cases:
  \begin{itemize}
    \item $a_j\neq b_j$: This case is trivial because $\mathsf{T}_j(O)=\mathsf{T}_j(\mathcal{F}(O))=0$.
    \item $a_j=b_j=0$: This case is also trivial because $\mathsf{T}_j(O)=O$. and $\mathsf{T}_j(\mathcal{F}(O))=\mathcal{F}(O)$.
    \item $a_j=b_j=1$: We note $c_jc^\dagger_j$ generates the same parity as $c^\dagger_jc_j$ when they commute with $c^\dagger_i$ or $c_i$ when $i\neq j$. This implies $\mathsf{T}_j\mathcal{F}(O)=\mathcal{F}\mathsf{T}_j(O)$.
  \end{itemize}

\end{proof}

\begin{proof}[Proof of Lemma~\ref{lem:local_fix}] We only need to consider the case when $a_j=b_j$ for $j\in J$. In this case
  \[
  \mathsf{T}_J(O)=\frac{I_J}{2^{|J|}}\otimes_F O_{\{1,2,\dots,N\}\setminus J}\,,
  \]
  Because $O$ takes the form of~\eqref{eqn:O_form}, if we expand $\mathsf{T}_J(O)$ into the form of~\eqref{eqn:O_form}, we must have $\sum_{i\notin J} a_i+b_i$ is an even number. Then, we have
  \[
  \mathcal{L}^\dagger(\mathsf{T}_J(O))=\mathcal{L}^\dagger\left(\frac{I_J}{2^{|J|}}\right)\otimes_F O_{\{1,2,\dots,N\}\setminus J}=0\,.
  \]
  This concludes the proof.
\end{proof}

\end{document}